\documentclass[11pt]{article}
\usepackage[margin=0.8in]{geometry}

\usepackage{graphicx}
\usepackage{framed}
\usepackage[normalem]{ulem}
\usepackage{amsmath}
\usepackage{xfrac, amsthm,thm-restate,thmtools}
 \usepackage{amssymb}
\usepackage{amsbsy}
\usepackage{amsfonts}
\usepackage{cancel}
\usepackage{enumerate}
\usepackage{float}
\usepackage{etoolbox}
\usepackage{subcaption}
\usepackage{rotating}
\usepackage{url}
\usepackage{enumitem}
\usepackage{mathtools}
\usepackage{comment}
\usepackage[font={small,it}]{caption}
\usepackage{bbm}
\usepackage[usenames,dvipsnames]{color}
\usepackage{xcolor}
\usepackage{algorithmic} %[noend]
\usepackage{algorithm}
\usepackage{array}
\usepackage{nicefrac}
\usepackage{bm}
\usepackage{tikz}
\usepackage{rotating}
\usepackage{booktabs}
\usepackage{tabularx}
\usepackage{array}
\usepackage{amssymb}
\usepackage{pifont}
\usetikzlibrary{arrows.meta, positioning, calc}
\usetikzlibrary{shapes.geometric, arrows.meta, positioning}
\usepackage{placeins}
\usepackage{mdframed}
\usepackage{ragged2e} 
\definecolor{gray230}{RGB}{240,240,240}
\usepackage{pgfplots}
\usepackage{subcaption}
\pgfplotsset{compat=1.18}
\usepackage[colorlinks=true, allcolors=black]{hyperref}
\hypersetup{    
  colorlinks   = true, %Colours links instead of ugly boxes
  urlcolor     = blue, %Colour for external hyperlinks
  linkcolor    = blue, %Colour of internal links
  citecolor   = red %Colour of citations, could be ``red''
}
\providecommand{\keywords}[1]{\textbf{\textit{Index terms---}} #1}

\usepackage{thmtools,thm-restate}
\usepackage{tcolorbox}

\DeclareMathOperator*{\argmax}{arg\,max}

\newtheorem{theorem}{Theorem}
\newtheorem{lemma}{Lemma}

\newtheorem{definition}{Definition}
\newtheorem{observation}{Observation}
\newtheorem*{example}{Example}

\newcommand{\x}{\mathbf x}

\newcommand{\cmark}{\ding{51}} % check mark
\newcommand{\xmark}{\ding{55}} % cross mark
\newcolumntype{Y}{>{\raggedright\arraybackslash}X}
\newcolumntype{C}{>{\centering\arraybackslash}m{0.9cm}}
\title{\bfseries TRUST-SC: Truthful Multi-Task Double Auction for Quality-Aware Spatial Crowdsourcing in Strategic Environment}

\author{Chattu Bhargavi\thanks{\textcolor{blue}{School of Computer Science and Engineering, VIT-AP University, Amaravati, India.} {\tt \textcolor{blue}{bhargavi.chattu506@gmail.com}}}~\href{https://orcid.org/0000-0003-4481-827X}{\includegraphics[scale=0.0045]{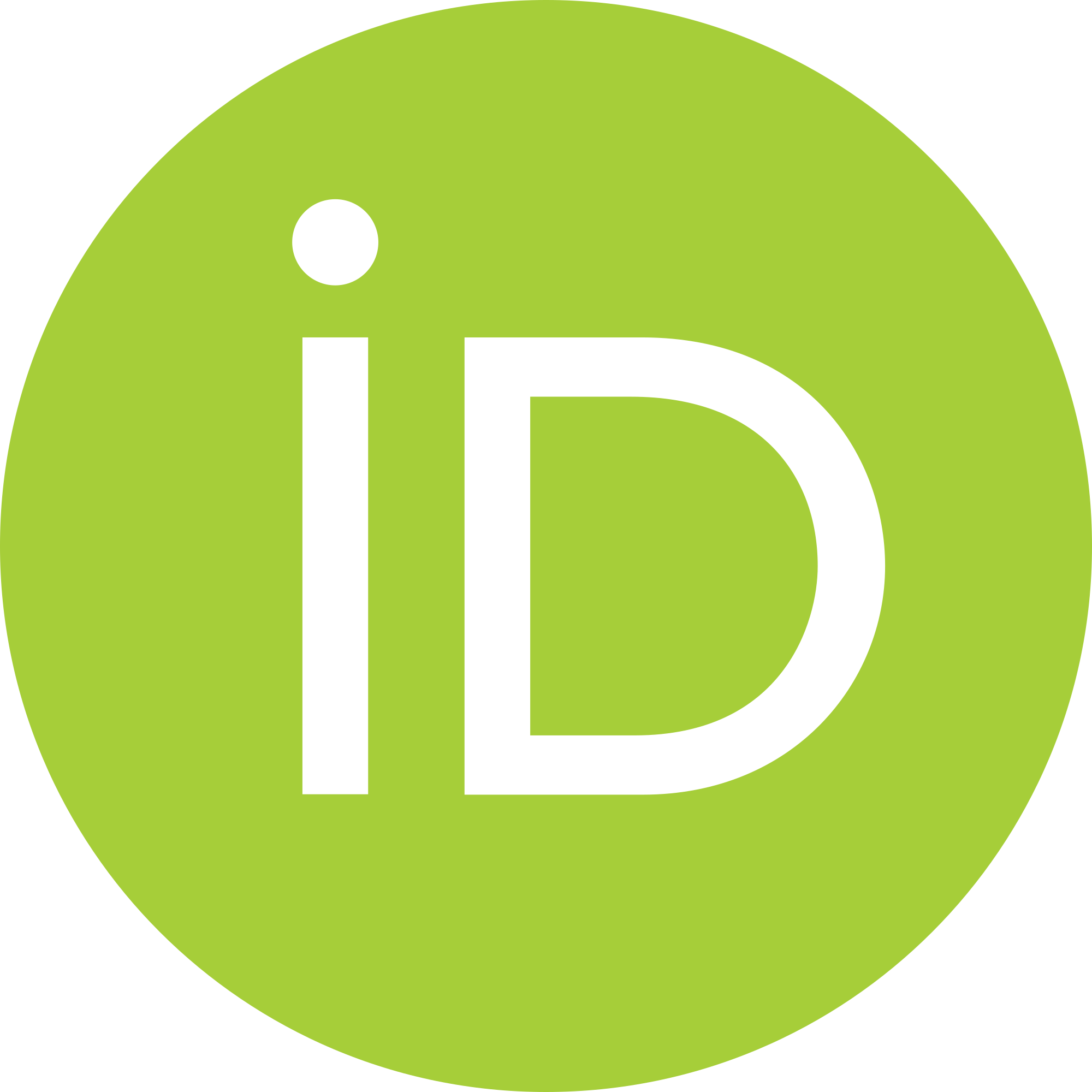}}\and Vikash Kumar Singh\thanks{\textcolor{blue}{School of Computer Science and Engineering, VIT-AP University, Amaravati, India.} {\tt \textcolor{blue}{vikash.singh@vitap.ac.in}}}~\href{https://orcid.org/0000-0002-8747-1627}{\includegraphics[scale=0.0047]{orcid.png}} \and Alok Kumar Shukla\thanks{\textcolor{blue}{Thapar Institute of Engineering \& Technology, Patiala, Punjab, India.} {\tt \textcolor{blue}{alok.shukla@thapar.edu}} }}
\date{}
\begin{document}
\maketitle

\begin{abstract}
Spatial crowdsourcing (SC) enables the assignment of location-based tasks to mobile users who must travel to specific locations to perform sensing or service activities. However, SC systems often operate in strategic environments where both task requesters and task executors possess private information about valuation, leading to challenges in designing efficient and truthful incentive mechanisms. To address these issues, this paper proposes a truthful multi-task double Auction for quality-aware spatial crowdsourcing (TRUST-SC). The proposed framework adopts a three-tier architecture. First, task executors are grouped into spatial clusters to improve scalability and reduce allocation complexity. Second, reliable executors are identified using a majority voting–based quality evaluation process. Third, tasks are allocated, and payments are determined through a multi-unit double auction mechanism that guarantees incentive compatibility and individual rationality. Theoretical analysis and simulation results demonstrate that the proposed mechanism achieves efficient task allocation, reliable executor selection, and improved performance compared with existing benchmark mechanisms.

 \end{abstract}
\keywords{Spatial crowdsourcing, multi-unit double auction, incentive compatibility, strategic agents, computational efficiency.}
\section{Introduction}
\label{sec:intro}
Spatial crowdsourcing (SC) \cite{11202225, 11163677, 10.1145/3291933, 10.1007/s00778-019-00568-7} involves assigning geographical location-based tasks to crowd workers, where successful task completion requires physical presence at the locations to gather information or carry out the tasks \cite{tong2017survey, tong2020survey, gummidi2019sc_survey}. In SC systems, a large number of sensor-equipped mobile users is employed to execute location-based sensing or actuation activities. It has become more popular in real-world applications such as \emph{traffic monitoring} \cite{9458714}, \emph{smart transportation} \cite{feng2024three}, \emph{environmental monitoring} \cite{wang2024sqcs}, \emph{infrastructure auditing} \cite{puerta2024spatial}, and \emph{emergency response} \cite{nhess-2018-206, alamri2024geospatial}. Such systems typically involve two parties: \emph{task requesters}- who publish location-dependent tasks, and task executors (e.g., mobile users or IoT devices)- who complete these tasks in return for rewards (may be monetary). However, real-world spatial crowdsourcing faces several challenges, including strategic behavior of task requesters or task executors or both, limited budgets, task heterogeneity, privacy concerns, and most importantly, the need for well-designed incentive schemes that motivate self-interested agents (task requesters and task executors) to participate in the SC system.\\
    \indent In real-world SC settings, there is a two-sided market including \textit{task requesters} who float tasks associated with particular geographic locations and define their execution constraints and \textit{task executors} who may consist of mobile users capable of successfully performing these tasks. Each party behaves strategically and possesses private information regarding task valuations and execution costs. Designing mechanisms for such a market that are \textit{truthful}, \textit{budget feasible}, \textit{computationally efficient}, and \textit{individually rational} presents a major research challenge. Furthermore, task diversity and heterogeneous executor capabilities require mechanisms that handle such \textit{heterogeneity} while still achieving efficient task assignment and fair rewards. A growing body of literature addresses these issues. For example, Xu \textit{et al.} (2022) \cite{9795698} proposed a three-stage Stackelberg game model to identify socially aware executors with uncertain quality. Zhou \textit{et al.} (2022) \cite{Zhou2022BiObjective} proposed a bi-objective mechanism aimed at maximizing both value and coverage subject to budget constraints. Feng (2025) \cite{feng2025spatial} studies role division in heterogeneous multi-task allocation using the RD-ISM framework. While these models address the key components of SC—\emph{truthfulness}, \emph{social trust}, \emph{quality awareness}, and real-time freshness, they largely tackle these challenges independently. Wang \textit{et al.} (2024) propose the SQCS framework, which models the interaction among workers, requesters, and platforms using evolutionary game theory to ensure sustainable quality control \cite{sqcs2024}. The framework derives conditions for stable system behavior and reports issues such as free-riding and false reporting. Other works focus on heterogeneous task allocation and reliability-aware assignment. For example, Feng \textit{et al.} (2025) study heterogeneous multi-task allocation and highlight the importance of considering diverse worker capabilities and task requirements \cite{feng2025}. Additionally, graph-based allocation approaches have been proposed to handle spatial heterogeneity and large-scale data efficiently \cite{li2023heterogeneous}. Strategic behavior by both requesters and executors in SC systems raises several research questions.\\

\noindent\textbf{\underline{Research Questions}:}  
\begin{itemize}
 \item How can we design a mechanism that ensures truthfulness and quality-awareness, supporting efficient task allocation across diverse types and locations?
 \item What strategies can effectively match heterogeneous tasks to capable executors in spatial crowdsourcing systems?
 \item Can a multi-unit double auction framework ensure truthfulness and individual rationality for both requesters and executors?
 \item What are the analytical and empirical performance bounds of such mechanisms in large-scale spatial crowdsourcing systems?
\end{itemize}
Motivated by the above-discussed scenarios and the research questions, this paper presents a novel three-tiered (or phase) framework for SC markets, which integrates ideas from double auction theory. An overall illustration of the proposed framework is provided in Fig. \ref{Fig:1}. The proposed setup consists of multiple task requesters, endowed with location-based multiple distinct tasks along with a budget (or cost) for each task. On the other side, multiple task executors request a set of tasks for execution purposes and, in return, ask for incentives reported as \emph{bids}. In the first tier of the proposed framework, considering the location information of the tasks and the location of the task executors, the tasks and the task executors are placed in clusters. The reason behind placing the task executors and the tasks into different clusters is that, in practical scenarios, the number of task executors and tasks can be very large, and directly matching all executors with tasks may lead to high computational complexity and inefficient allocation. By forming spatial clusters, the platform reduces the search space and performs task allocation within smaller localized groups. So, the output of the first tier is the set of clusters containing the task executors and the tasks.\\
\indent Once the clusters are formed, the next objective of the proposed framework is to determine the quality task executors in each of the clusters. For that purpose, as a second tier, the small part of the task(s) from a respective cluster will be provided to the task executors for execution. The task executors will be executing the submitted tasks and will return the completed tasks to the platform. On receiving the completed tasks, the platform gives them to the other task executors for grading purposes. The process iterates until all the task executors are graded. At the end of the process, the quality task executors are determined from each of the clusters. The output of the second tier is a set of quality task executors in each cluster.
\begin{figure}[H]
\centering
\includegraphics[scale=0.65]{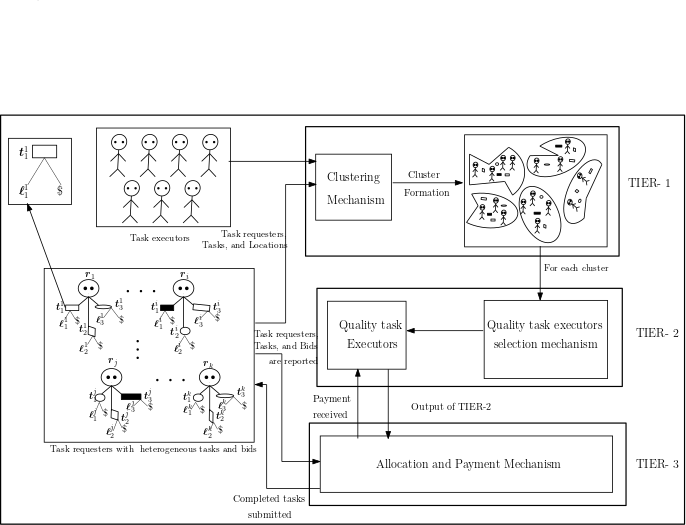}
\caption{A Three-tier multi-unit double auction framework for spatial crowdsourcing.}
\label{Fig:1}
\end{figure}
\indent Once the quality task executors are determined, in the third phase, the objective is to allocate the tasks to the quality task executors in each cluster and decide their payment. Once the tasks are completed, the results are returned to the platform and subsequently delivered to the corresponding task requesters. The task executors receive payments for their services according to the allocation and payment mechanism. Throughout the process, the bid values submitted by task requesters and the ask values reported by task executors are treated as private information and remain undisclosed to other participants.\\
 \indent To address the challenges of efficient task allocation and reliable executor selection in spatial crowdsourcing environments, this paper proposes a novel three-tier mechanism termed the TRUthful multi-task double Auction for quality-aware Spatial Crowdsourcing (TRUST-SC) in a strategic setting. In the first tier, task executors are grouped into spatial clusters based on their geographical locations, motivated by \cite{likas2003global, chong2021k, oti2021comprehensive}, which reduces the search space and improves the efficiency of task allocation. In the second tier, reliable task executors are identified through a quality determination subroutine inspired by \cite{NNis_Pre_2007, roughgarden2016cs269i}, ensuring that only high-quality participants are considered for task execution. In the third tier, a subset of these qualified executors is selected as winners, and their payments are determined according to allocation and pricing rules motivated by \cite{SegalHalevi2018MUDAAT, doi:10.1287/moor.2021.1124}. The main contributions of this work are summarized in the following subsection.
\subsection{Our Contributions}
Our main contributions are as follows:
\begin{itemize}
 \item To ensure reliable task execution, we first introduce a spatial clustering strategy that groups tasks and executors based on geographical proximity, which improves scalability and reduces the computational complexity of task allocation.
 \item We then develop a majority voting–based task executor evaluation mechanism to identify reliable task executors within each cluster, thereby enhancing the quality and reliability of crowdsourced data.
 \item Furthermore, we design a randomized split-market multi-unit double auction mechanism that determines equilibrium prices and allocates tasks efficiently while mitigating strategic manipulation by participating agents.
 \item We theoretically analyze the proposed framework and show that it satisfies important economic properties, including truthfulness (incentive compatibility), individual rationality, and computational efficiency.
 \item Finally, extensive simulations demonstrate that the proposed mechanism significantly improves task allocation efficiency, payment fairness, and system reliability compared with existing baseline approaches.
\end{itemize}

\subsection{Paper Organization}
The remainder of this paper is organized as follows. Section \ref{CABsec:rpw} reviews the related work. Section \ref{CABsec:wum} presents the system model and problem formulation. Section \ref{section:algorithm} describes the proposed TRUST-SC mechanism. Section \ref{sec:TheA} provides Economic Analysis and Probabilistic Guarantees. Section \ref{sec:Sim} presents experimental results, and Section \ref{se:conc} concludes the paper.

\section{Related Prior Works}\label{CABsec:rpw}
In this section, we review prior research that forms the foundation of our study. The discussion is structured into two parts. First, we examine the literature on double auction mechanisms, which provides key insights into efficient and truthful resource allocation among multiple buyers and sellers. Next, we turn to spatial crowdsourcing, where location-dependent tasks and incentive models are discussed. 
% Finally, we explore contributions in the broader field of the Internet of Things (IoT), emphasizing how IoT-enabled devices and platforms support large-scale data collection and task execution.

\subsection{Double auction mechanism}
The double auction mechanism serves as a subroutine in the proposed mechanism; hence, related works in this area are reviewed here. In a typical double auction, multiple sellers seek to sell items, multiple buyers are interested in purchasing them, and an auctioneer (or intermediary) facilitates the process. The mechanism allows both buyers and sellers to simultaneously disclose their private types, while also permitting them to act strategically \cite{RPM_The_1992, NNis_Pre_2007, MYERSON1983265}. A double auction essentially answers the key question: \emph{who buys which item, and at what price?} Its applications span diverse domains such as the spectrum market \cite{5062011, 7420720, Leyton-Brown7202}, Internet advertising \cite{DBLP:journals/corr/FeldmanG16}, and emission trading \cite{10.1007/s10479-018-2826-y}. For two-sided markets, Myerson and Satterthwaite’s impossibility result \cite{MYERSON1983265} establishes that no mechanism can simultaneously satisfy \emph{incentive compatibility} (IC), \emph{budget balance} (BB), and \emph{individual rationality} (IR) while also maximizing gain-from-trade\footnote{Gain-from-trade is defined as the difference between the total valuations of buyers and the total asks of sellers.}.\\
\indent Over the years, numerous double auction mechanisms have been proposed \cite{PLOTT1990245, RPM_The_1992, 10.1007/3-540-45749-6_34}. McAfee’s mechanism \cite{RPM_The_1992} is a landmark contribution, designed for a market where each seller offers a single unit and each buyer demands a single unit. It is \emph{truthful}, \emph{individually rational}, \emph{budget-balanced}, and \emph{prior-free}—that is, it does not rely on assumptions about the agents’ valuation distributions. Other mechanisms, however, often assume specific forms of valuations: \emph{decreasing marginal returns} (DMR) \cite{Blumrosen2014ReallocationM, SegalHalevi2018MUDAAT}, \emph{additive valuations} \cite{doi:10.1287/moor.2021.1124, 45749, DBLP:journals/corr/FeldmanG16}, or \emph{gross substitutes} \cite{ijcai2018-68}, with some models reducing valuations to a single parameter \cite{Gonen2017DYCOMAD, 10.1145/1250910.1250914}. For example, Blumrosen and Dobzinski \cite{Blumrosen2014ReallocationM} proposed a mechanism for multi-parametric agents with DMR valuations, though it is not asymptotically efficient and achieves only a $\tfrac{1}{48}$ competitive ratio. Segal-Halevi et al. \cite{SegalHalevi2018MUDAAT} extended the model to allow multiple-unit trading per agent under DMR valuations, achieving truthfulness and a competitive ratio of at least $1 - O\left(M \sqrt{\tfrac{\ln(k)}{k}}\right)$, where
$k$ is again the total number of units traded in the optimal situation and $M$ is the maximum number of units per trader. Similarly, in \cite{ijcai2018-68}, a truthful mechanism was designed for multiple sellers and buyers where sellers own multiple distinct goods, and buyers have unlimited budgets, under the gross substitute assumption. Further, Segal-Halevi et al. \cite{ijcai2018-68} introduce Multi-Item Double Auction (MIDA), a mechanism for two-sided markets with multiple kinds of goods where traders have gross-substitute valuations. Extending the random market-halving technique, MIDA achieves \emph{truthfulness}, \emph{strong budget balance}, and \emph{prior-freeness}, while guaranteeing asymptotically near-optimal gain-from-trade as market sizes grow.\\
\indent In our paper, we adopt the ideas presented in \cite{ijcai2018-68} to allocate tasks among task executors and determine their payments, once the tasks are distributed in respective clusters and the set of quality task executors has been identified.

\subsection{Spatial crowdsourcing}
 To understand the recent developments in Spatial crowdsourcing (SC), readers may refer to \cite{tong2017survey, 10.1145/3291933, 10.1007/s00778-019-00568-7}. The term \emph{spatial crowdsourcing} (SC) $-$ workers who physically travel to task destinations can be assigned location-dependent tasks. \emph{Mobility}, \emph{travel expenses}, \emph{privacy}, and \emph{fairness} are some of the specific challenges that SC poses. Designing the incentive mechanism is more difficult.\\ %Comprehensive survey of SC provides in the work \cite{tong2017survey, 10.1145/3291933, 10.1007/s00778-019-00568-7} offers an extensive assessment of SC, covering the task assignment, quality control, incentive mechanisms, and privacy.\\
\indent Tong et al. \cite{tong2018dynamic} introduced the Global Dynamic Pricing problem and proposed a matching-based regional pricing strategy (MAPS) that balances worker availability and task demand to maximize revenue. The approach uses task-worker matching to guide spatial price adjustments. Later ride-hailing studies showed that similar location-aware pricing strategies outperform uniform pricing and simple bonuses. Xu et al. \cite{xu2022incentive} address incentives in spatial crowdsourcing with unknown worker quality and social influence among workers. It proposes TACT, which combines learning and game-theoretic pricing to jointly maximize the utilities of all participants. In \cite{wang2024bundle}, bundling tasks with uncertain utilities in spatial crowdsourcing was addressed via a combinatorial incentive mechanism. Similarly, recent studies on auction-based pricing models aim to ensure strategy-proofness and budget feasibility in multi-task allocation scenarios \cite{liu2025group}. 
More recent work has investigated group strategy-proof mechanisms and clock-auction-based pricing, which improve robustness against collusion and strategic manipulation in spatial task allocation \cite{liu2025group}. These mechanisms highlight the importance of combining economic properties such as truthfulness, individual rationality, and budget balance.\\
\begin{sidewaystable}[!htbp]
\centering
\caption{Comparative study of TRUST-SC with existing works with spatial crowdsourcing}
\label{tab:comp1}

\renewcommand{\arraystretch}{1.65}
\setlength{\tabcolsep}{8pt}

\begin{tabularx}{\textheight}{c c Y C C C C C C}
\toprule
\textbf{SI. No.} & \textbf{Paper} & \textbf{Set-up} & \textbf{T} & \textbf{SW} & \textbf{IR} & \textbf{QoTE} & \textbf{RT} & \textbf{CL} \\
\midrule

1  &  \cite{tong2018dynamic} & Multiple task requesters  and multiple spatial task executors with dynamic pricing and matching                                         & \xmark & \xmark & \xmark & \xmark & \cmark & \xmark \\
2  &  \cite{wang2024bundle} & Multiple task requesters (tasks) and multiple workers with task bundles                                        & \cmark & \cmark & \cmark & \xmark & \cmark & \xmark \\
3  &  \cite{liu2017budgetaware} & single task requester with multiple locations and multiple workers                                        & \xmark & \xmark & \xmark & \cmark & \cmark & \xmark \\

4  &  \cite{feng2025hmta} & Multiple task requesters with multiple workers  & \xmark & \xmark & \xmark & \cmark & \cmark & \xmark \\

5 &  \cite{liu2022multi} & Timestamp of task, platform, and set of multiple workers with skills     & \xmark & \xmark & \xmark & \xmark & \cmark & \xmark \\

6 &  \cite{doi:10.1177/09266801251388373} & multiple task providers and multiple task executors & \xmark & \xmark & \xmark & \xmark & \cmark & \xmark \\
7  & \cite{feng2025spatial}       & Multiple heterogeneous tasks and multiple workers                                              & \xmark & \cmark & \xmark & \cmark & \cmark & \xmark\\
8  & \cite{puerta2024spatial}     & Multiple spatial tasks and multiple volunteers                                                & \xmark & \xmark & \xmark & \cmark & \xmark & \xmark\\
9  & \cite{wang2024sqcs}          & Multiple task requesters and multiple workers                                                 & \xmark & \xmark & \xmark & \cmark & \cmark & \xmark\\
10 & \cite{xie2023satisfaction}   & Multiple task requesters and multiple task executors                                          & \xmark & \xmark & \cmark & \cmark & \cmark & \xmark\\
11 & \cite{xu2022incentive}       & Multiple task requesters and multiple workers                                                 & \xmark & \cmark & \xmark & \cmark & \cmark & \xmark\\
12 & \cite{liu2022multi}          & Multiple task requesters and multiple mobile workers                                          & \xmark & \cmark & \xmark & \xmark & \cmark & \xmark\\
13 & \cite{ye2021task}            & Multiple task requesters with spatial tasks and multiple task executors                       & \xmark & \xmark & \xmark & \xmark & \cmark & \cmark\\
14 & \cite{ma2024clustering}      & Multiple task requesters and multiple task executors                                          & \xmark & \xmark & \xmark & \xmark & \cmark & \cmark\\
15 & \textbf{TRUST-SC}              & Multiple task requesters and multiple workers                                                 & \cmark & \cmark & \cmark & \cmark & \cmark & \cmark\\

\bottomrule
\end{tabularx}
\vspace{1mm}
\footnotesize
\textbf{T}: Truthfulness, \textbf{SW}: Social Welfare, \textbf{IR}: Individual Rationality, \textbf{QoTE}: Quality of Task Executors, \textbf{RT}: Running Time, \textbf{CL}: Clustering.
\end{sidewaystable}
\begin{sidewaystable}[!htbp]
\centering
\caption{Comparison of Existing Spatial Crowdsourcing Mechanisms with TRUST-SC}
\label{tab:comparison}
\renewcommand{\arraystretch}{1.65}
\setlength{\tabcolsep}{5pt}
\begin{tabular}{lcccccc}
\toprule
%\hline
\textbf{Work} & \textbf{Incentive} & \textbf{Quality-aware} & \textbf{Truthful} & \textbf{Multi-unit} & \textbf{Learning-based} & \textbf{Application} \\
\midrule

Xu et al. (2023) \cite{xu2023} & Game-theoretic & Partial & Yes & No & No & General SC \\

Wang et al. (SQCS, 2024) \cite{sqcs2024} & Game-theoretic & Yes & No & No & No & Environmental monitoring \\

Feng et al. (2025) \cite{feng2025} & Optimization-based & Partial & No & Yes & No & Heterogeneous tasks \\

Liu et al. (2025) \cite{liu2025group} & Auction-based & No & Yes & Yes & No & Multi-type SC \\

Li et al. (2023) \cite{li2023heterogeneous} & Matching-based & Partial & No & No & No & Large-scale SC \\

RL-based Allocation (2025) \cite{rl2025} & Learning-based & Yes & No & No & Yes & Dynamic SC \\

Fairness-based (2023) \cite{zhou2023} & Game-theoretic & No & Partial & Yes & No & Fair allocation \\

\hline
\textbf{TRUST-SC (Proposed)} & Auction-based & Yes & Yes & Yes & No & Environmental SC \\
\hline
\bottomrule
\end{tabular}
\end{sidewaystable}
\indent In Liu et al. \cite{liu2017budgetaware} to balance task collection across underserved rural areas, budget-aware dynamic incentive mechanisms (DFBA and DABA) were created. These mechanisms modify remuneration to workers in remote areas. In order to maximize user satisfaction measures beyond cost, Xie et al. \cite{xie2023satisfaction} investigate satisfaction-aware task assignment (SATA). Evolutionary game theory is used by Wang et al. \cite{wang2024sqcs} to predict long-term work engagement and fatigue dynamics in SQCS, a sustainable quality control system. Feng \cite{feng2025hmta} investigates service-quality-based incentive schemes, demonstrating increased employee engagement when quality is explicitly rewarded. Freshness vs. cost trade-offs for vehicle-based mapping assignments are optimized by new models of AoI-based incentives for HD map crowdsourcing \cite{ye2024draim}. \cite{huang2024crowdsourcing} proposes a learning-based incentive mechanism for spatial crowdsourcing that handles unknown worker quality and social influence. This approach improves overall utility through smart worker selection and payment design. This \cite{wu2024incentive} addresses task imbalance in multi-type spatial crowdsourcing, where workers prefer low-cost tasks over high-cost ones. This paper counters the issue using a gamified, team-based incentive mechanism with adaptive rewards that balance participation across task types. \cite{wang2021incentive} explores cooperative revenue allocation for spatial crowdsourcing platforms, using cooperative game theory. It proposes a Shapley value-based allocation method and a coalition-based approximation to ensure fairness and computational efficiency. To boost worker enthusiasm in spatial crowdsourcing, this \cite{peng2022research} proposes quality-aware incentive strategies to enhance worker participation based on service quality and idle time, and an optimized task allocation method. It achieves higher incentive effectiveness and fairness in experiments by applying an improved glowworm swarm algorithm. Quality control has become a critical issue in spatial crowdsourcing systems. Wang \textit{et al.} (2024) propose the SQCS framework, which models the interaction among workers, requesters, and platforms using evolutionary game theory to ensure sustainable quality control \cite{sqcs2024}. The framework derives conditions for stable system behavior and mitigates issues such as free-riding and false reporting. Other works focus on heterogeneous task allocation and reliability-aware assignment. For example, Feng \textit{et al.} (2025) study heterogeneous multi-task allocation and highlight the importance of considering diverse worker capabilities and task requirements \cite{feng2025}. Additionally, graph-based allocation approaches have been proposed to handle spatial heterogeneity and large-scale data efficiently \cite{li2023heterogeneous}.
 \\
\indent Xiao et al. \cite{huang2024crowdsourced} highlights how citizen-generated geospatial data is transforming urban science by enabling real-time, large-scale insights into cities. It reviews the applications across mobility, infrastructure, perception, and tourism, while discussing key challenges and future directions. \cite{liu2022multi} explores how to allocate workers to complex spatial crowdsourcing tasks that consist of multiple dependent stages. It introduces a new optimization model and develops greedy and game-theoretic methods to maximize overall profit while respecting skills, deadlines, and dependencies. Debnath et al. \cite{doi:10.1177/09266801251388373} propose a quality-aware task allocation mechanism for spatial crowdsourcing that leverages AI-driven decision strategies to efficiently match location-based tasks with suitable workers. The framework incorporates a procrastination prevention mechanism to ensure timely task execution while maintaining task reliability. Experimental results demonstrate improvements in task allocation efficiency and overall system performance. The \cite{ye2021task} tackles a spatial crowdsourcing framework with multiple service centers, where workers are restricted to specific regions. It introduces an adaptive region-partitioning method with a reinforcement learning strategy to balance workload while maximizing task completion. \cite{ma2024clustering} focuses on how to assign time and location-sensitive tasks to mobile workers in a crowdsourcing system. It proposes an efficient optimization-based mechanism that increases task completion while reducing system cost and worker overload.

\section{Preliminaries}\label{CABsec:wum}
In this section, first of all, the model and notation are discussed. After that, the game-theoretic properties are presented. 
\subsection{Model and Notation}
In this section, we formulate the problem using the framework of a \emph{double auction}. In spatial crowdsourcing set-up, there are $n$ task requesters given as  $\boldsymbol{r} = \{\boldsymbol{r}_1, \boldsymbol{r}_2, \ldots, \boldsymbol{r}_n\}$ and $m$ task executors given as $\boldsymbol{e} = \{\boldsymbol{e}_1, \boldsymbol{e}_2, \ldots, \boldsymbol{e}_m\}$. Here, $n \ll m$. Each of the task requesters $\boldsymbol{r}_i$ is endowed with multiple distinct tasks and is given as $\boldsymbol{t}^i = \{\boldsymbol{t}_1^i, \boldsymbol{t}_2^i, \ldots, \boldsymbol{t}_{n_i}^i\}$ and an infinite budget. It means that any task requester $\boldsymbol{r}_i$ has $|\boldsymbol{n}_i|$ distinct tasks assigned to him/her (henceforth him). 
Each task $\boldsymbol{t}_i^j$ is associated with the location $\boldsymbol{\ell}_i^j$. The location vector of the tasks of $\boldsymbol{r}_k$ task requester is given as $\boldsymbol{\ell}^k = \{\boldsymbol{\ell}_1^k, \boldsymbol{\ell}_2^k, \ldots, \boldsymbol{\ell}_{n_k}^k\}$. The location vector of the tasks of all the task requesters is given as $\boldsymbol{\ell} = \{\boldsymbol{\ell}^1, \boldsymbol{\ell}^2, \ldots, \boldsymbol{\ell}^n\}$. Each of the task requester $\boldsymbol{r}_i$ reports a \emph{tuple} $<\boldsymbol{t}^i, \boldsymbol{\ell}^i>$. For all the task requesters, it is given as $\boldsymbol{\eta} = \{<\boldsymbol{t}^1, \boldsymbol{\ell}^1>, <\boldsymbol{t}^2, \boldsymbol{\ell}^2>, \ldots, <\boldsymbol{t}^i, \boldsymbol{\ell}^i>, \ldots, <\boldsymbol{t}^n, \boldsymbol{\ell}^n>\}$.\\
\indent The set of all the tasks of all the task requesters is given as $\boldsymbol{t} = \{\boldsymbol{t}^1, \boldsymbol{t}^2, \ldots, \boldsymbol{t}^n\}$. In the proposed model, both task requesters and task executors are jointly termed as \emph{agents}, and each may act strategically. There are $\chi$  distinct types of tasks, and each agent is capable of holding at most $\beta$ units of each type of task, for some constants $\chi$ and $\beta$. A bundle of tasks is a vector $Y \in \{0, \ldots, \beta\}^{\chi}$. The empty bundle is represented as $\phi = (0, 0, \ldots, 0)$. The agents are inherently \emph{heterogeneous}.  By \emph{heterogeneity}, we mean that a task requester may submit different types of tasks---for instance, \emph{assessing road conditions}, \emph{traffic congestion monitoring}, \emph{parking availability}, \emph{monitoring air pollution}, etc.---while task executors differ in the types of tasks they are capable of executing. Every task executor declares both its bundle of tasks for execution and the corresponding asking price, which represents the cost charged for executing tasks. For the $j^{th}$ task executor, the maximum number of tasks it can handle is $\Delta_j$, where $\Delta_j$ is much smaller than the total task count $m$. \\
\indent In the discussed setup, the problem is studied as a three-tier process. In the first tier, the task requesters, along with the tasks and the task executors, are first grouped into clusters based on their locations. This clustering step helps reduce the search space and improves the efficiency of task allocation. For clustering, the number of clusters to be formed is reported apriori, say $k$. The set of cluster centers is represented by $\boldsymbol{C}$. It holds randomly selected $k$ points.Let us say any $i^{th}$ cluster out of $k$ clusters is represented as $\xi_i$. So, the set of clusters is represented as $\xi =\{ \xi_1, \xi_2, \dots \xi_k\}$. In the second tier, from the pool of available task executors within each cluster, a subset of quality task executors is identified for task execution. To assess the task executors' quality in each cluster, the platform initially assigns a small task to the task executors. Upon receiving these tasks, the task executors execute them and submit the completed tasks back to the platform. Each completed task is then forwarded to peer task executors within the same cluster for grading purposes (or ranking). Based on the task executors' ranking, the platform determines the set of quality task executors. Once the quality task executors are determined, the mechanism addresses the following key challenges in the third tier, in each cluster:
\begin{itemize}
    \item Which subset of quality task executors should be selected for final task execution?
    \item What payment should be assigned to the selected quality task executors?
\end{itemize}
\indent In the third tier, the above-mentioned two challenges are resolved; for that, the platform distributes the tasks among the qualified task executors for execution. Once the tasks are executed, the executed tasks are returned to the platform and delivered to the respective task requesters, while the task executors receive payment for their services. The bid values of task requesters and the ask values of task executors are treated as \emph{private} information. The bids of task requesters and the asks of task executors are together referred to as \emph{valuations}. Agents are equipped with valuation functions $v$ that return the value of holding $\ell$ units, with zero value assigned to holding none. In this setting, all agents are assumed to have Gross Substitute (GS) valuations. Prices are represented by a vector  $\boldsymbol{p} = (\boldsymbol{p}_1, \boldsymbol{p}_2, \ldots, \boldsymbol{p}_\chi)$, and the price of a task bundle $Y$ is given by $\boldsymbol{p} \cdot Y$. Each agent $i$ has a value function $v_i$ over bundles of tasks, normalized such that $v_i(\phi) = 0$. Since the agents follow a quasi-linear utility model, the utility of a requester $\boldsymbol{r}_i$ is given by $u_i^ {\boldsymbol{r}}$ as:

\begin{equation}
\label{equ:1}
u_i^ {\boldsymbol{r}} (Y, \boldsymbol{p}) = 
    \begin{cases}
        v_i(Y) - \boldsymbol{p} \cdot Y,~ if~\boldsymbol{r}_i~wins\\
        0,~ Otherwise
    \end{cases}
\end{equation}
The utility of task executor $\boldsymbol{e}_j$ is given by the quantity $u_i^{\boldsymbol{e}}$ and is given as:
\begin{equation}
\label{equ:2}
u_j^{\boldsymbol{e}} = 
    \begin{cases}
        \boldsymbol{p} \cdot Y - (v_j(\boldsymbol{e}_j - Y) - v_j(\boldsymbol{e}_j) ),~ if~\boldsymbol{e}_j~wins\\
        0,~ Otherwise
    \end{cases}
\end{equation}
\begin{table}[H]
\centering
\caption{List of Symbols and Their Descriptions}
\label{tab:symbols}
\begin{tabular}{|c|p{10cm}|}
\hline
\textbf{Symbol} & \textbf{Description} \\
\hline
$\boldsymbol{r} = \{\boldsymbol{r}_1, \boldsymbol{r}_2, \ldots, \boldsymbol{r}_n\}$ & Set of $n$ task requesters \\
$\boldsymbol{e} = \{\boldsymbol{e}_1, \boldsymbol{e}_2, \ldots, \boldsymbol{e}_m\}$ & Set of $m$ task executors \\
$\boldsymbol{t}_i = \{\boldsymbol{t}_{i1}, \boldsymbol{t}_{i2}, \ldots, \boldsymbol{t}_{i n_i}\}$ & Tasks associated with task requester $\boldsymbol{r}_i$ \\

$\boldsymbol{\ell}_{ji}$ & Location of task $\boldsymbol{t}_{ji}$ \\
$\boldsymbol{\ell}^k = \{\boldsymbol{\ell}_1^k, \boldsymbol{\ell}_2^k, \ldots, \boldsymbol{\ell}_{n_k}^k\}$ & Location vector for tasks of $\boldsymbol{r}_k$ \\

$\boldsymbol{\ell} = \{\boldsymbol{\ell}^1, \boldsymbol{\ell}^2, \ldots, \boldsymbol{\ell}^n\}$ & Location vectors of all the tasks for all task requesters \\
$\eta = \{ \langle \boldsymbol{t}_1, \boldsymbol{\ell}_1 \rangle, \ldots, \langle \boldsymbol{t}_n, \boldsymbol{\ell}_n \rangle \}$ & Tuple of tasks and locations from all requesters \\
$\chi$ & Number of different types of tasks \\
$\beta$ & Maximum number of units of each task type an agent can hold \\
$Y \in \{0, \ldots, \beta\}^{\chi}$ & Bundle of tasks (multi-type) \\
$\phi$ & Empty bundle $(0, 0, \ldots, 0)$ \\
$\Delta_j$ & Maximum number of similar tasks that executor $\boldsymbol{e}_j$ can execute \\
$v_i$ & Valuation function of agent $i$ \\
$\boldsymbol{p} = (\boldsymbol{p}_1, \boldsymbol{p}_2, \ldots, \boldsymbol{p}_\chi)$ & Price vector for task types \\
$\boldsymbol{p} \cdot Y$ & Total price for bundle $Y$ \\
$u^{\boldsymbol{r}}_i(Y, \boldsymbol{p})$ & Utility of requester $\boldsymbol{r}_i$ for bundle $Y$ at price $\boldsymbol{p}$ \\
$u^{\boldsymbol{e}}_j$ & Utility of executor $\boldsymbol{e}_j$ \\
$\pi(\boldsymbol{\ell}, z)$ & Distance between task location $\boldsymbol{\ell}$ and cluster center $z$ \\
$\xi$ & Set of clusters formed \\
$\boldsymbol{C}$ & Set of cluster centers \\
$Z$ & Set of possible points (locations) for clustering \\
$\hat{z}, \hat{z}_i$ & Randomly selected or recomputed cluster center \\
$d^{\boldsymbol{r}(L)}$ & Demand of task requesters in left spatial crowdsourcing zone. \\
$s^{\boldsymbol{e}(L)}$ & Supply of task executors in left spatial crowdsourcing zone. \\
$\boldsymbol{p}^L$ & Equilibrium price of left spatial crowdsourcing zone.\\
$\boldsymbol{e}^R$ & Set of task executors in the right spatial crowdsourcing zone.\\
$\boldsymbol{r}^R$ & Set of task requesters in the right spatial crowdsourcing zone.\\
$d_i^ {\boldsymbol{r}} (\boldsymbol{p})$ & The demand of task requester $\boldsymbol{r}_i$ at price $\boldsymbol{p}$.\\
$s_i^ {\boldsymbol{r}} (\boldsymbol{p})$ & The supply of task executor $\boldsymbol{e}_i$ at price $\boldsymbol{p}$.\\
$\Tilde{\boldsymbol{r}}^R$ & Store all the selected task requesters in a queue in the right spatial crowdsourcing zone. \\
$\Tilde{\boldsymbol{e}}^R$ &  Store all the selected task executors in a queue in the right spatial crowdsourcing zone.\\
$\boldsymbol{r}^W$ &  The selected requesters to the Winning set. \\
$\boldsymbol{e}^W$ & Selected executors to the Winning set. \\
$d^{\boldsymbol{r}(L)}(\boldsymbol{p})$ &  Demand of requesters in the left spatial crowdsourcing Zone at price $\boldsymbol{p}$.\\
$s^{\boldsymbol{e}(L)}(\boldsymbol{p})$ &  Supply of executors in the left spatial crowdsourcing Zone at price $\boldsymbol{p}$.\\
\hline
\end{tabular}
\end{table}

The demand of task requester $\boldsymbol{r}_i$ at an equilibrium price $p$ is the set of available tasks, say $Y$ that maximizes the utility value $i.e.$  $u^{\boldsymbol{r}}_i(Y, \boldsymbol{p})$ and is given as: 
\begin{equation}
\label{equ:1a}
d_i^ {\boldsymbol{r}} (\boldsymbol{p}) = \argmax_{Y \in [0, \chi.\beta ]} u_i^{\boldsymbol{r}}(Y, \boldsymbol{p})    
\end{equation}
The total demand of the task requesters at an equilibrium price $\boldsymbol{p}$ is given as:
\begin{equation}
\label{equ:1b}
d^ {\boldsymbol{r}} (\boldsymbol{p}) = \displaystyle\sum_{i = 1}^{n}  d_i^ {\boldsymbol{r}} (\boldsymbol{p})
\end{equation}
Similarly, the supply of any task executor $\boldsymbol{e}_j$ at an equilibrium price $p$ is the set of tasks say $Y$ that maximizes the utility value $i.e.$  $u^{\boldsymbol{e}}_j(Y, \boldsymbol{p})$ and is given as:
\begin{equation}
\label{equ:1c}
s_i^ {\boldsymbol{e}} (\boldsymbol{p}) = \argmax_{Y \in [0, \chi.\beta ]} u_i^{\boldsymbol{e}}(Y, \boldsymbol{p})    
\end{equation}
The total demand of the task executors at an equilibrium price $p$ is given as: 
\begin{equation}
\label{equ:1b}
s^ {\boldsymbol{e}} (\boldsymbol{p}) = \displaystyle\sum_{i = 1}^{m}  s_i^ {\boldsymbol{e}} (\boldsymbol{p})
\end{equation}

 \subsection{Game theoretic properties}
The game theoretic properties that will be utilized throughout the paper are presented here. 

\begin{tcolorbox}[title=\textbf{Truthful or Incentive Compatible (IC) \cite{NNisa_Pre_2007}}]
\begin{definition}
\label{def:1}
A mechanism is said to be incentive compatible (IC) when the task executor obtains the maximum possible utility only by revealing its true private valuation, with no benefit from lying.
\end{definition}
\end{tcolorbox}
\begin{tcolorbox}[title=\textbf{Individual Rationality (IR) \cite{NNisa_Pre_2007}}]
\begin{definition}
\label{def:2}
A mechanism is said to be individually rational if every participating task executor in the spatial crowdsourcing market obtains a non-negative utility. $i.e.$  $u^{\boldsymbol{r}}_i(Y, \boldsymbol{p}) \geq 0 $ or $u^{\boldsymbol{e}}_j \geq 0 $.
\end{definition}
\end{tcolorbox}

\begin{tcolorbox}[title=\textbf{Computational Efficiency}]
\begin{definition}
\label{def:3}
A mechanism is said to be computationally efficient when every step of the algorithm runs in polynomial time.
\end{definition}
\end{tcolorbox}   

Let us try to understand GS valuation using an example. 
\begin{tcolorbox}[title=Gross Substitute Valuation: Definition and Illustrative Example, before upper=\setlength{\parindent}{2em}]
\begin{definition}[\textbf{Gross Substitute (GS) \cite{T.roughgarden_20145, NNisa_Pre_2007}}]
 \label{def:gs1}
The task executor's valuation satisfies the Gross Substitutes condition if increasing the prices of some tasks does not reduce its demand for tasks whose prices remain unchanged.
  \end{definition}
  \begin{example}
Running the example from the task requester side. Consider a SC market with two items $a$ and $b$. A task requester $\boldsymbol{r}_i$ has valuation $v_i(a) = 10\$$ for item $a$, $v_i(b)$ = $10\$$ for item $b$, and $v_i({a,b}) = 20\$$ for both the items $a$ and $b$. Let us assume that the equilibrium price of items $a$ and $b$ is $\boldsymbol{p}_a = 5\$$ and $\boldsymbol{p}_b = 5\$$ respectively. The task requester $\boldsymbol{r}_i$ chooses the bundle that maximizes the utility function given in equation \ref{equ:1}. Following equation \ref{equ:1}, if task requester $\boldsymbol{r}_i$ ask for bundle $\{a\}$ then the utility of $\boldsymbol{r}_i$ will be $u_i^{\boldsymbol{r}}$ = $10\$ - 5\$$ = $5\$$. But, if task requester $\boldsymbol{r}_i$ ask for bundle $\{b\}$ then the utility of $\boldsymbol{r}_i$ will be $u_i^{\boldsymbol{r}} (\{b\},~5\$) = 10\$ - 5\$ = 5\$$. However, if the task requester $\boldsymbol{r}_i$ ask for bundle $\{a,~b\}$ then the utility of $\boldsymbol{r}_i$ will be $u_i^{\boldsymbol{r}} (\{a,~b\},~5\$) = 20\$ - (5\$ + 5\$) = 10\$$. Considering the above scenario, task requester $\boldsymbol{r}_i$ will ask for the bundle $\{a,~b\}$ because it gives maximum utility of $10\$$.\\
\begin{figure}[H]
\centering
\includegraphics[scale=0.47]{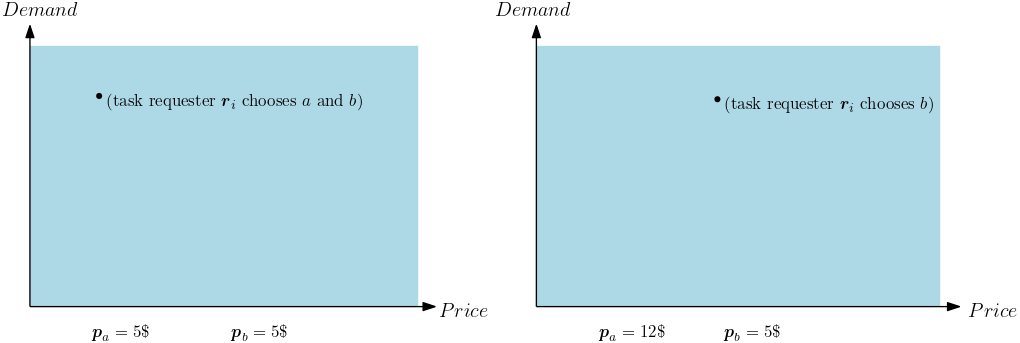}
\caption{Figure on the left shows that when the price of items was $5\$$ each, then task requester $\boldsymbol{r}_i$ asks for both the items, i.e., $a$ and $b$. The figure on the right shows that when the price of the item $a$ got increased to $12\$$, then task requester $\boldsymbol{r}_i$ demands for item $b$ (whose price has not increased and remains the same) only.}
\label{Fig:2}
\end{figure}
\indent Now, let us say the price of item $a$ is increased from $5\$$ to $12\$$ and the price of item $b$ remains same i.e. $\boldsymbol{p}_a = 12\$$ and $\boldsymbol{p}_b = 5\$$. In this case, if task requester $\boldsymbol{r}_i$ ask for bundle $\{a\}$ then the utility of $\boldsymbol{r}_i$ will be $u_i^{\boldsymbol{r}}$ = $10\$ - 12\$$ = $-2\$$. But, if task requester $\boldsymbol{r}_i$ ask for bundle $\{b\}$ then the utility of $\boldsymbol{r}_i$ will be $u_i^{\boldsymbol{r}} (\{b\},~5\$) = 10\$ - 5\$ = 5\$$. However, if the task requester $\boldsymbol{r}_i$ ask for bundle $\{a,~b\}$ then the utility of $\boldsymbol{r}_i$ will be $u_i^{\boldsymbol{r}} (\{a,~b\},~5\$) = 20\$ - (12\$ + 5\$) = 3\$$. Considering the above scenario, task requester $\boldsymbol{r}_i$ will ask for the bundle $\{b\}$ because it gives maximum utility of $5\$$.\\
\indent From the above two scenarios, it can be observed that with the increase in the price of item $a$, the task requester $\boldsymbol{r}_i$ stopped demanding item $a$ but still demands item $b$ (whose price remains the same and has not gone up). Similarly, the discussion could be carried out from the task executors' side. 
\end{example}
\end{tcolorbox}

\section{\textsc{TRUthful multi-task double Auction for quality-aware for Spatial Crowdsourcing Mechanism (TRUST-SC)}}
\label{section:algorithm}
 TRUST-SC consists of the following components:
\begin{itemize}
    \item \textcolor{blue}{\textbf{Cluster Formation (CF)}} $-$ It groups tasks and executors based on geographical proximity, which improves scalability and reduces the computational complexity of task allocation.  
     \item \textcolor{blue}{\textbf{Quality Task Executors Selection Mechanism (QTESM)}} $-$ It determines the quality task executors among the available ones.  
    
    \item \textcolor{blue}{\textbf{Allocation and Pricing Mechanism (APM)}} $-$ Hires the quality task executors and decide their payments.
\end{itemize}
\subsection{\textsc{Cluster Formation}}
\begin {algorithm}[H]
\caption{\textsc{Cluster Formation}} 
\label{algo:1}
\noindent
\begin{algorithmic}[1]
\REQUIRE Set of requesters $\boldsymbol{r}$, number of clusters $k$, and task locations $\boldsymbol{\ell}$
\ENSURE Set of clusters $\xi \leftarrow \phi$
\STATE $\boldsymbol{\mathcal{C}} \gets \emptyset$
\STATE \textbf{while} $|\boldsymbol{\mathcal{C}}| \neq k$ \textbf{do}
    \STATE \hspace{0.5cm} $\hat{z} \gets$ random$(\boldsymbol{\ell})$ \COMMENT{\textcolor{blue}{Pick a random point $z_\ell \in \boldsymbol{\ell}$}}
    \STATE \hspace{0.5cm} $\boldsymbol{\mathcal{C}} \gets \boldsymbol{\mathcal{C}} \cup \{\hat{z}\}$ \COMMENT{\textcolor{blue}{$\boldsymbol{\mathcal{C}}$~ holds each of the randomly selected points.}}
\STATE \textbf{end while}
\REPEAT
    \STATE $\xi \gets \emptyset$, $\xi_i \gets \emptyset$ \COMMENT{\textcolor{blue}{$\xi$ will hold the set of clusters that will be formed. $\xi_i$ hold the minimum distance of any task $\boldsymbol{t}_j^i$ from randomly selected point $z_j$. }}
    \FOR{each $\boldsymbol{r}_i \in \boldsymbol{r}$}
    \FOR{each task $\boldsymbol{t}_k^i \in \boldsymbol{t}^i$}
        \FOR{each $z_j \in \ell$}
            \STATE $\boldsymbol{\Delta}_i \gets \boldsymbol{\Delta}_i \cup \{ \pi(\boldsymbol{\ell}_k^i, z_j) \}$ \COMMENT{\textcolor{blue}{Distance between $\boldsymbol{\ell}_k^i$ and $z_j$ and stored in $\boldsymbol{\Delta}_i$.}}
        \ENDFOR
        \ENDFOR
        \STATE $\hat{j} \gets \arg\min_j \boldsymbol{\Delta}_i$ \COMMENT{\textcolor{blue}{Task with minimum distance from the randomly selected point is determined and held in $\hat{j}$.}}
        \STATE $\xi_{\hat{j}} \gets \xi_{\hat{j}} \cup \{\boldsymbol{t}_{\hat{j}}^i\}$ \COMMENT{\textcolor{blue}{The selected task in Line 14 is stored in $\xi_{\hat{j}}$.}}
    \ENDFOR
    \STATE $\boldsymbol{\mathcal{C}} \gets \emptyset$ \COMMENT{\textcolor{blue}{$\boldsymbol{\mathcal{C}}$ initialized to $\phi$.}}
    \FOR{$i = 1$ \textbf{to} $k$}
        \STATE $\xi \gets \xi \cup \xi_i$ \COMMENT{\textcolor{blue}{$\xi$ holds the set of all the formed clusters.}}
    \ENDFOR
    \FOR{each $\xi_i \in \xi$}
        \STATE $\hat{z}_i \gets \frac{1}{|\xi_i|} \sum_{z_\ell \in \xi_i} z_\ell$ \COMMENT{\textcolor{blue}{$z_i$ is the centroid of cluster $\xi_i$}}
        \STATE $\boldsymbol{\mathcal{C}} \gets \boldsymbol{\mathcal{C}} \cup \{ \hat{z}_i \}$
    \ENDFOR
\UNTIL{no change in clusters takes place}
\STATE \textbf{return} $\xi$ \COMMENT{\textcolor{blue}{Returns the set of formed clusters.}}
\end{algorithmic}
\end{algorithm}
The input to the clustering algorithm is the set of task requesters $\boldsymbol{r}$, the number of clusters $k$, and task locations $\boldsymbol{\ell}$. The output of the clustering algorithm is the set of clusters. Line 1 Initialize the current set of cluster centers to empty. Lines 2-5 begin seeding$-$keep choosing initial centers until exactly $k$ centers are selected. Line 3 selects the point at random from the set $\boldsymbol{\ell}$ and holds it in $\hat{z}$. Add the sampled point to the current set of points in line 4. Lines 6-25 will iterate until there is a change in the clusters that are identified. In line 7, $\xi$ and $\xi_i$ are set to $\phi$. Lines 8-16 iterate over each requester $\boldsymbol{r}_i$ and over each task $t_k^i$ associated with requester $\boldsymbol{r}_i$. Line 11 compute and collect the distance $\pi(\boldsymbol{\ell}_k^i, z_j)$ between the task location $\boldsymbol{\ell}_k^i$ and center $z_j$ into a temporary set (or array) $\boldsymbol{\Delta}_i$. Line 14 chooses the index of the nearest center for the (current) task according to the distances stored in $\boldsymbol{\Delta}_i$. Assign the task to the cluster $\hat{j}$ (append the task into cluster bucket $\xi_{\hat{j}}$). Line 17 clears the center set before recomputing updated centers (centroids). Lines 18-20 iterate over clusters to aggregate them into the global set $\xi$. In line 22, for every cluster $\xi_i \in \xi$, recompute its new center (centroid) and hold it in $\boldsymbol{\mathcal{C}}$ in line 23. Once there are no changes seen in the clusters, in line 26, the set of clusters formed is returned.

\subsubsection{\textsc{Illustration of Cluster Formation with an Example}}
Let us consider an example to understand Algorithm \ref{algo:1} in a detailed manner. Suppose there are five task requesters $\boldsymbol{r}=\{\boldsymbol{r}_1, \boldsymbol{r}_2, \boldsymbol{r}_3, \boldsymbol{r}_4, \boldsymbol{r}_5\}$, where each task requester has specific tasks located at certain coordinates. Say, requester $\boldsymbol{r}_1$ has two tasks $\boldsymbol{t}_{1}^1, \boldsymbol{t}_{2}^1$ and is located at $(4, 1)$ and $(3, 3)$ respectively. Requester, $\boldsymbol{r}_2$ has two tasks, $\boldsymbol{t}_{1}^2$ and $\boldsymbol{t}_{2}^2$, located at $(3, 2)$ and $(2, 2)$ respectively. Requester $\boldsymbol{r}_3$ has three tasks, $\boldsymbol{t}_{1}^3$, $\boldsymbol{t}_{2}^3$, and $\boldsymbol{t}_{3}^3$ located at $(1, 5)$, $(2, 3)$, and $(1, 1)$ respectively. Requester $\boldsymbol{r}_4$ has two tasks, $\boldsymbol{t}_{1}^4$ and $\boldsymbol{t}_{2}^4$, located at $(3, 5)$ and $(4, 2)$ respectively. Requester $\boldsymbol{r}_5$ has one task, $\boldsymbol{t}_{1}^5$, located at $(5, 3)$. These location points are depicted in the left-most figure of Figure \ref{Fig:3}. Now, we aim to group all these tasks into $k=3$ clusters based on their location points. Firstly, the set of cluster centroids $\boldsymbol{\mathcal{C}}$ is empty. By following the Algorithm \ref{algo:1}, we randomly select three initial cluster centroids from the available task location points, i.e., $(4, 1)$, $(1, 5)$, and $(5, 3)$ respectively. So, now the cluster centroids becomes $\boldsymbol{\mathcal{C}}= \{$(4, 1)$, $(1, 5)$, $(5, 3)$\}$. In the first iteration, we initialize the clusters $\xi_1$, $\xi_2$, and $\xi_3$ to be empty. Then, for each location point, we compute the Euclidean distance $\left(i.e., d = \sqrt{(x_2 - x_1)^2 + (y_2 - y_1)^2}\right)$  to each of the three cluster centroids and assign the point to the cluster whose centroid is closest to it. The location points $(4, 1)$, $(3, 2)$, $(2, 2)$, $(1, 1)$, and $(4, 2)$ are found to be closest to the same centroid compared to the other centroids. Therefore, they are grouped in cluster 1. Thus, cluster $\xi_1$ becomes: $\xi_1 = \{$(4, 1)$, $(3, 2)$, $(2, 2)$, $(1, 1)$, $(4, 2)$\}$.  Next, the algorithm computes the new centroid of cluster 1 by averaging the coordinates: $x = 4+3+2+1+4 / 5 = 2.8$, $y = 1+2+2+1+2 / 5 = 1.6$. Hence, the new cluster 1 centroid becomes $(2.8, 1.6)$. Similarly, the location points $(1, 5), (2, 3)$, and $(3, 5)$ are closer to the cluster 2 centroid compared to the remaining centroids. Therefore, these points are grouped into cluster 2. Thus, cluster $\xi_2$ becomes: $\xi_2 = \{(1, 5), (2, 3), (3, 5)\}$.
\begin{figure}[H]
\centering
\includegraphics[scale=0.50]{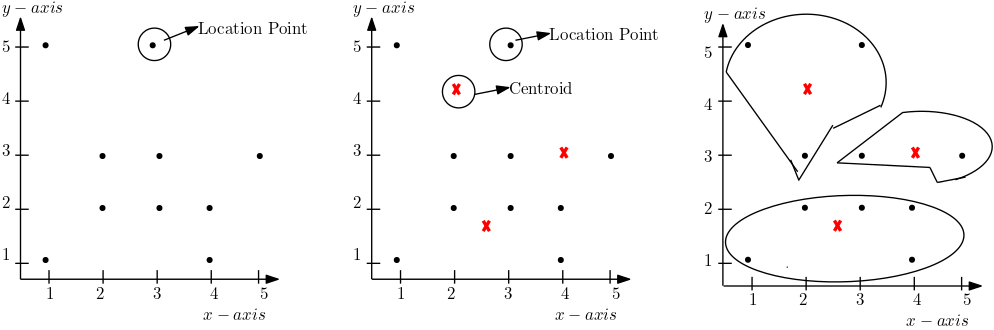}
\caption{Figure on the left shows the location points of tasks. The figure in the middle depicts the centroid points selected for the given location points. Given the centroid points and the location points of the tasks, the 3 clusters are formed.}
\label{Fig:3}
\end{figure}
 Next, the algorithm computes the new centroid of cluster 2 by averaging the coordinates: $x = 1+2+3 / 3 = 2$, $y = 5+3+5 / 3 = 4.3$. Hence, the new cluster 2 centroid becomes $(2, 4.3)$.   The remaining location points $(3, 3), (5, 3)$ are closer to the cluster 3 centroid and therefore formed cluster 3. Thus, cluster $\xi_3$ becomes: $\xi_3 = \{($(3, 3), (5, 3)$\}$. The centroid is calculated as:  $x = 3 + 5 / 2 = 4$, $y = 3 + 3 / 2 = 3$. Thus, the cluster $\xi_3$ new centroid becomes $(4, 3)$. These new centroids now serve as updated cluster centers. \\
\indent In the second iteration, the algorithm again computes the Euclidean distance of each task location point to the new centroids obtained in previous iterations. The location point $(4, 1)$, $(3, 2)$, $(2, 2)$, $(1, 1)$, and $(4, 2)$ are closer to the centroid $(2.8, 1.6)$, therefore it remains in same cluster $\xi_1$. Similarly, the location points $(1, 5), (2, 3)$, and $(3, 5)$ are closer to the centroid $(2, 4.3)$ and stays in the same cluster $\xi_2$. The remaining location points $(3, 3)$, and $(5, 3)$ are closer to the centroid $(4, 3)$ and remains in the same cluster $\xi_3$. Since no change in cluster centroids occurs in this iteration, the algorithm converges. The centroids determined after the second iteration are shown in the second figure of Figure \ref{Fig:3}. \\
\indent Finally, the Algorithm \ref{algo:1} returns the output as the set of clusters formed, i.e., $\xi_1$ containing locations points $(4, 1)$, $(3, 2)$, $(2, 2)$, $(1, 1)$, and $(4, 2)$ at centroid $(2.8, 1.6)$; $\xi_2$ containing locations points $(1, 5)$, $(2, 3)$, and $(3, 5)$ at centroid $(2, 4.3)$; and $\xi_3$ containing locations points $(3, 3)$, and $(5, 3)$ at centroid $(4, 3)$. The final clusters formed are shown in the right-most figure of Figure \ref{Fig:3}.

\subsection{Selection of Quality Task Executors}
Once the clusters are formed, the first objective is to select the quality task executors among the available task executors. For this purpose, the idea of the \emph{majority voting rule} \cite{levin2010review5, roughgarden2016cs269i} is utilized. In Algorithm \ref{algo:2}, line 1 initializes  ${\boldsymbol{e}}''$ and ${\boldsymbol{e}}'$ by $ \boldsymbol{e}$, $W$ to $\phi$, $i$ to 1, and $T$ to $\phi$. In lines 2 - 28, the task executor's quality is determined. In lines 2 - 28, the while loop will iterate until the condition in line 2 is satisfied. In line 3, the \textit{f} task executors are selected from ${\boldsymbol{e}}'$ that are to be voted and are held in $\eta$. In line 4, a random selection of task executors is made from the set ${\boldsymbol{e}}'' \setminus \eta$ and stored in $\alpha$. The completed tasks by the task executors in $\eta$ are ranked by the task executors in $\alpha$. For the ranking purposes, the idea of the majority voting rule is utilized. From lines 5 - 7, each task executor $\boldsymbol{e}_i$ in $\alpha$ is assigned a random preference ordering over the task executors in $\eta$ and is stored in  $P(\boldsymbol{e}_i)$. Lines 8 - 10 initialize vote counts for all the executors in the selected group to 0. In lines 11-14, the first preference of each of the task executors $\boldsymbol{e}_i \in \alpha$ is stored in $X$. Based on the number of first preferences received, the vote count is increased using line 13. In line 15, the maximum vote any task executor $\boldsymbol{e}_j$ received is stored in $M$.\\
\indent Further, the set of task executors receiving the first preference equal to the maximum number of first preferences received by any task executor $\boldsymbol{e}_j \in \eta$ is stored in $T$.  In lines 17 - 24, if there is a tie between the selected executor's votes, additionally, we have to check their next preferences until one executor gains a majority of votes among the tied ones. Line 25 selects the executor with the most votes as the winner for that round (or iteration).  Line 26 adds the winning task executor to the winning set $W$. Line 27 removes the evaluated executors from the set ${\boldsymbol{e}}'$. Lines 2 - 28 repeat until all executors get ranked. In line 29, the set of quality task executors is returned.
\subsubsection{\textsc{Illustrative of Selection of Quality Task Executors with an Example}}
\label{algo:Example}
Let us consider an example to understand Algorithm \ref{algo:2}. Let us suppose that there are 12 task executors $\boldsymbol{e} =$ $\{\boldsymbol{e}_1, \boldsymbol{e}_2, \dots, \boldsymbol{e}_{12}\}$ in one of the clusters. Now, out of these 12 task executors, the goal is to select the subset of quality task executors. For this purpose, let us apply Algorithm \ref{algo:2}. We have considered randomly selecting the few task executors in the first iteration, i.e., $\eta =$ $\{\boldsymbol{e}_3, \boldsymbol{e}_7, \boldsymbol{e}_4, \boldsymbol{e}_2\}$ that are to be voted on and the remaining task executors $\alpha =$ $\{\boldsymbol{e}_1, \boldsymbol{e}_5, \boldsymbol{e}_6, \boldsymbol{e}_8, \boldsymbol{e}_9, \boldsymbol{e}_{10}, \boldsymbol{e}_{11}, \boldsymbol{e}_{12}\}$ are provide preference over the set of $\eta$.  Now, each task executor in $\alpha$ gives preference over the selected task executors in $\eta$. The task executor $\boldsymbol{e}_1$ gives preference to \{$e_4, e_3, e_7, e_2$\}. Further, $e_5$ gives preference to \{$\boldsymbol{e}_3, \boldsymbol{e}_2, \boldsymbol{e}_4, \boldsymbol{e}_7$\}. $\boldsymbol{e}_6$ gives preference to \{$\boldsymbol{e}_4, \boldsymbol{e}_3, \boldsymbol{e}_7, \boldsymbol{e}_2$\}. In a similar manner $\boldsymbol{e}_8, \boldsymbol{e}_9, \boldsymbol{e}_{10}, \boldsymbol{e}_{11}$, and $\boldsymbol{e}_{12}$ are giving preferences on \{$\boldsymbol{e}_4, \boldsymbol{e}_3, \boldsymbol{e}_2, \boldsymbol{e}_7$\}; \{$\boldsymbol{e}_3, \boldsymbol{e}_7, \boldsymbol{e}_4, \boldsymbol{e}_2$\}; \{$\boldsymbol{e}_4, \boldsymbol{e}_3, \boldsymbol{e}_2, \boldsymbol{e}_7$\}; \{$\boldsymbol{e}_4, \boldsymbol{e}_2, \boldsymbol{e}_3, \boldsymbol{e}_7$\}; and \{$\boldsymbol{e}_7, \boldsymbol{e}_4, \boldsymbol{e}_2, \boldsymbol{e}_3$\} respectively. Here, $\boldsymbol{e}_3$ got 2 votes, $\boldsymbol{e}_7$ got 1 vote, $\boldsymbol{e}_4$ got 5 votes, and $\boldsymbol{e}_2$ got 0 votes.
\begin {algorithm}[H]
\caption{\textsc{Selection of Quality Task Executors}} 
\label{algo:2}
\noindent
\begin{algorithmic}[1]
\REQUIRE Set of Task Executors.
\ENSURE Set of quality Task Executors.
\STATE ${\boldsymbol{e}}'' = {\boldsymbol{e}}' =\boldsymbol{e}, W \gets \phi, i \gets 1, T \gets \phi$ 
\WHILE{${\boldsymbol{e}}' \neq \phi$ }
\STATE $\eta \gets$ \textsc{Rand Select} (${\boldsymbol{e}}', f)$ \COMMENT{\textcolor{blue}{Selects \textit{f} task executors from ${\boldsymbol{e}}'$ that are to be voted.}}
\STATE $\alpha \gets$  \textsc{Rand Select} (${\boldsymbol{e}}''\setminus \eta,~ g$) 
\COMMENT{\textcolor{blue}{Selects $g$ task executors from ${\boldsymbol{e}}'' \setminus \eta$ that will provide preferences over the completed tasks of task executors in $\eta$.}}
\FOR{each ${\boldsymbol{e}}_i \in \alpha$}
\STATE $P(\boldsymbol{e}_i) \gets$ \textsc{Rand preference}($\eta)$
\COMMENT{\textcolor{blue}{ Assign random preference order over $\eta$ to voter $e_i$ and stored in $P(\boldsymbol{e}_i)$.}}
\ENDFOR
\FOR{each ${\boldsymbol{e}}_j \in \eta$}
\STATE vote$(\boldsymbol{e}_j) \gets 0$
\COMMENT{\textcolor{blue}{Initialize vote count for each task executor $\boldsymbol{e}_j$ $\in$ $\alpha$.}}
\ENDFOR
\FOR{each ${\boldsymbol{e}}_i \in \alpha$}
\STATE $X \gets$ First preference $(P(\boldsymbol{e}_i))$ \COMMENT{\textcolor{blue}{Each voter gives one vote to its top preferred task executor}}
\STATE $vote(X) \gets vote(X)+1$ \COMMENT{\textcolor{blue}{Keeps track of number of votes received by each task executor $\boldsymbol{e}_j \in \alpha$.}} 
\ENDFOR
\STATE $M \gets$ max(vote($\boldsymbol{e}_j$)), $\forall{\boldsymbol{e}_j \in \eta}$ \COMMENT{\textcolor{blue}{$M$ holds the maximum value of votes that is received by any task executor $\boldsymbol{e}_j$.}}
%\STATE $T \gets$ {$\boldsymbol{e}_j \in \eta|$ vote$(\boldsymbol{e}_j) = M$} \COMMENT{\textcolor{blue}{Set of executors receiving max votes.}}
\STATE $T \gets T \cup \{\boldsymbol{e}_j| $ vote$(\boldsymbol{e}_j) = M, \boldsymbol{e}_j\in \eta \}$ \COMMENT{\textcolor{blue}{$T$ holds the task executors $\boldsymbol{e}_j \in \eta$ whose number of votes equal to the maximum vote by any task executor stored in $M$.}} 
\WHILE{$|T| > 1$ and remaining preferences exist} 
\FOR{ each ${\boldsymbol{e}}_i \in \alpha$}
\STATE  $X \gets$ Next preference $(P(\boldsymbol{e}_i), T)$ \COMMENT{\textcolor{blue}{Each voter votes for the next preferred executor among tied ones.}}
\STATE $vote(X) \gets vote(X)+1$ \COMMENT{\textcolor{blue}{Keeps track of number of votes received by each task executor $\boldsymbol{e}_i \in \alpha$.}}
\ENDFOR
\STATE $M \gets$ max(vote($\boldsymbol{e}_j$)), $\forall {\boldsymbol{e}_j \in T}$  \COMMENT{\textcolor{blue}{$M$ holds the maximum value of votes that is received by any task executor $\boldsymbol{e}_j \in T$.}}
\STATE $T \gets T \cup \{\boldsymbol{e}_j| $ vote$(\boldsymbol{e}_j) = M, \boldsymbol{e}_j\in \eta \}$ \COMMENT{\textcolor{blue}{$T$ holds the task executors $\boldsymbol{e}_j \in \eta$ whose number of votes equal to the maximum vote by any task executor stored in $M$.}}   
\ENDWHILE
\STATE $i \gets argmax(vote(\boldsymbol{e}_j)), \boldsymbol{e}_j \in T$ \COMMENT{\textcolor{blue}{Executor with majority votes is the quality task executor of this iteration.}}
\STATE $W \gets W \cup \{i\}$ \COMMENT{\textcolor{blue}{Add the quality task executors in set $W$.}}
\STATE ${\boldsymbol{e}}' \gets {\boldsymbol{e}}'\setminus \eta$ \COMMENT{\textcolor{blue}{Remove the evaluated executors from the set ${\boldsymbol{e}}'$.}}
\ENDWHILE
\STATE Return $W$ \COMMENT{\textcolor{blue}{Return the set of quality task executors selected through majority voting.}}
\end{algorithmic}
\end{algorithm}
  Now, the task executor who got the maximum votes, that task executor is selected as the winning executor, i.e., $\boldsymbol{e}_4$  in the first iteration. In the second iteration, the task executors in $\eta =$ $\{\boldsymbol{e}_9, \boldsymbol{e}_8, \boldsymbol{e}_{10}, \boldsymbol{e}_{11}\}$ that are to be voted on and the remaining task executors $\alpha =$ $\{\boldsymbol{e}_1, \boldsymbol{e}_2, \boldsymbol{e}_3, \boldsymbol{e}_4, \boldsymbol{e}_5, \boldsymbol{e}_6, \boldsymbol{e}_7, \boldsymbol{e}_{12}\}$ are provide preference over the set of $\eta$. Now, each task executor in $\alpha$ gives preference over the selected task executors in $\eta$. The task executor $\boldsymbol{e}_1$ gives preference to \{$\boldsymbol{e}_9, \boldsymbol{e}_{11}, \boldsymbol{e}_8, \boldsymbol{e}_{10}$\}. Further, $\boldsymbol{e}_2$ gives preference to \{$\boldsymbol{e}_{11}, \boldsymbol{e}_9, \boldsymbol{e}_8, \boldsymbol{e}_{10}$\}. $\boldsymbol{e}_3$ gives preference to \{$\boldsymbol{e}_{11}, \boldsymbol{e}_8, \boldsymbol{e}_{10}, \boldsymbol{e}_9$\}. In a similar manner $\boldsymbol{e}_4, \boldsymbol{e}_5, \boldsymbol{e}_6, \boldsymbol{e}_7$, and $\boldsymbol{e}_{12}$ are giving preferences on \{$\boldsymbol{e}_9, \boldsymbol{e}_{10}, \boldsymbol{e}_8, \boldsymbol{e}_{11}$\}; \{$\boldsymbol{e}_8, \boldsymbol{e}_{11}, \boldsymbol{e}_{10}, \boldsymbol{e}_9$\}; \{$\boldsymbol{e}_{11}, \boldsymbol{e}_9, \boldsymbol{e}_{10}, \boldsymbol{e}_8$\}; \{$\boldsymbol{e}_9, \boldsymbol{e}_{10}, \boldsymbol{e}_8, \boldsymbol{e}_{11}$\}; and \{$\boldsymbol{e}_{11}, \boldsymbol{e}_8, \boldsymbol{e}_{10}, \boldsymbol{e}_9$\} respectively. Here, $\boldsymbol{e}_9$ got 3 votes, $\boldsymbol{e}_8$ got 1 vote, $\boldsymbol{e}_{10}$ got 0 votes, and $\boldsymbol{e}_{11}$ got 4 votes. Now, the task executor who got the maximum votes, that task executor is selected as the winning executor, i.e., $\boldsymbol{e}_{11}$  in the second iteration.
Similarly, repeat the same process until all the task executors have voted. In the third iteration, $\eta =$ $\{\boldsymbol{e}_1, \boldsymbol{e}_5, \boldsymbol{e}_{12}, \boldsymbol{e}_6\}$ that are to be voted on and the remaining task executors $\alpha =$ $\{\boldsymbol{e}_2, \boldsymbol{e}_3, \boldsymbol{e}_4, \boldsymbol{e}_7, \boldsymbol{e}_8, \boldsymbol{e}_9, \boldsymbol{e}_{10}, \boldsymbol{e}_{11}\}$ are provide preference over the set of $\eta$. Now, each task executor in $\alpha$ gives preference over the selected task executors in $\eta$. The task executor $\boldsymbol{e}_2$ gives preference to \{$\boldsymbol{e}_{12}, \boldsymbol{e}_6, \boldsymbol{e}_5, \boldsymbol{e}_1$\}. Further, $\boldsymbol{e}_3$ gives preference to \{$\boldsymbol{e}_{12}, \boldsymbol{e}_1, \boldsymbol{e}_5, \boldsymbol{e}_6$\}. $\boldsymbol{e}_4$ gives preference to \{$\boldsymbol{e}_6, \boldsymbol{e}_5, \boldsymbol{e}_{12}, \boldsymbol{e}_1$\}. In a similar manner $\boldsymbol{e}_7, \boldsymbol{e}_8, \boldsymbol{e}_9, \boldsymbol{e}_{10}$, and $\boldsymbol{e}_{11}$ are giving preferences on \{$\boldsymbol{e}_5, \boldsymbol{e}_{12}, \boldsymbol{e}_1, \boldsymbol{e}_6$\}; \{$\boldsymbol{e}_6, \boldsymbol{e}_{12}, \boldsymbol{e}_1, \boldsymbol{e}_5$\}; \{$\boldsymbol{e}_1, \boldsymbol{e}_6, \boldsymbol{e}_5, \boldsymbol{e}_{12}$\}; \{$\boldsymbol{e}_5, \boldsymbol{e}_{12}, \boldsymbol{e}_6, \boldsymbol{e}_1$\}; and \{$\boldsymbol{e}_1, \boldsymbol{e}_{12}, \boldsymbol{e}_5, \boldsymbol{e}_6$\} respectively. Here, $\boldsymbol{e}_1$ got 2 votes, $\boldsymbol{e}_5$ got 2 votes, $\boldsymbol{e}_{12}$ got 2 votes, and $\boldsymbol{e}_6$ got 2 votes. Here we have a tie, all the selected task executors got 2 votes each. Now, by following lines 17 - 21 in Algorithm \ref{algo:2}, for the selected task executors, we have to check their second preferences in the list. By checking the second preference of each task executor in $\alpha$, the task executor $\boldsymbol{e}_{12}$ got 4 votes in the third iteration. Final winning task executors from all the iterations are $\{\boldsymbol{e}_4, \boldsymbol{e}_{11}, \boldsymbol{e}_{12}\}$. Repeat the same process for the remaining clusters. 
\subsection{Splitting and Equilibrium Price Determination}  
\begin{algorithm}[H]
\caption{\textsc{Splitting and Equilibrium Price Determination} ($\boldsymbol{r}$, $\boldsymbol{e}$)} 
\label{algo:3}
\noindent
\textbf{Output:} $\boldsymbol{p}^L$.
\begin{algorithmic}[1]
\STATE Divide the spatial crowdsourcing zone into two sub-zones: (1) Left Spatial Crowdsourcing Zone (LSCZ), and (2) Right Spatial Crowdsourcing Zone (RSCZ).\\
\STATE Randomly assign each task executors $\boldsymbol{e}_i \in \boldsymbol{e}$ and task requesters $\boldsymbol{r}_i \in \boldsymbol{r}$ independently to LSCZ and RSCZ respectively with probability 1/2.\\ 
\textcolor{blue}{/* Equilibrium price determination in LSCZ */}
\STATE $\boldsymbol{p} \gets 0$, $d^{\boldsymbol{r}(L)} = \infty$, $s^{\boldsymbol{e}(L)} = 0$     \COMMENT{\textcolor{blue}{initialized the variables to $0$ and $\infty$.}}
\WHILE{$d^{\boldsymbol{r}(L)} \neq s^{\boldsymbol{e}(L)}$ }
\STATE $\boldsymbol{p} \gets \boldsymbol{p}+\epsilon$   \COMMENT{\textcolor{blue}{The price $\boldsymbol{p}$ is increased each time by $\epsilon$.}}
\FOR{each $\boldsymbol{r}_i \in \boldsymbol{r}$}
\STATE $d_i^ {\boldsymbol{r}} (\boldsymbol{p}) = \argmax_{Y \in [0, \chi.\beta ]} u_i^{\boldsymbol{r}}(Y, \boldsymbol{p})$
\COMMENT{\textcolor{blue}{Determines the maximum demand of task requester $\boldsymbol{r}_i$ at price $\boldsymbol{p}$ and hold the demand in $d_i^ {\boldsymbol{r}} (\boldsymbol{p})$.}}
\ENDFOR
\FOR{each $\boldsymbol{e}_i \in \boldsymbol{e}$}
\STATE $s_i^ {\boldsymbol{e}} (\boldsymbol{p}) = \argmax_{Y \in [0, \chi.\beta ]} u_i^{\boldsymbol{e}}(Y, \boldsymbol{p})$
\COMMENT{\textcolor{blue}{Determines the maximum supply of task executor $\boldsymbol{e}_i$ at price $\boldsymbol{p}$ and hold the supply in  $s_i^ {\boldsymbol{e}} (\boldsymbol{p})$.}}
\ENDFOR
\STATE $d^ {\boldsymbol{r}} (\boldsymbol{p}) = \displaystyle\sum_{i = 1}^{n}  d_i^ {\boldsymbol{r}} (\boldsymbol{p})$
\COMMENT{\textcolor{blue}{Total demand at price $\boldsymbol{p}$ is calculated for all the $n$ task requesters and is stored in $d^ {\boldsymbol{r}} (\boldsymbol{p})$.}}
\STATE $s^ {\boldsymbol{e}} (\boldsymbol{p}) = \displaystyle\sum_{i = 1}^{m}  s_i^ {\boldsymbol{e}} (\boldsymbol{p})$
\COMMENT{\textcolor{blue}{Total supply at price $\boldsymbol{p}$ is calculated for all the $m$ task executors and is stored in $s^ {\boldsymbol{e}} (\boldsymbol{p})$.}}
\ENDWHILE
\STATE $\boldsymbol{p}^L \gets \boldsymbol{p}$     \COMMENT{\textcolor{blue}{The equilibrium price of LSCZ is stored in $\boldsymbol{p}^L$.}}
\STATE Returns $\boldsymbol{p}^L$.         \COMMENT{\textcolor{blue}{Returns equilibrium price of LSCZ.}}
\end{algorithmic}
\end{algorithm}
In this section, details of the Splitting and Equilibrium Price Determination are carried out. The spatial crowdsourcing zone is divided into two sub-zones, namely the Left spatial crowdsourcing Zone (LSCZ) and the Right spatial crowdsourcing Zone (RSCZ) in line 1 of the Algorithm \ref{algo:3}. In line 2, randomly assigned each task executor or task requester independently to LSCZ or RSCZ, respectively, with a 1/2 probability. Now, we have considered the LSCZ to determine the equilibrium price. Similarly, one can determine the equilibrium price of RSCZ. In line 3, initially, the equilibrium price $\boldsymbol{p}$ and supply $s^{\boldsymbol{e}(L)}$ are initialized as 0, and the demand $d^{\boldsymbol{r}(L)}$ at price $\boldsymbol{p} = 0$ is initialized as $\infty$. The while loop is taking care of determining the demand and supply of each task requester and task executor from lines 4-14, and once the demand and supply become equal, it terminates the loop. For every iteration of the while loop, the price $\boldsymbol{p}$ is increased by $\epsilon$. In lines 6-8, each task requester in LSCZ has a demand calculated at a price $\boldsymbol{p}$. Similarly, the supply of each task executor in LSCZ at price $\boldsymbol{p}$ is computed. The entire supply and demand for agents in $d^{\boldsymbol{r}(L)}$ and $s^{\boldsymbol{e}(L)}$, respectively, are shown in lines 12 and 13. The equilibrium price $\boldsymbol{p}$ is stored in $\boldsymbol{p}^L$ in line 15. Finally, the algorithm returns the list of $\boldsymbol{p}^L$ that equilibrates total supply and demand in the LSCZ in line 16. Similarly, the same process will run for the RSCZ to determine the equilibrium price $\boldsymbol{p}^R$. 
\subsubsection{Illustrative Example for Splitting and Equilibrium Price Determination}
\label{Example1}
Let us say, we have a set of 4 task requesters $\boldsymbol{r} = \{\boldsymbol{r}_1, \boldsymbol{r}_2, \boldsymbol{r}_3, \boldsymbol{r}_4\}$ and each task requester is endowed with multiple distinct tasks, i.e., $\boldsymbol{r}_1 = \{\boldsymbol{t}_1^1, \boldsymbol{t}_2^1, \boldsymbol{t}_3^1\}$ along with budget $\{12, 9, 5\}$; $\boldsymbol{r}_2 = \{\boldsymbol{t}_1^2, \boldsymbol{t}_2^2, \boldsymbol{t}_3^2\}$ along with budget $\{15, 10, 6\}$; $\boldsymbol{r}_3 = \{\boldsymbol{t}_1^3, \boldsymbol{t}_2^3, \boldsymbol{t}_3^3\}$ along with budget $\{10, 12, 4\}$; similarly, $\boldsymbol{r}_4 = \{\boldsymbol{t}_1^4, \boldsymbol{t}_2^4, \boldsymbol{t}_3^4\}$ along with budget $\{5, 4, 7\}$ respectively. On the other side, we have 3 quality task executors that are determined in subsection \ref{algo:Example}, i.e., $\boldsymbol{e} = \{\boldsymbol{e}_4, \boldsymbol{e}_{11}, \boldsymbol{e}_{12}\}$. Each task executor has a preference list over a bundle of tasks along with an estimated budget. Say, $\boldsymbol{e}_4$ gives preference over a bundle of tasks $\{\boldsymbol{t}_1^1, \boldsymbol{t}_1^3, \boldsymbol{t}_3^3\}$ with budget $\{10, 9, 9\}$; $\boldsymbol{e}_{11}$ gives preference over a bundle of tasks $\{\boldsymbol{t}_1^4, \boldsymbol{t}_3^2\}$ with budget $\{4, 8\}$; and  $\boldsymbol{e}_{12}$ gives preference over a bundle of tasks $\{\boldsymbol{t}_1^2, \boldsymbol{t}_2^4, \boldsymbol{t}_2^2, \boldsymbol{t}_3^4\}$ with budget $\{10, 9, 6, 4\}$ respectively. Now, by following the algorithm \ref{algo:3}, split the spatial crowdsourcing zone into two sub-zones, namely (1) Left Spatial Crowdsourcing Zone (LSCZ), and (2) Right Spatial Crowdsourcing Zone (RSCZ) by randomly assigning each task requester and task executor independently to LSCZ and RSCZ, respectively, with probability 1/2. After splitting, in LSCZ, we have task requesters $\boldsymbol{r}_1$, $\boldsymbol{r}_3$ along with their tasks and budgets, and task executors $\boldsymbol{e}_4$, $\boldsymbol{e}_{11}$ along with their preference bundle of tasks with budget. Let us first determine the equilibrium price $\boldsymbol{p}$ in LSCZ. So, let us initialize the price $\boldsymbol{p}$ to 0. At this price, all the task requesters want their tasks to get executed (i.e., $d^{\boldsymbol{r}(L)} = 6$), but on the other hand, none of the task executors want to execute the tasks (i.e., $s^{\boldsymbol{e}(L)} = 0$). Line 4 of the Algorithm \ref{algo:3} is satisfied, and $\boldsymbol{p}$ value is incremented by 3. Now at price $\boldsymbol{p} = 3$, the demand from the task requesters is 6 (i.e., $d^{\boldsymbol{r}(L)} = 6$) and the supply of the task executors is zero (i.e., $s^{\boldsymbol{e}(L)} = 0$). Again, the $\boldsymbol{p}$ value is increased by 3 since the condition in line 4 of the algorithm \ref{algo:3} is satisfied. Now $\boldsymbol{p}$ is 6. The $d^{\boldsymbol{r}(L)} = 4$ and $s^{\boldsymbol{e}(L)} = 1$, still, the demand and supply are not equal. Once again, the $\boldsymbol{p}$ value is incremented by 3, at $\boldsymbol{p} = 9$. The demand of the task requesters is 4 (i.e., $d^{\boldsymbol{r}(L)} = 4$), and the supply of the task executors is 4 (i.e., $s^{\boldsymbol{e}(L)} = 4$). Now the demand and supply are equal (i.e., $d^{\boldsymbol{r}(L)} = s^{\boldsymbol{e}(L)}$). So the equilibrium price $\boldsymbol{p}$ in the LSCZ is 9. Similarly, the equilibrium price of RSCZ can be calculated, and it is 6. 
\subsection{Demand and Supply Calculation}
In this section, the equilibrium price of LSCZ is used to determine the demand and supply of task requesters and task executors in RSCZ, respectively, in Algorithm \ref{algo:4}. In lines 1-7, for each task requester $\boldsymbol{r}_i \in \boldsymbol{r}^R$, the individual demand is computed. Specifically, in line 2, each requester determines its optimal demand $d_i^ {\boldsymbol{r}} (\boldsymbol{p}^L)$ by maximizing its utility over all feasible allocations $Y \in [0, \chi.\beta]$. In line 3, it is checked whether the obtained demand is positive or not. If the demand is greater than zero, then in line 4, the corresponding task requester is added to the winning task requesters set $\Tilde{\boldsymbol{r}}^R$. Subsequently, in line 5, the total demand  $d^{\boldsymbol{r}(R)}$ in RSCZ is updated by adding the demand of that task requester. This process continues for all task requesters, and only those with positive demand at price $\boldsymbol{p}^L$ are considered.\\
\indent In lines 8-14, for each task executor $\boldsymbol{e}_i \in \boldsymbol{e}^R$, the supply is computed. In line 9, each task executor determines its optimal supply $s_i^ {\boldsymbol{e}} (\boldsymbol{p}^L)$ by maximizing its utility over all feasible allocations $Y \in [0, \chi.\beta]$. In line 10, it is verified whether the computed supply is positive or not. If the supply is greater than zero, then in line 11, the task executor is added to the set of active task executors $\Tilde{\boldsymbol{e}}^R$. In line 12, the total supply $s^{\boldsymbol{e}(R)}$ is updated by adding the supply of that executor. This process is repeated for all task executors, and only those with positive supply at price $\boldsymbol{p}^L$ are retained. Finally, in line 15, the algorithm returns the total demand $d^{\boldsymbol{r}(R)}$ and total supply $s^{\boldsymbol{e}(R)}$ computed for RSCZ.

\begin{algorithm}[H]
\caption{\textsc{Demand and Supply Calculation} ($\boldsymbol{e}^R$, $\boldsymbol{r}^R$, $\boldsymbol{p}^L$).} 
\label{algo:4}
\noindent
\textbf{Output:} $d^{\boldsymbol{r}(R)} \gets 0, s^{\boldsymbol{e}(R)} \gets 0$
\begin{algorithmic}[1]
%\STATE /*\textcolor{blue}{Demand of task executors in RSCZ at price $p^L$ }*/
\FOR{each $\boldsymbol{r}_i \in \boldsymbol{r}^R$}
\STATE $d_i^ {\boldsymbol{r}} (\boldsymbol{p}^L) = \argmax_{Y \in [0, \chi.\beta ]} u_i^{\boldsymbol{r}}(Y, \boldsymbol{p}^L)$ \COMMENT{\textcolor{blue}{Calculating the demand of $\boldsymbol{r}_i$ at equilibrium price $\boldsymbol{p}^L$.}}
\IF{$d_i^ {\boldsymbol{r}} (\boldsymbol{p}^L)  > 0$}
\STATE $\Tilde{\boldsymbol{r}}^R \gets \Tilde{\boldsymbol{r}}^R \cup \{\boldsymbol{r}_i\}$   \COMMENT{\textcolor{blue}{Each time $\Tilde{\boldsymbol{r}}^R$ 
holds the task requester $\boldsymbol{r}_i$ if the criteria in line 3 is satisfied.}}
\STATE  $d^{\boldsymbol{r}(R)} =  d^{\boldsymbol{r}(R)} +  d_i^{\boldsymbol{r}(R)}(\boldsymbol{p}^L)$      \COMMENT{\textcolor{blue}{Total demand in RSCZ is calculated and held in $d^{\boldsymbol{r}(R)}$.}}
\ENDIF
\ENDFOR
\\
\textcolor{blue}{/* Supply of task executors in RSCZ at price $\boldsymbol{p}^L$ */}\\
\FOR{each $\boldsymbol{e}_i \in \boldsymbol{e}^R$}
\STATE $s_i^ {\boldsymbol{e}} (\boldsymbol{p}^L) = \argmax_{Y \in [0, \chi.\beta ]} u_i^{\boldsymbol{e}}(Y, \boldsymbol{p})$     \COMMENT{\textcolor{blue}{Each time $\boldsymbol{e}^R$ holds the task executor $\boldsymbol{e}_i$ if the criteria in line 10 is satisfied.}}
\IF{$s_i^ {\boldsymbol{e}} (\boldsymbol{p}^L) > 0$}
\STATE $\Tilde{\boldsymbol{e}}^R \gets \Tilde{\boldsymbol{e}}^R \cup \{\boldsymbol{e}_i\}$       \COMMENT{\textcolor{blue}{Each time $\Tilde{\boldsymbol{e}}^R$ 
holds the task executor $\boldsymbol{e}_i$ if the criteria in line 10 is satisfied.}}
\STATE $s^{\boldsymbol{e}(R)}= s^{\boldsymbol{e}(R)} + s_i^ {\boldsymbol{e}(\boldsymbol{R})}( \boldsymbol{p}^L)$     \COMMENT{\textcolor{blue}{Total supply in RSCZ is calculated and held in $s^{\boldsymbol{e}(R)}$.}}
\ENDIF
\ENDFOR
\STATE Return $d^{\boldsymbol{r}(R)}, s^{\boldsymbol{e}(R)}$    \COMMENT{\textcolor{blue}{Returns, total demand and supply from RSCZ.}}
\end{algorithmic}
\end{algorithm}
\subsubsection{Illustrative Example for Demand and Supply Calculation}
\label{Example3}
We are considering the equilibrium price calculated in LSCZ  determined in subsection \ref{Example1}. In RSCZ, we have 2 task requesters $\boldsymbol{r}_2$, $\boldsymbol{r}_4$, along with their tasks and budgets, i.e.,  $\boldsymbol{r}_2 = \{\boldsymbol{t}_1^2, \boldsymbol{t}_2^2, \boldsymbol{t}_3^2\}$ along with its budget $\{15, 10, 6\}$ and $\boldsymbol{r}_4 = \{\boldsymbol{t}_1^4, \boldsymbol{t}_2^4, \boldsymbol{t}_3^4\}$ along with its budget $\{5, 4, 7\}$ respectively. Similarly, we have 1 task executors $\boldsymbol{e}_{12}$ along with their preferences on a bundle of tasks with a budget, i.e., $\boldsymbol{e}_{12}$ gives preference over a bundle of tasks $\{\boldsymbol{t}_1^2, \boldsymbol{t}_2^4, \boldsymbol{t}_2^2, \boldsymbol{t}_3^4\}$ with budget $\{10, 9, 6, 4\}$ respectively. 
By following the Algorithm \ref{algo:4}, in RSCZ at $\boldsymbol{p}^L = 9$, $d^{\boldsymbol{r}(R)}$ is 2 and $s^{\boldsymbol{r}(R)}$ is 3. Similarly, the demand and supply calculation for LSCZ at $\boldsymbol{p}^R = 6$, $d^{\boldsymbol{r}(L)}$ is 4, and $s^{\boldsymbol{r}(L)}$ is 1.
\subsection{Winner determination and payment}
This section explains the procedure for identifying winners and calculating their payments under RSCZ. The corresponding winner determination and payment rules for LSCZ follow directly from Algorithm \ref{algo:5}. After computing the demand of task requesters and the supply of task executors in RSCZ with respect to $\boldsymbol{p}^L$, the three possible cases can occur: (1) $d^{\boldsymbol{r}(R)} == s^{\boldsymbol{e}(R)}$, (2) $d^{\boldsymbol{r}(R)} > s^{\boldsymbol{e}(R)}$, and (3) $d^{\boldsymbol{r}(R)} < s^{\boldsymbol{e}(R)}$. In lines 1-7, the case where $d^{\boldsymbol{r}(R)} == s^{\boldsymbol{e}(R)}$ is handled. In this scenario, the total demand exactly matches the total supply; therefore, all the agents can be accommodated. The algorithm determines the set of winning task requesters $\boldsymbol{r}^W$, the set of winning task executors $\boldsymbol{e}^W$, and the payments for each. In lines 2-4, include each task requester's tasks from the $\boldsymbol{r}_i$ to the winning set $\boldsymbol{r}^W$. Similarly, in lines 5-7, include each task executor's $\boldsymbol{e}_i$ in the winning set $\boldsymbol{e}^W$. Lines 8-20 describe the procedure for the case when demand exceeds supply, i.e., $d^{\boldsymbol{r}(R)} > s^{\boldsymbol{e}(R)}$. 
\begin{algorithm}[H]
\caption{\textsc{Winner determination and payment} ($d^{\boldsymbol{r}(R)}$, $s^{\boldsymbol{e}(R)}$, $\boldsymbol{e}^R$, $\boldsymbol{r}^R$, $\boldsymbol{p}^L$).} 
\label{algo:5}
\noindent
\begin{algorithmic}[1]
\IF{$d^{\boldsymbol{r}(R)} == s^{\boldsymbol{e}(R)}$} %\COMMENT{\textcolor{blue}{Total demand equals to total supply.}}
\FOR{each $\boldsymbol{r}_i \in \boldsymbol{r}^R$} %\COMMENT{\textcolor{blue}{For each requester $\boldsymbol{r}_i^R$.}}
%\FOR{each $\boldsymbol{r}^R_{i,l} \in \boldsymbol{r}^R_i$} %\COMMENT{\textcolor{blue}{For each requester of $\boldsymbol{r}^R_{i,l}$.}}
\STATE $\boldsymbol{r}^W \gets \boldsymbol{r}^W \cup \{\boldsymbol{r}_{i}\}$ \COMMENT{\textcolor{blue}{Add all task requesters to the winning set}}
%\ENDFOR
\ENDFOR
\FOR{each $\boldsymbol{e}_i \in \boldsymbol{e}^R$} %\COMMENT{\textcolor{blue}{For each executor $\boldsymbol{e}_i^R$}}
%\FOR{each $\boldsymbol{e}^R_{i,q} \in \boldsymbol{e}^R_i$} %\COMMENT{\textcolor{blue}{For each executor capacity unit}}
\STATE $\boldsymbol{e}^W \gets \boldsymbol{e}^W \cup \{\boldsymbol{e}_{i}\}$ \COMMENT{\textcolor{blue}{Add all task executors to the winning set.}}
%\ENDFOR
\ENDFOR
\ELSIF{$d^{\boldsymbol{r}(R)} > s^{\boldsymbol{e}(R)}$} % \COMMENT{\textcolor{blue}{Demand exceeds supply}}

\STATE $\Tilde{\boldsymbol{r}}^R \gets EnQueue~(\boldsymbol{r}^R)$ \COMMENT{\textcolor{blue}{Store all the task requesters in a queue}}
\WHILE{$|\boldsymbol{r}^W| \neq s^{\boldsymbol{e}(R)}$} %\COMMENT{\textcolor{blue}{Continue until supply is fully allocated}}
\STATE $\boldsymbol{r}^* \gets dequeue~(\Tilde{\boldsymbol{r}}^R)$ \COMMENT{\textcolor{blue}{Remove next task requester from the queue}}
\FOR{each $\boldsymbol{r}_{i} \in \boldsymbol{r}^*$}
\WHILE{$ v_i > \boldsymbol{p}^R$ and $|\boldsymbol{r}^W| \neq s^{\boldsymbol{e}(R)}$} %\COMMENT{\textcolor{blue}{Select tasks whose value is below price threshold}}
\STATE $\boldsymbol{r}^W \gets \boldsymbol{r}^W \cup \{\boldsymbol{r}^R_{i}\}$ \COMMENT{\textcolor{blue}{Add selected task requester to the winning set.}}
\ENDWHILE
\ENDFOR
\ENDWHILE
\FOR{each $\boldsymbol{e}_{i} \in \boldsymbol{e}^R$} %\COMMENT{\textcolor{blue}{Assign executors corresponding to selected tasks.}}
\STATE $\boldsymbol{e}^W \gets \boldsymbol{e}^W \cup \{\boldsymbol{e}_{i}\}$  \COMMENT{\textcolor{blue}{Update winner set by including the chosen executor unit.}}
\ENDFOR
\ELSE  %\COMMENT{\textcolor{blue}{Supply exceeds demand.}}
\STATE $\Tilde{\boldsymbol{e}}^R \gets EnQueue~(\boldsymbol{e}^R)$ \COMMENT{\textcolor{blue}{Store all task executors in a queue.}}
\WHILE{$|\boldsymbol{e}^W| \neq d^{\boldsymbol{r}(R)}$}%\COMMENT{\textcolor{blue}{Continue until demand is fulfilled}}
\STATE $\boldsymbol{e}^* \gets dequeue(\Tilde{\boldsymbol{e}}^R)$ \COMMENT{\textcolor{blue}{Select the next task executor from the queue.}}
\FOR{each $\boldsymbol{e}_{i} \in \boldsymbol{e}^*$}%\COMMENT{\textcolor{blue}{Iterate over each capacity unit of the selected executor}}
\WHILE{$ v_i < \boldsymbol{p}^R$ and $|\boldsymbol{e}^W| \neq d^{\boldsymbol{r}(R)}$}%\COMMENT{\textcolor{blue}{Add executor units whose value satisfies the threshold until demand is met.}}
\STATE $\boldsymbol{e}^W \gets \boldsymbol{e}^W \cup \{\boldsymbol{e}^R_{i}\}$  \COMMENT{\textcolor{blue}{Add the selected executor unit to the winning set.}}
\ENDWHILE
\ENDFOR
\ENDWHILE
\FOR{each $\boldsymbol{r}_{i} \in \boldsymbol{r}^R$} %\COMMENT{\textcolor{blue}{Assign all requesters since demand is smaller}}
\STATE $\boldsymbol{r}^W \gets \boldsymbol{r}^W \cup \{\boldsymbol{r}_{i}\}$ \COMMENT{\textcolor{blue}{Add selected task requester to the winning set.}}
\ENDFOR
\ENDIF
\STATE Return $\boldsymbol{r}^W$, $\boldsymbol{e}^W$, $\boldsymbol{p}^L$ \COMMENT{\textcolor{blue}{Return winning task requesters, executors and market price.}}
\end{algorithmic}
\end{algorithm}
In line 9, the selected task requester's tasks are placed in the queue. The while loop in lines 10-17 continues until the number of selected winning task requester's task units becomes equal to the available supply. In line 11, one task requester's task is dequeued from the queue at a time. For each dequeued requester, the for loop in line 12 iterates through its task units. The stopping condition in line 13 checks whether the valuation of each task unit satisfies the feasibility condition with respect to the market price $\boldsymbol{p}^R$, and whether the supply constraint is still respected. If the condition is satisfied, the corresponding unit is added to the winning requester set $\boldsymbol{r}^W$ as shown in line 14. Once the supply is exhausted, the loop terminates. The for loop in lines 18-20 then iterates through all the task executors and places them into the winning executor set $\boldsymbol{e}^W$. In lines 21-34, the case with excess supply $d^{\boldsymbol{r}(R)} < s^{\boldsymbol{e}(R)}$ is addressed. In line 22, the selected task executors are placed into a queue in their arrival order. The while loop in lines 23-30 continues until the number of selected task executor units becomes equal to the total demand. In line 24, one task executor is dequeued from the queue at a time. The for loop in line 25 iterates over the preferred task list of the selected task executors.  The stopping condition in line 26 checks whether the cost of each task is feasible with respect to the market price $\boldsymbol{p}^R$ and whether the demand constraint is still satisfied. If the condition holds, the task is added to the winning task executor set $\boldsymbol{e}^W$ as shown in line 27. Once the demand is fulfilled, the loop terminates. Finally, the for loop in lines 31-33 iterates over all the task requesters' tasks and places them into the winning requester set $\boldsymbol{r}^W$, since all task requesters are accepted in the excess supply case. Finally, in line 35, the algorithm returns the set of winning task requesters $\boldsymbol{r}^W$, the set of winning task executors $\boldsymbol{e}^W$, and the clear price $\boldsymbol{p}^L$, which is used for determining the payment. 

\subsection{Illustrative Example for Winner determination and payment}
Now, let us apply Algorithm \ref{algo:5} to the output obtained from from subsection \ref{Example3}, i.e., in RSCZ at $\boldsymbol{p}^L = 9$, $d^{\boldsymbol{r}(R)}$ is 2, and $s^{\boldsymbol{e}(R)}$ is 3. For RSCZ, lines 21-34 of Algorithm \ref{algo:5} will be activated as $d^{\boldsymbol{r}(R)} < s^{\boldsymbol{e}(R)}$. Following line 22, the selected task executors' tasks are placed in queue $\Tilde{\boldsymbol{e}}^R = [\boldsymbol{t}_2^4, \boldsymbol{t}_2^2, \boldsymbol{t}_3^4]$. After that, using lines 23-30, the winning task executors' preferred tasks are determined by sequentially dequeuing each preferred executor's tasks and checking the feasibility condition with respect to the opposite market price $\boldsymbol{p}$, which is 9. The winning task executor's preferred tasks are placed in queue $\Tilde{\boldsymbol{e}}^R = [\boldsymbol{t}_2^4, \boldsymbol{t}_2^2, \boldsymbol{t}_3^4]$. In this, ready to execute the task is $\boldsymbol{e}^W = [\boldsymbol{t}_2^2]$ with budget 6 and utility is $9-6 = 3$. Following lines 31-33, the winning task requesters' tasks are determined by sequentially dequeuing each task and checking the feasibility condition with respect to the opposite market price $\boldsymbol{p}$, which is 9. The selected task requester's task is placed in queue $\Tilde{\boldsymbol{r}}^R = [\boldsymbol{t}_1^2, \boldsymbol{t}_2^2]$. In this, ready to execute the tasks is $\boldsymbol{r}^W = [\boldsymbol{t}_2^2]$ with  budget 10 and utility is $10-9 = 1$. The price at which the trading takes place is $\boldsymbol{p}^L = 9$.

\section{Economic Analysis and Probabilistic Guarantees}
\label{sec:TheA}
\begin{lemma}
\label{lemma:1}
TRUST-SC is computationally efficient.
\end{lemma}
\begin{proof}
\indent In Algorithm \ref{algo:1}, performs centroid-based clustering on the set of tasks. During each iteration, every task is compared with each of the $k$ cluster centers to determine the nearest cluster. Therefore, the assignment step requires $O(Tk)$ operations. After assigning tasks to clusters, the algorithm recomputes cluster centroids, which takes $O(T)$ time. Since the clustering process iterates until convergence, let $N $ denote the number of iterations required for stabilization. Thus, the total running time becomes $O(N n ^2 k)$.\\
\indent  In Algorithm \ref{algo:2}, let $m$ denote the number of task executors in a cluster. In each iteration, a subset of executors is evaluated, and the remaining executors provide preference rankings over them. Vote aggregation and preference processing require scanning the participating executors. In the worst case, the voting process and tie-breaking procedure may involve comparing all executors, leading to $O(m)$ operations per round. Since each executor can participate in at most one evaluation round, the overall running time is bounded by $O(m^2)$.\\
\indent In Algorithm \ref{algo:3}, first splits the spatial crowdsourcing zone into two sub-zones by randomly assigning each requester and executor, which takes $O(n+m)$ time. It then iteratively increases the price until market equilibrium is reached. In each iteration, the algorithm computes the demand of all $n$ requesters and the supply of all $m$ executors, which requires $O(n+m)$ time. Let $L$ denote the number of price increments required to reach equilibrium. Therefore, the total running time of Algorithm 3 is $O(L(n+m))$.\\
\indent In Algorithm \ref{algo:2}, the demand of requesters and the supply of executors at the equilibrium price. This requires scanning all $n$ requesters and $m$ executors once, resulting in a running time of $O(n+m)$.\\
\indent In Algorithm \ref{algo:5}, determines the winning task requesters and executors based on the computed demand and supply values. The procedure involves queue operations and sequential checks over requester tasks and executor preferences. In the worst case, the algorithm scans all requesters and executors once, resulting in a running time of $O(n+m)$. \\
\indent The proposed TRUST-SC mechanism runs in polynomial time. Combining the complexities of Algorithms 1--5, the overall running time of the TRUST-SC mechanism is $O(N n ^2 k + m^2 + L(n+m))$.
Since $N$, $k$, and $L$ are bounded by polynomial factors in the input size, the overall complexity of the proposed mechanism is polynomial. Hence, TRUST-SC satisfies the computational efficiency requirement.
\end{proof}

\begin{lemma}
TRUST-SC is truthful.
\end{lemma}

\begin{proof}
The TRUST-SC mechanism consists of three major phases: (i) clustering of task executors (Algorithm~1), (ii) quality-based filtering of executors (Algorithm~2), and (iii) iterative double auction-based allocation and payment computation (Algorithms~3--5). Since the first two phases are independent of economic bids and asks, truthfulness depends entirely on the auction phase. Let us consider task requesters.\\
\indent Let $\boldsymbol{r}_i$ be a task requester who demands $Y$ units of completed tasks with true valuation $v_i(Y)$ and reported bid $\hat{v}_i(Y)$. Let $\boldsymbol{p}$ denote the per-unit payment determined by the mechanism. The utility of a requester $\boldsymbol{r}_i$ is given by $u_i^ {\boldsymbol{r}}$ as:

\begin{equation}
\label{equ:1}
u_i^ {\boldsymbol{r}} (Y, \boldsymbol{p}) = 
    \begin{cases}
        v_i(Y) - \boldsymbol{p} \cdot Y,~ if~\boldsymbol{r}_i~wins\\
        0,~ Otherwise
    \end{cases}
\end{equation}

The winner determination in TRUST-SC is performed iteratively by matching the highest bids with the lowest asks until no feasible trade exists. The payment is determined based on a critical threshold (uniform clearing price).

\paragraph{Case 1: Underbidding ($\hat{v}_i(Y) < v_i(Y)$)}

Let $\boldsymbol{r}_i$ misreport by lowering his bid. Two scenarios arise:

\emph{(i) Final iteration case:} If the $i$-th iteration is the final iteration where demand equals supply, then $\boldsymbol{r}_i$ remains a winner. The payment remains unchanged, and hence
\[
\hat{u}_i^ {\boldsymbol{r}} (Y, \boldsymbol{p}) = v_i(Y) - \boldsymbol{p} \cdot Y = u_i^ {\boldsymbol{r}} (Y, \boldsymbol{p}).
\]

\emph{(ii) Continuation case:} If the mechanism proceeds to the $(i+1)$-th iteration, the reduced bid may cause $\boldsymbol{r}_i$ to drop out of the allocation process. In this case, $\boldsymbol{r}_i$ becomes a loser and obtains zero utility. Thus, underbidding cannot increase utility and may strictly decrease it.

\paragraph{Case 2: Overbidding ($\hat{v}_i(Y) > v_i(Y)$)}

If $r_i^j$ overbids, two scenarios arise:

\emph{(i) No change in outcome:} If the requester would win even under truthful bidding, then the allocation and payment remain unchanged, and utility is unaffected.

\emph{(ii) Unprofitable winning:} Overbidding may cause $\boldsymbol{r}_i$ to win in a situation where the clearing price exceeds his true valuation, i.e., $\boldsymbol{p} \cdot Y > v_i(Y)$. In this case, the utility becomes
\[
\hat{u}_i^ {\boldsymbol{r}} (Y, \boldsymbol{p}) = v_i(Y) - \boldsymbol{p} \cdot Y < 0.
\]
which is strictly worse than not participating. Thus, overbidding cannot increase utility.

\medskip

Combining both cases, truthful reporting maximizes the utility of each task requester. In similar line, the proof can be carried out for task executors. \\

Since both task requesters and task executors maximize their utility by reporting truthfully, and no participant can gain by unilateral deviation, the TRUST-SC mechanism is truthful.
\end{proof}

\begin{lemma}
TRUST-SC is individually rational.
\end{lemma}
\begin{proof}
In the TRUST-SC mechanism, after clustering (Algorithm~1) and quality-based filtering (Algorithm~2), the allocation and pricing are determined in the auction phase (Algorithms~3--5). At each iteration, a price $\boldsymbol{p}$ is announced, and agents are queried for their willingness to trade at this price. Consider a task requester $\boldsymbol{r}_i$ demanding $Y$ units of tasks with valuation $v_i(Y)$. The requester agrees to participate if and only if $v_i(Y) \ge \boldsymbol{p} \cdot Y$. If $v_i(Y) \ge \boldsymbol{p} \cdot Y$, then $\boldsymbol{r}_i$ participates in the trade and pays $\boldsymbol{p} \cdot Y$. Hence, the utility is
\[
u_i^{\boldsymbol{r}} (Y, \boldsymbol{p}) = v_i(Y) - \boldsymbol{p} \cdot Y \geq 0.
\]

If $v_i(Y) < \boldsymbol{p}$, then the requester declines participation and exits the market. In this case, the utility is
\[
u_i^{\boldsymbol{r}} (Y, \boldsymbol{p}) = 0.
\]

Thus, in all cases, the utility of each task requester is non-negative. A similar argument can be given for the task executors.\\

Since both task requesters and task executors obtain non-negative utility under all circumstances, the TRUST-SC mechanism satisfies individual rationality. 
\end{proof}

\begin{lemma}
Under standard boundedness assumptions, this random splitting preserves market efficiency approximately with high probability.
\end{lemma}
\begin{proof}
Let

\[
d^ {\boldsymbol{r}} (\boldsymbol{p}) = \displaystyle\sum_{i = 1}^{n}  d_i^ {\boldsymbol{r}} (\boldsymbol{p}) \qquad \text{and} \qquad s^ {\boldsymbol{e}} (\boldsymbol{p}) = \displaystyle\sum_{i = 1}^{m}  s_i^ {\boldsymbol{e}} (\boldsymbol{p})
\]
denote the total demand and total supply in the full market at price $\boldsymbol{p}$, where $d_i^ {\boldsymbol{r}} (\boldsymbol{p})$ is the demand of requester $r_i$ and $s_i^ {\boldsymbol{e}} (\boldsymbol{p})$ is the supply of executor $e_j$. Let
\[
d^{\boldsymbol{r}(L)}(\boldsymbol{p}),\;s^{\boldsymbol{e}(L)}(\boldsymbol{p}),\;d^{\boldsymbol{r}(R)}(\boldsymbol{p}),\;s^{\boldsymbol{e}(R)}(\boldsymbol{p}),\;
%D_L(p),\; S_L(p),\; D_R(p),\; S_R(p)
\]
denote the corresponding quantities in LSCZ and RSCZ, respectively. Define the market imbalance at price $\boldsymbol{p}$ as
\[
\Gamma(\boldsymbol{p}) = d^{\boldsymbol{r}}(\boldsymbol{p})-s^{\boldsymbol{e}}(\boldsymbol{p})
%D(p)-S(p),
\]
and the imbalance in LSCZ and RSCZ as
\[
\Gamma_L(\boldsymbol{p})=d^{\boldsymbol{r}(L)}(\boldsymbol{p})-s^{\boldsymbol{e}(L)}(\boldsymbol{p}),  \qquad \Gamma_R(\boldsymbol{p})= d^{\boldsymbol{r}(R)}(\boldsymbol{p})-s^{\boldsymbol{e}(R)}(\boldsymbol{p})
%D_L(p)-S_L(p), \qquad \Gamma_R(p)=D_R(p)-S_R(p).
\]

Since each requester and executor is independently assigned to LSCZ with probability $1/2$, we have
\[
\mathbb{E}[\Gamma_L(\boldsymbol{p})] = \frac{\Gamma(\boldsymbol{p})}{2}, \qquad \mathbb{E}[\Gamma_R(\boldsymbol{p})] = \frac{\Gamma(\boldsymbol{p})}{2}.
\]

Assume that, for all requesters and executors,
\[
0 \le d_i^ {\boldsymbol{r}} (\boldsymbol{p}) \le \chi \cdot \beta, \qquad 0 \le s_i^ {\boldsymbol{e}} (\boldsymbol{p}) \le \chi \cdot \beta.
\]
Then the following approximation result holds.
\end{proof}
\begin{lemma}
For any fixed price $\boldsymbol{p}$ and any $t>0$,
\[
\Pr\!\left(\left|\Gamma_L(\boldsymbol{p})-\frac{\Gamma(\boldsymbol{p})}{2}\right| \ge t\right) \le 2\exp\!\left(-\frac{2t^2}{n (\chi \cdot \beta)^2 + m (\chi \cdot \beta)^2} \right).
\]
An identical bound holds for $\Gamma_R(\boldsymbol{p})$.
\end{lemma}

\begin{proof}
The quantity $\Gamma_L(\boldsymbol{p})$ can be written as
\[
\Gamma_L(\boldsymbol{p})=\sum_{i=1}^{n} X_i d_i^ {\boldsymbol{r}} (\boldsymbol{p}) - \sum_{j=1}^{m} Y_j s_i^ {\boldsymbol{e}} (\boldsymbol{p}),
\]
where $X_i$ and $Y_j$ are independent Bernoulli random variables with parameter $1/2$, representing whether requester $r_i$ and executor $\boldsymbol{e}_j$ are assigned to LSCZ. Therefore,
\[
\mathbb{E}[X_id_i^ {\boldsymbol{r}} (\boldsymbol{p})] = \frac{d_i^ {\boldsymbol{r}} (\boldsymbol{p})}{2}, \qquad \mathbb{E}[Y_js_i^ {\boldsymbol{e}} (\boldsymbol{p})] = \frac{s_i^ {\boldsymbol{e}} (\boldsymbol{p})}{2},
\]
which implies
\[
\mathbb{E}[\Gamma_L(\boldsymbol{p})] = \frac{\Gamma(\boldsymbol{p})}{2}.
\]
Since each term is bounded, Hoeffding's inequality yields the stated concentration bound.
\end{proof}
The above lemma shows that the imbalance observed in a split market is well concentrated around half of the imbalance in the original market. Therefore, the equilibrium price computed in one sub-market is a reliable estimator of the market state in the other sub-market. To connect this to efficiency, let $W(\boldsymbol{p})$ denote the total feasible trade volume at price $\boldsymbol{p}$, defined as
\[
W(\boldsymbol{p})=\min\{d^{\boldsymbol{r}}(\boldsymbol{p}), s^{\boldsymbol{e}}(\boldsymbol{p})\},
\]
and similarly define
\[
W_L(\boldsymbol{p})=\min\{d^{\boldsymbol{r}(L)}(\boldsymbol{p}),\;s^{\boldsymbol{e}(L)}(\boldsymbol{p})\}, \qquad W_R(\boldsymbol{p})=\min\{d^{\boldsymbol{r}(R)}(\boldsymbol{p}),\;s^{\boldsymbol{e}(R)}(\boldsymbol{p})\}.
\]
Then, by linearity of expectation,
\[
\mathbb{E}[d^{\boldsymbol{r}(L)}(\boldsymbol{p})] = \frac{d^{\boldsymbol{r}}(\boldsymbol{p})}{2}, \qquad \mathbb{E}[s^{\boldsymbol{e}(L)}(\boldsymbol{p})] = \frac{s^{\boldsymbol{e}}(\boldsymbol{p})}{2},
\]
and similarly for RSCZ. Hence, in expectation,
\[
\mathbb{E}[W_L(\boldsymbol{p})+W_R(\boldsymbol{p})] \approx W(\boldsymbol{p}),
\]
up to the loss induced by imbalance fluctuations due to random splitting.

\begin{theorem}[Approximate Market Efficiency under Random Splitting] Suppose the demand and supply contributions of all agents are bounded as above. Then, for any fixed price $\boldsymbol{p}$, the total feasible trade realized by the two split markets satisfies
\[
W_L(\boldsymbol{p})+W_R(\boldsymbol{p}) \ge W(\boldsymbol{p})-2t
\]
with probability at least
\[
1-4\exp\!\left(-\frac{2t^2}{n (\chi \cdot \beta)^2 + m (\chi \cdot \beta)^2} \right).
\]
\end{theorem}

\begin{proof}
By the concentration bound in the previous lemma, both $\Gamma_L(\boldsymbol{p})$ and $\Gamma_R(\boldsymbol{p})$ deviate from their expectations by at most $t$ with high probability. Since the realized trade volume in each sub-market is limited by the mismatch between demand and supply, the loss in feasible trade due to random splitting is at most proportional to the imbalance deviation. Summing the losses over the two sub-markets yields the bound
\[
W_L(\boldsymbol{p})+W_R(\boldsymbol{p}) \ge W(\boldsymbol{p})-2t.
\]
Applying the union bound over the two sub-markets completes the proof.
\end{proof}

The theorem above implies that the efficiency loss caused by the random splitting step is small with high probability. In particular, as the market size grows, the deviation term becomes negligible relative to the total market volume, showing that TRUST-SC preserves approximate market efficiency while benefiting from the truthfulness and robustness advantages of split-market price determination.

\begin{theorem}
In each cluster $\chi_i$, the expected number of quality task executors determined from among the available task executors, say, $m_i$, is given as $\mathbb{E}[|W|] \approx \frac{m_i}{f_i}$, where $f_i$ is the number of task executors selected for evaluating the other task executors in $\chi_i$ cluster.
\end{theorem}
\begin{proof}
Fix cluster $\chi_i$. In Algorithm \ref{algo:2}, during each iteration \emph{while} loop in lines 2-28, a subset $\eta \subseteq e'$ of size $f_i$ is selected for evaluation, and exactly one executor is chosen as the winner of that round through the majority voting rule. After that, the evaluated set $\eta$ is removed from further consideration. Therefore, each iteration reduces the number of remaining executors by $f$ and contributes exactly one quality task executor to the winning set $W$.

If $m_i$ is divisible by $f_i$, then the number of iterations is exactly
\[
\frac{m_i}{f_i},
\]
and hence the number of quality task executors selected is
\[
|W| = \frac{m_i}{f_i}.
\]

If $m_i$ is not divisible by $f_i$, then Algorithm \ref{algo:2} terminates after
\[
\left\lceil \frac{m_i}{f_i} \right\rceil
\]
iterations. Since one quality executor is selected in each iteration, the total number of quality task executors is
\[
|W| = \left\lceil \frac{m_i}{f_i} \right\rceil.
\]

Thus, in expectation, the number of quality task executors selected from the available task executors is given by
\[
\mathbb{E}[|W|] = \left\lceil \frac{m_i}{f_i} \right\rceil,
\]
or equivalently,
\[
\mathbb{E}[|W|] \approx \frac{m_i}{f_i}
\]
for large $m_i$.
\end{proof}

\subsection{Task Success Rate Analysis}

Let $T$ denote the set of assigned tasks, where $|T|=N$. For each task $i$, define an indicator random variable
\[
X_i =
\begin{cases}
1, & \text{if task } i \text{ is successfully completed},\\
0, & \text{otherwise}.
\end{cases}
\]

Then, the total number of successfully completed tasks is
\[
X = \sum_{i=1}^{N} X_i.
\]

Accordingly, the task success rate is defined as
\[
\text{TSR} = \frac{X}{N}\times 100.
\]

\begin{theorem}
Let $N$ be the total number of assigned tasks, and suppose that each selected quality executor completes an assigned task successfully with probability at least $p_q$. Then, the expected task success rate of the proposed mechanism satisfies
\[
\mathbb{E}[\text{TSR}] \ge p_q \times 100.
\]
\end{theorem}

\begin{proof}
For each task $i$, let $X_i$ be the indicator variable defined above. Since the success probability of each selected executor is at least $p_q$, we have
\[
\mathbb{E}[X_i] \ge p_q, \qquad \forall i.
\]

By linearity of expectation,
\[
\mathbb{E}[X] = \sum_{i=1}^{N}\mathbb{E}[X_i] \ge \sum_{i=1}^{N} p_q = Np_q.
\]

Since
\[
\text{TSR} = \frac{X}{N}\times 100,
\]
Taking the expectation on both sides yields
\[
\mathbb{E}[\text{TSR}] = \frac{\mathbb{E}[X]}{N}\times 100 \ge \frac{Np_q}{N}\times 100 = p_q \times 100.
\]

Hence proved.
\end{proof}

\begin{theorem}
Suppose the average success probability of all available executors is $\bar{p}$, while the average success probability of the selected quality executors is $\bar{p}_q$, where $\bar{p}_q > \bar{p}$. Then the expected task success rate under the proposed quality-selection mechanism is strictly higher than that of the non-quality-selection mechanism.
\end{theorem}

\begin{proof}
Without quality selection, the expected task success rate is
\[
\mathbb{E}[\text{TSR}_{\text{base}}] = \bar{p}\times 100.
\]

With quality selection, the expected task success rate is
\[
\mathbb{E}[\text{TSR}_{\text{quality}}] = \bar{p}_q\times 100.
\]

Since $\bar{p}_q > \bar{p}$, we obtain
\[
\mathbb{E}[\text{TSR}_{\text{quality}}] > \mathbb{E}[\text{TSR}_{\text{base}}].
\]

Hence, the proposed quality-selection mechanism strictly improves the expected task success rate.
\end{proof}

\begin{theorem}
Let $e^\ast$ be the true best executor. Suppose each of the $g$ voters independently votes in favor of $e^\ast$ with probability $p > \frac{1}{2}$. Then the probability that Algorithm \ref{algo:2} fails to select the true best executor is bounded by
\begin{equation}
\Pr\!\left(S_g \le \frac{g}{2}\right)
\le
\exp\!\left(-2g\left(p-\frac{1}{2}\right)^2\right),
\end{equation}
where $S_g$ denotes the total number of correct votes. Equivalently, the selection probability satisfies
\begin{equation}
P_{\text{select}}(g)
=
\Pr\!\left(S_g > \frac{g}{2}\right)
\ge
1 - \exp\!\left(-2g\left(p-\frac{1}{2}\right)^2\right).
\end{equation}
\end{theorem}

\begin{proof}
Let $X_\ell$ be the indicator random variable such that
\begin{equation}
X_\ell =
\begin{cases}
1, & \text{if task executor } \ell \text{ votes correctly for } e^\ast,\\
0, & \text{otherwise}.
\end{cases}
\end{equation}
Then $\mathbb{E}[X_\ell] = p$, and
\begin{equation}
\label{equ:sg}
S_g = \sum_{\ell=1}^{g} X_\ell
\end{equation}
denotes the total number of correct votes. Taking expectation both side of equation \ref{equ:sg}, we get
\begin{equation}
\label{equ:sg2}
\mathbb{E}[S_g] = \mathbb{E}\bigg[\sum_{\ell=1}^{g} X_\ell\bigg]
\end{equation}
By linearity of expectation, we have  
\begin{equation}
\label{equ:sg1}
\mathbb{E}[S_g] = \sum_{\ell=1}^{g} \mathbb{E}[X_\ell]
\end{equation}
Substituting the value of $\mathbb{E}[S_g]$ as $p$ in equation \ref{equ:sg1}, we get
\begin{equation}
\label{equ:sg3}
\mathbb{E}[S_g] = \sum_{\ell=1}^{g} p
\end{equation}
\[
\mathbb{E}[S_g] = gp.
\]
Algorithm \ref{algo:2} fails only if $S_g \le \frac{g}{2}$. Therefore,
\[
\Pr\!\left(S_g \le \frac{g}{2}\right)
=
\Pr\!\left(S_g - \mathbb{E}[S_g] \le \frac{g}{2} - gp\right).
\]
Since $p > \frac{1}{2}$, we have
\[
\frac{g}{2} - gp = -g\left(p-\frac{1}{2}\right).
\]
Thus,
\[
\Pr\!\left(S_g \le \frac{g}{2}\right)
=
\Pr\!\left(S_g - \mathbb{E}[S_g] \le -g\left(p-\frac{1}{2}\right)\right).
\]
By Hoeffding's inequality,
\[
\Pr\!\left(S_g - \mathbb{E}[S_g] \le -t\right)
\le
\exp\!\left(-\frac{2t^2}{g}\right).
\]
Setting $t = g\left(p-\frac{1}{2}\right)$ gives
\[
\Pr\!\left(S_g \le \frac{g}{2}\right)
\le
\exp\!\left(-2g\left(p-\frac{1}{2}\right)^2\right).
\]
Therefore,
\[
P_{\text{select}}(g)
=
1-\Pr\!\left(S_g \le \frac{g}{2}\right)
\ge
1-\exp\!\left(-2g\left(p-\frac{1}{2}\right)^2\right).
\]
Hence proved.
\end{proof}
\begin{observation}
The above bound shows that the failure probability of Algorithm \ref{algo:2} decreases exponentially with the number of task executors as voters. Therefore, when each evaluator is better than random guessing, increasing the number of voters significantly improves the reliability of selecting the true best executor. This provides a strong theoretical justification for using multiple evaluators in the quality-selection phase.
\end{observation}

\section{Experimental Analysis}
\label{sec:Sim}
In this section, the simulations for the three tiers of the proposed framework are conducted independently. For the first tier, the cluster formation mechanism is evaluated using the following performance metrics: (i) \emph{clustering time}, (ii) \emph{social welfare}, (iii) \emph{average intra-cluster distance}, and (iv) \emph{allocation time}. To assess the effectiveness of clustering, social welfare is analyzed under two scenarios: with clustering and without clustering. The results demonstrate that clustering consistently enhances social welfare across all system sizes by improving the matching efficiency between task requesters and executors. Similarly, the impact of clustering on computational efficiency is evaluated through allocation time. The comparison shows that clustering significantly reduces the allocation time, particularly as the number of executors increases. This improvement is attributed to the reduced search space achieved through spatial grouping, which enables faster and more efficient task assignment.\\
\indent In the second tier, the QTESM is tested using the following parameters: (i) \emph{task success rate}, (ii) \emph{number of quality task executors determined among the available task executors}, (iii) \emph{running time}, and (iv) \emph{selection probability}. In the third tier of the proposed framework, experiments are conducted to compare winner determination and payment rule with the existing mechanisms MUDA, PPM, and McAfee, in terms of truthfulness as the first metric. It will help us show that the winner-determination and payment rule is truthful relative to existing benchmark mechanisms. The second metric that is considered for comparing the winner determination and payment rule with the respective existing mechanisms is the total payment made to the task executors. The third metric that is considered for comparing winner determination and payment rule with respective existing mechanisms is the running time of the mechanisms. It will help us to show that the winner determination and payment rule is scalable.

\subsection{Simulation Set-up for first-tier}
\label{sec:Sima}
To evaluate the performance of the cluster formation, simulations are conducted in a two-dimensional spatial environment representing the geographical area of a spatial crowdsourcing system. 
\begin{table}[H]
\centering
\caption{Simulation Parameters and Experimental Environment}
\label{tab:simulation_parameters}
\renewcommand{\arraystretch}{1.2}
\setlength{\tabcolsep}{4pt}
\begin{tabular}{p{1.5cm} p{4.2cm} p{3.5cm}}
\hline
\textbf{Symbol} & \textbf{Parameter} & \textbf{Configuration} \\
\hline
$A$ & Simulation region & $100 \times 100$ spatial grid \\
$m$ & Task executors & 100--1600 \\
$n$ & Task requesters & 50--300 \\
$k$ & Number of clusters & 2--20 \\
$v_i$ & Requester valuation & $[8,30]$ \\
$c_j$ & Executor cost & $[5,25]$ \\
$R$ & Simulation repetitions & 30 runs \\
Dist & Distance function & Euclidean \\
Init & Initialization & Random centroids \\
\hline
\multicolumn{3}{c}{\textbf{Hardware and Software Configuration}} \\
\hline
-- & Processor & Intel Core i5 \\
-- & RAM & 16 GB \\
-- & Operating System & Windows \\
-- & Programming Language & Python \\
-- & Libraries/Tools & NumPy\\
-- & Plotting Tools & Gnuplot \\
\hline
\end{tabular}
\end{table}
Task executors are randomly distributed over the simulation region, and the clustering algorithm groups them based on their spatial proximity. The simulation area is set to $100 \times 100$ units. The number of task executors $m$ is varied from 100 to 1600 to analyze scalability. The number of clusters $k$ is varied from 2 to 30 to study its effect on cluster compactness and allocation efficiency. The initial cluster centers are selected randomly, and the algorithm iterates until cluster assignments stabilize or a maximum number of iterations is reached. To study the impact of clustering on the overall system, simulations are also carried out for the task allocation mechanism with and without clustering. The corresponding system-level metrics include social welfare and task allocation time. Each experiment is repeated 30 times with different random seeds, and the average results are reported.
\subsubsection{Simulation metrics for first tier}
\begin{itemize}
\item \textbf{Clustering time:} The time taken to determine the k-clusters from the given set of points. 
\item \textbf{Social welfare:} The sum of the valuations of the task executors.
\item \textbf{Average intra-cluster distance:} Intra-cluster distance measures the compactness of clusters by quantifying the average distance between data points within the same cluster. Lower intra-cluster distance indicates that the points are more tightly grouped, leading to more efficient spatial task allocation and reduced communication overhead. The intra-cluster distance of a cluster $C$ is defined as
\[
\text{IntraDist}(C) = \frac{1}{|C|} \sum_{x_i \in C} d(x_i, \mu_C),
\]
where $\mu_C$ denotes the centroid of cluster $C$, and $d(\cdot)$ represents a distance metric (e.g., Euclidean distance).

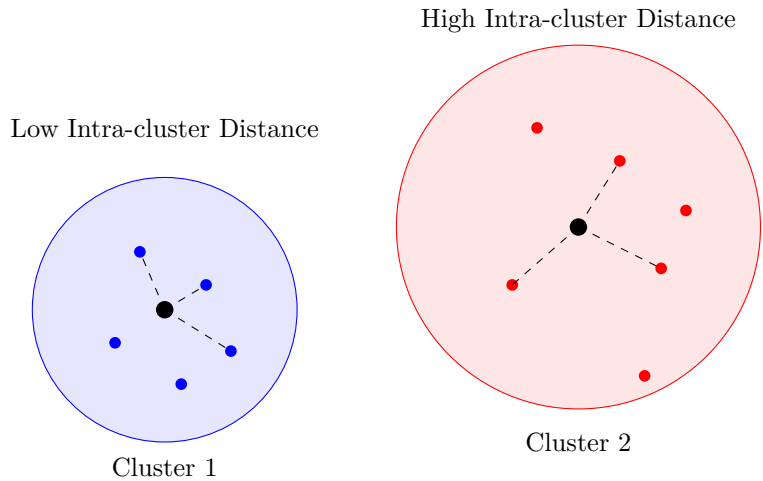
\begin{figure}[H]
\centering
\begin{tikzpicture}[scale=1.1]

% --- Cluster 1 ---
\draw[fill=blue!10, draw=blue] (0,0) circle (1.6);
\foreach \x/\y in {0.5/0.3, -0.3/0.7, 0.8/-0.5, -0.6/-0.4, 0.2/-0.9}
    \fill[blue] (\x,\y) circle (2pt);

% Centroid
\fill[black] (0,0) circle (3pt);
\node at (0,-1.9) {\small Cluster 1};

% Distance lines
\draw[dashed] (0,0) -- (0.5,0.3);
\draw[dashed] (0,0) -- (-0.3,0.7);
\draw[dashed] (0,0) -- (0.8,-0.5);

% --- Cluster 2 ---
\draw[fill=red!10, draw=red] (5,1) circle (2.2);
\foreach \x/\y in {5.5/1.8, 4.2/0.3, 6.0/0.5, 5.8/-0.8, 4.5/2.2, 6.3/1.2}
    \fill[red] (\x,\y) circle (2pt);

% Centroid
\fill[black] (5,1) circle (3pt);
\node at (5,-1.6) {\small Cluster 2};

% Distance lines
\draw[dashed] (5,1) -- (5.5,1.8);
\draw[dashed] (5,1) -- (4.2,0.3);
\draw[dashed] (5,1) -- (6.0,0.5);

% Labels
\node at (0,2.2) {\small Low Intra-cluster Distance};
\node at (5,3.5) {\small High Intra-cluster Distance};

\end{tikzpicture}
\caption{Graphical intuition of intra-cluster distance. Cluster 1 exhibits low intra-cluster distance with tightly grouped points, whereas Cluster 2 shows higher intra-cluster distance due to more dispersed points.}
\label{fig:intra_cluster_intuition}
\end{figure}
\item \textbf{Allocation time:} The time taken to allocate the tasks to the task executors.
\end{itemize}
\subsection{Simulation Set-up for second tier}
\label{sec:Simb}
To empirically validate the effect of the number of voters on the reliability of the quality-selection mechanism, we conduct a Monte Carlo simulation. The goal is to estimate the probability of correctly selecting the true best executor under varying numbers of voters. In each simulation round, a set of $f$ candidate executors is considered, among which one executor is designated as the true best executor $e^\ast$. Each voter independently evaluates the candidates and ranks $e^\ast$ above any other executor with probability $p > 0.5$, representing the accuracy of the evaluator. For a given number of voters $g$, the selection of the winning executor is determined using a quality determination mechanism. The experiment is repeated for $R = 10{,}000$ independent trials, and the probability of correctly selecting the true best executor is estimated as the fraction of trials in which $e^\ast$ is selected. The key simulation parameters are set as given in Table \ref{tab:quality_selection}. The resulting selection probability is then analyzed as a function of $g$.

\begin{table}[H]
\centering
\caption{Simulation Parameters for Quality Selection Analysis}
\label{tab:quality_selection}
\renewcommand{\arraystretch}{1.1}
\begin{tabular}{lll}
\hline
\textbf{Symbol} & \textbf{Parameter} & \textbf{Value} \\
\hline
$f$ & Number of candidate executors & 4 \\
$e^\ast$ & True best executor & Fixed per trial \\
$p$ & Voter accuracy & 0.7 \\
$g$ & Number of voters & 3--15 \\
$R$ & Number of simulation runs & 10,000 \\
Rule & Selection method & Quality Task Executors Determination \\
\hline
\end{tabular}
\end{table}
\begin{table}[H]
\centering
\caption{Empirical probability of selecting the true best executor under different numbers of voters}
\label{tab:true_best_probability}
\renewcommand{\arraystretch}{1.1}
\begin{tabular}{cc}
\hline
\textbf{Number of voters ($g$)} & \textbf{Selection probability} \\
\hline
3  & 0.71 \\
5  & 0.79 \\
7  & 0.85 \\
9  & 0.89 \\
11 & 0.92 \\
13 & 0.94 \\
15 & 0.96 \\
\hline
\end{tabular}
\end{table}

\subsubsection{Simulation metrics for second tier}
\begin{itemize}
\item \textbf{Task Success Rate (TSR):} \[
\text{Task Success Rate (TSR)} = \frac{\text{Number of successfully completed tasks}}{\text{Total number of assigned tasks}} \times 100
\]
\item \textbf{Running time:} The total time taken to determine the quality task executors using Algorithm \ref{algo:2}. 
\item \textbf{Selection probability:} It refers to the probability that the mechanism correctly selects the true best executor from a set of candidates.
\end{itemize}

\subsection{Simulation Set-up for third tier}
To evaluate the scalability of the proposed mechanism, we vary the number of task requesters (TRs) in the range $[6000, 45000]$ and the number of task executors (TEs) in the range $[10000, 50000]$. This setup enables us to analyze the performance under large-scale spatial crowdsourcing environments. To study the impact of clustering granularity, we consider three different configurations: (i) 80 clusters, (ii) 100 clusters, and (iii) 120 clusters. In each configuration, the number of clusters is kept fixed across all simulation iterations. This allows us to examine how different cluster sizes influence system efficiency and computational overhead. We compare the proposed TRUST-SC mechanism with baseline mechanisms, namely McAfee, MUDA, and PPM, with respect to key performance metrics including truthfulness, total payment to task executors, and runtime of the agents. This comparative analysis provides insights into the economic efficiency and computational scalability of the proposed approach. The results demonstrate that TRUST-SC maintains truthful behavior while achieving competitive payment efficiency and significantly lower runtime compared to the baseline mechanisms, particularly as the system size increases.

\subsection{Performance analysis for first tier}
Fig.~\ref{sim:6aa} illustrates the scalability of the clustering phase with respect to the number of task executors. As the number of executors increases, the clustering time grows in a near-linear fashion. This behavior is expected since the algorithm processes spatial information and performs distance computations for grouping executors into clusters. Even if the system size is increasing rapidly, the growth in runtime remains moderate, indicating that the clustering algorithm is computationally efficient and suitable for large-scale spatial crowdsourcing environments. The results demonstrate that the cluster formation approach can handle a substantial number of executors without incurring significant computational overhead, thereby enabling its practical deployment in real-world applications.
\begin{figure}[H]
\begin{subfigure}{0.52\textwidth}
\includegraphics[scale = 0.62]{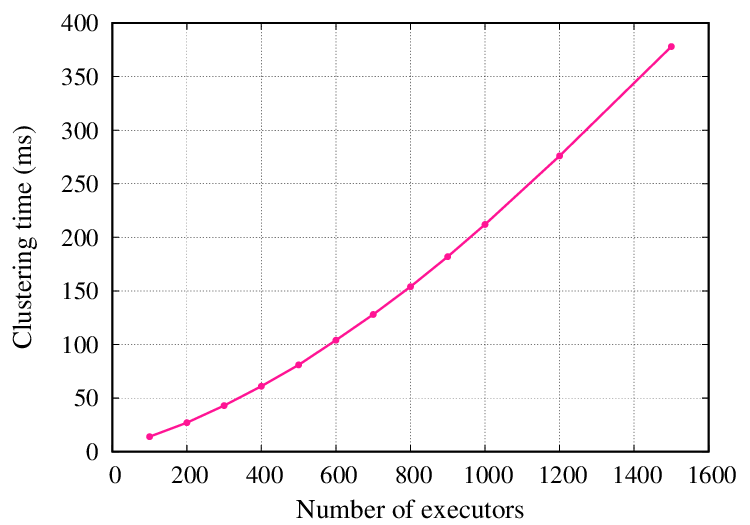}
\subcaption{Growth of clustering time with increasing number of task executors.}
\label{sim:6aa}
\end{subfigure}%
\begin{subfigure}{0.52\textwidth}
\includegraphics[scale = 0.62]{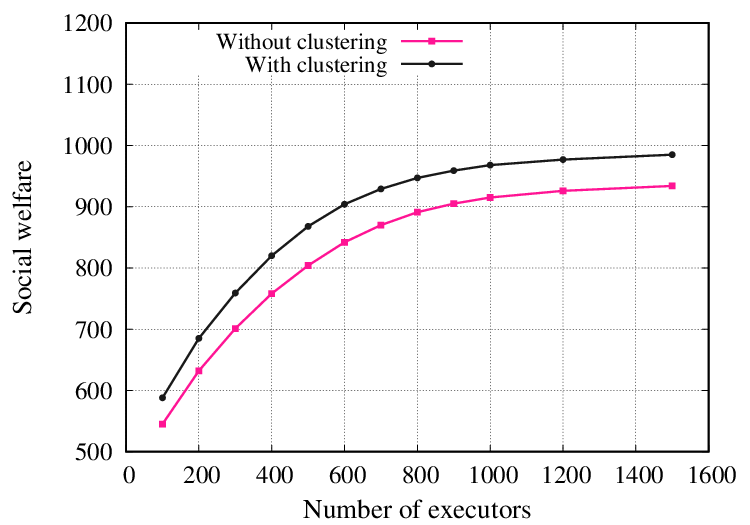}
\subcaption{Comparison of social welfare with the number of executors.}
\label{sim:6b}
\end{subfigure}
\begin{subfigure}{0.52\textwidth}
\includegraphics[scale = 0.62]{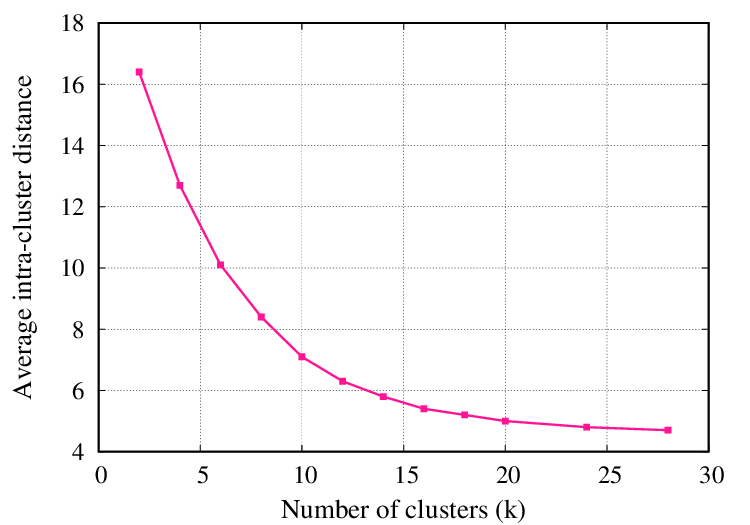}
\subcaption{Comparing average intra-cluster distance with number of clusters.}
\label{sim:6c}
\end{subfigure}%
\begin{subfigure}{0.52\textwidth}
\includegraphics[scale = 0.62]{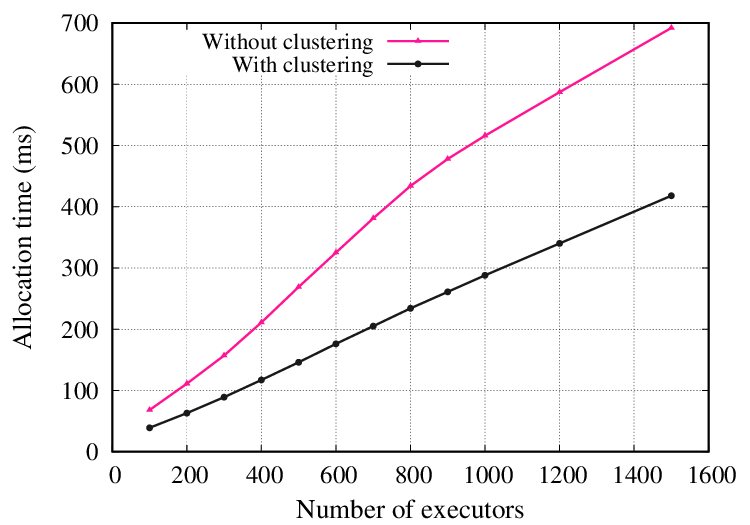}
\subcaption{Comparison of allocation time with number of executors.}
\label{sim:6d}
\end{subfigure}
\caption{Cluster formation mechanism evaluation: (a) clustering time, (b) social welfare with and without clustering, (c) average intra-cluster distance, and (d) task allocation time with and without clustering.}
\label{sim:6}
\end{figure}
\indent Fig.~\ref{sim:6b} compares the social welfare achieved with and without clustering. The results show that clustering consistently improves social welfare across all system sizes. This improvement arises because clustering restricts the matching process to spatially proximate executors, thereby reducing mismatches and improving allocation efficiency. The clustered approach enables better utilization of executor capabilities, leading to higher overall system efficiency.\\
\indent Fig.~\ref{sim:6c} shows the variation of the average intra-cluster distance with respect to the number of clusters. As the number of clusters increases, the intra-cluster distance decreases significantly, indicating improved spatial compactness of the clusters. However, the rate of improvement diminishes beyond a certain threshold, suggesting a trade-off between cluster granularity and computational efficiency. This observation highlights the importance of selecting an appropriate number of clusters to balance compactness and overhead.\\
\indent Fig.~\ref{sim:6d} presents the task allocation time under both clustered and non-clustered settings. It is evident that clustering significantly reduces the allocation time, particularly as the number of executors increases. This is because clustering reduces the effective search space from the entire executor set to smaller localized groups. Consequently, the allocation algorithm operates more efficiently, making the system scalable and responsive.

\subsection{Performance analysis for second tier}

\begin{figure}[H]
\begin{subfigure}{0.52\textwidth}
\includegraphics[scale = 0.62]{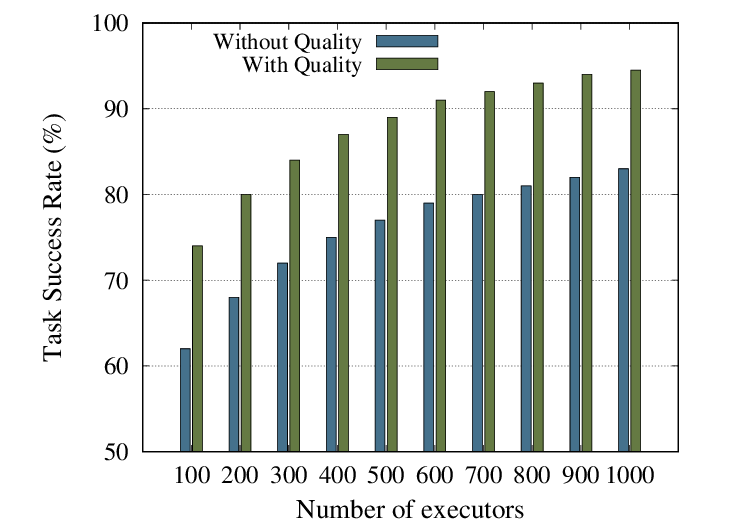}
\subcaption{Comparison of TSR with and without quality TEs}
\label{sim:7a}
\end{subfigure}%
\begin{subfigure}{0.52\textwidth}
\includegraphics[scale = 0.62]{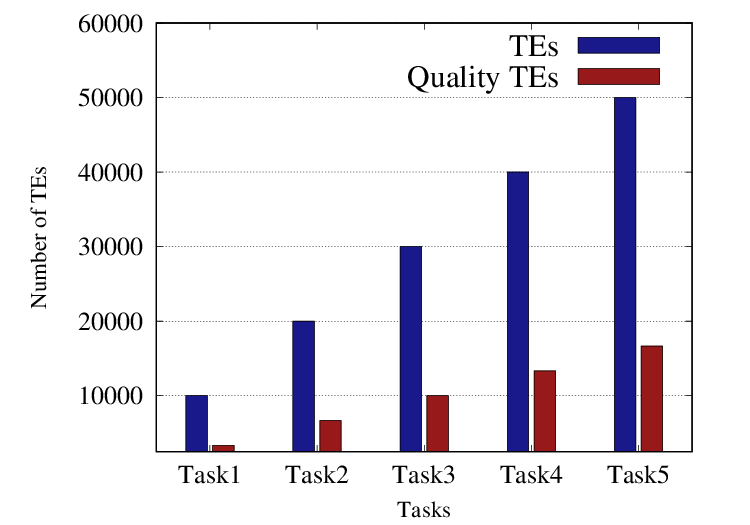}
\subcaption{Comparison of Number of TEs and Number of quality TEs with tasks}
\label{sim:7b}
\end{subfigure}
\begin{subfigure}{0.52\textwidth}
\includegraphics[scale = 0.62]{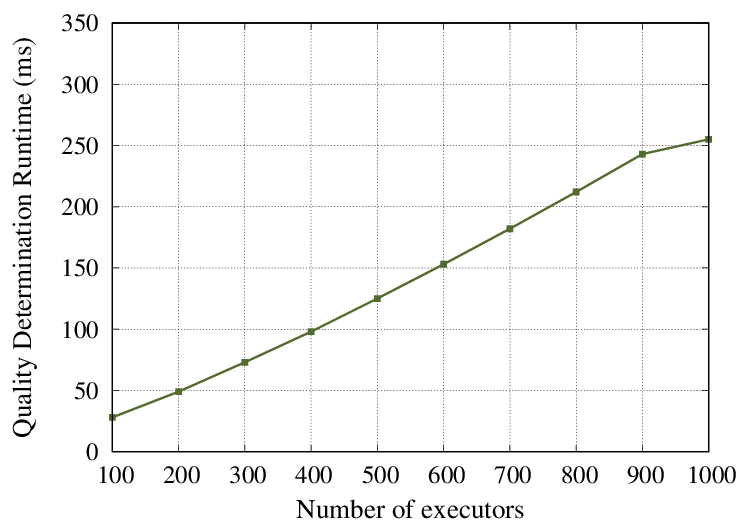}
\subcaption{Growth of running time of QTESM with increase in number of TEs}
\label{sim:7c}
\end{subfigure}%
\begin{subfigure}{0.52\textwidth}
\includegraphics[scale = 0.62]{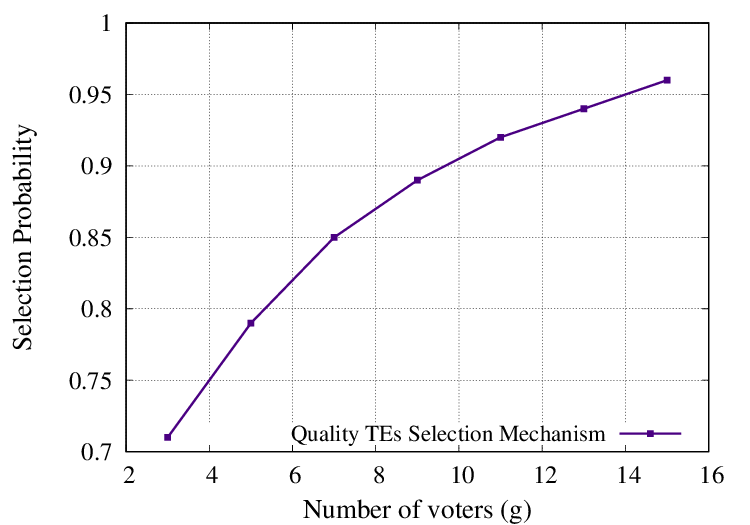}
\subcaption{Growth of selection probability with increase in number of voters}
\label{sim:7d}
\end{subfigure}
\caption{QTESM evaluation: (a) TSR, (b) Number of quality TEs, (3) running time, and (d) selection probability.}
\label{fig:quality_selection}
\end{figure}
Fig.~\ref{fig:quality_selection} illustrates the performance of the proposed Quality-based Task Executor Selection Mechanism (QTESM) under varying system parameters. In Fig. \ref{sim:7a}, the effect of quality task executors is shown on the success rate of the tasks. It is seen that when the tasks are carried out by the quality TEs, the TSR is higher compared to the case when the tasks are done by the TEs (a combination of both quality TEs and low-quality TEs). This behavior can be seen irrespective of the number of TEs participating in the market. In Fig. \ref{sim:7b}, the y-axis represents the number of TEs, and the x-axis represents the tasks. It compares the QTESM to determine how many TEs from among the available TEs are quality TEs. Fig. \ref{sim:7c} determines the running time of QTESM with respect to the number of TEs. As the number of TEs increases, the running time of QTESm also increases. As shown in Fig.~\ref{sim:7d}, the performance of QTESM is evaluated in terms of selection accuracy (or selection probability) with respect to the number of evaluators. It is observed that the probability of correctly selecting the true best executor increases monotonically as the number of evaluators increases. This behavior is consistent with the theoretical analysis, where aggregating multiple independent evaluations reduces the likelihood of incorrect selection. In particular, when the number of evaluators is small, the selection probability is relatively lower due to higher uncertainty in individual assessments. However, as the number of evaluators increases, the collective decision-making process becomes more reliable, leading to a significant improvement in selection accuracy. This trend demonstrates the robustness of QTESM against individual evaluation errors. Furthermore, the results in Fig.~\ref{fig:quality_selection} validate the probabilistic guarantees derived using Hoeffding’s inequality. The exponential reduction in error probability ensures that the mechanism can achieve near-perfect selection accuracy with a moderate number of evaluators. Overall, the results indicate that QTESM effectively enhances the reliability of task executor selection, which directly contributes to improved task success rates and overall system performance in the proposed TRUST-SC framework.

\subsection{Performance analysis for third tier}
\label{sec:RA}
The performance of TRUST-SC is compared with the existing benchmark mechanisms, namely McAfee double auction (DA) \cite{bredin2012models}, Truthful Multi-Unit Double-Auction Mechanism (MUDA) \cite{SegalHalevi2018MUDAAT}, and Posted Price Mechanism (PPM) \cite{chiu2003posted}. In our experimental section, we have considered cluster values of $80$, $100$, and $120$. The simulation results presented in Fig. \ref{Sim:1A} illustrate the performance of TRUST-SC with benchmark mechanisms on the grounds of the utility of TRs. In this, the x-axis of the graphs represents the number of TRs, and the y-axis represents the utility of TRs. Here, the utility of TRs is the sum of the utility of winning TRs. In this case, when the cluster number is $80$, then it can be seen in Fig. \ref{sim:1a} that the utility of winning TRs in TRUST-Sc is more than the utility of TRs in cases of McAfee and MUDA. It is due to the reason that in the case of TRUST-SC, the payment made to the winning task Requesters is higher than the payment made to the winning task requesters in the case of McAfee and MUDA.
\begin{figure}[H]
\begin{subfigure}{0.35\textwidth}
\includegraphics[scale = 0.16]{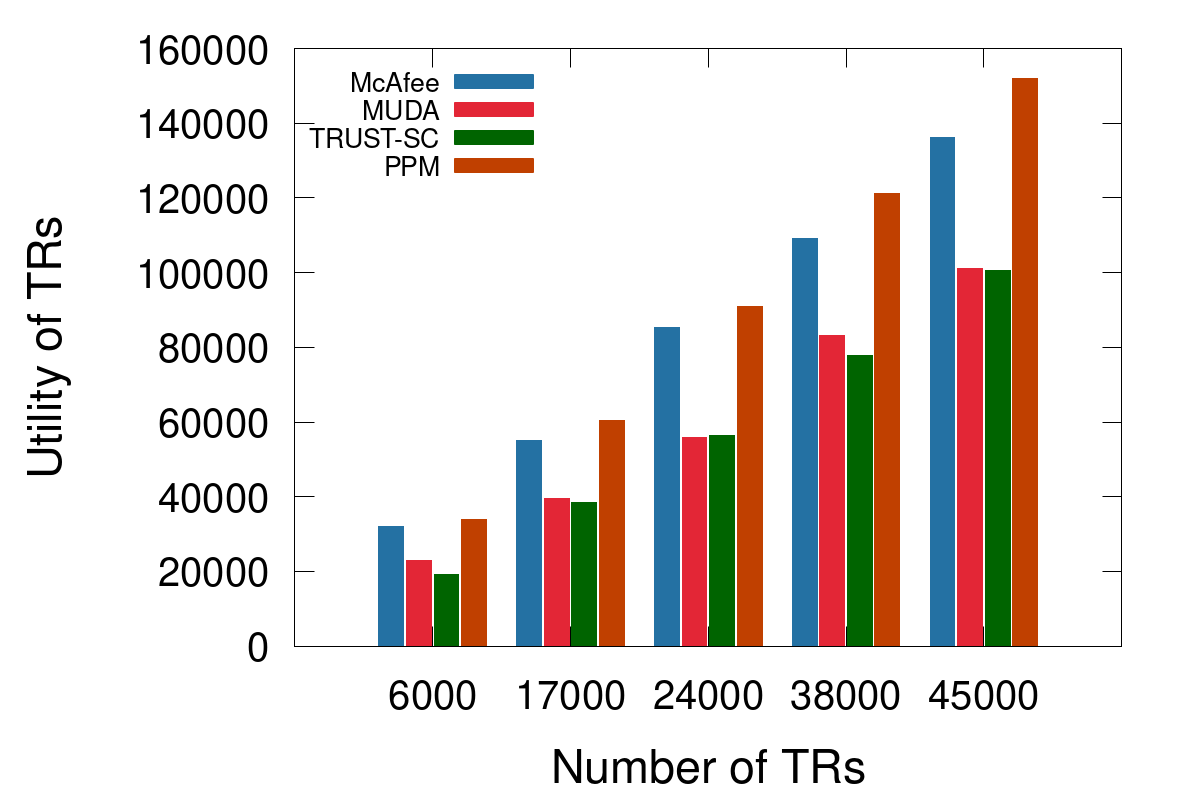}
\subcaption{Utility of TRs with $k = 80$}
\label{sim:1a}
\end{subfigure}%
\begin{subfigure}{0.35\textwidth}
\includegraphics[scale = 0.16]{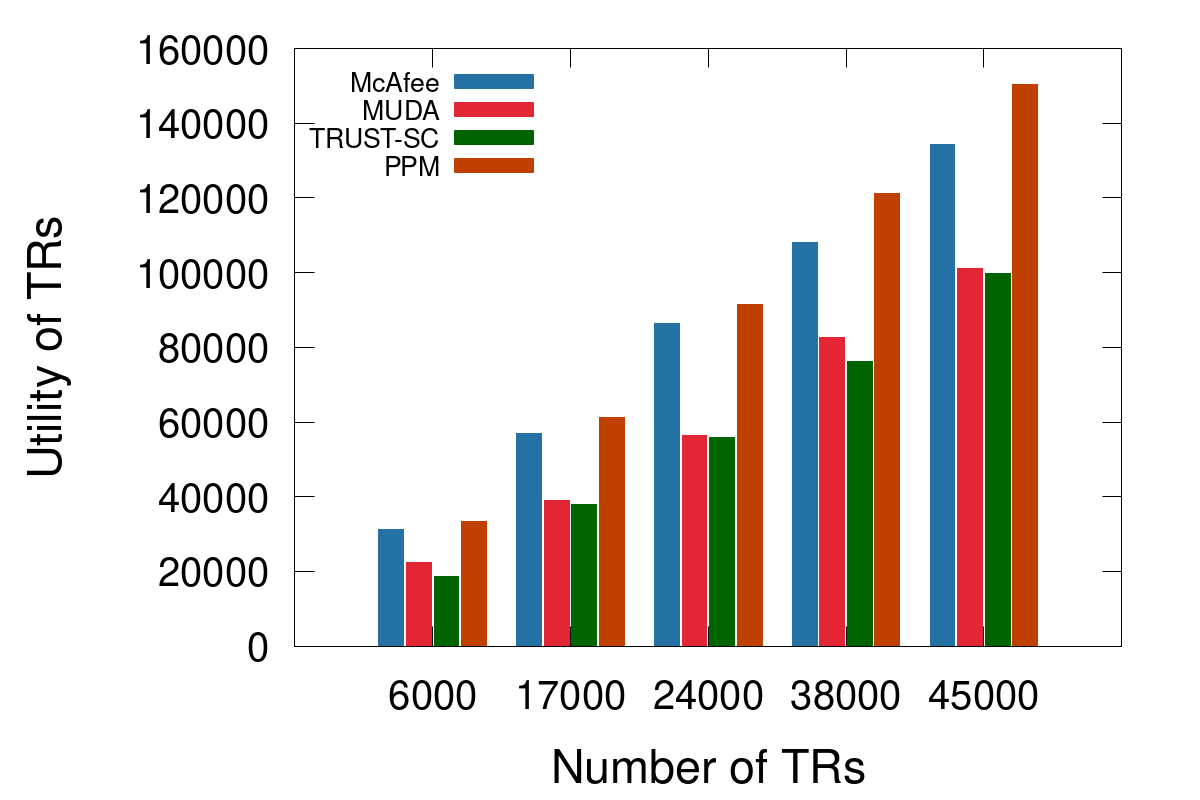}
\subcaption{Utility of TRs with $k = 100$}
\label{sim:11a}
\end{subfigure}%
\begin{subfigure}{0.35\textwidth}
\includegraphics[scale = 0.16]{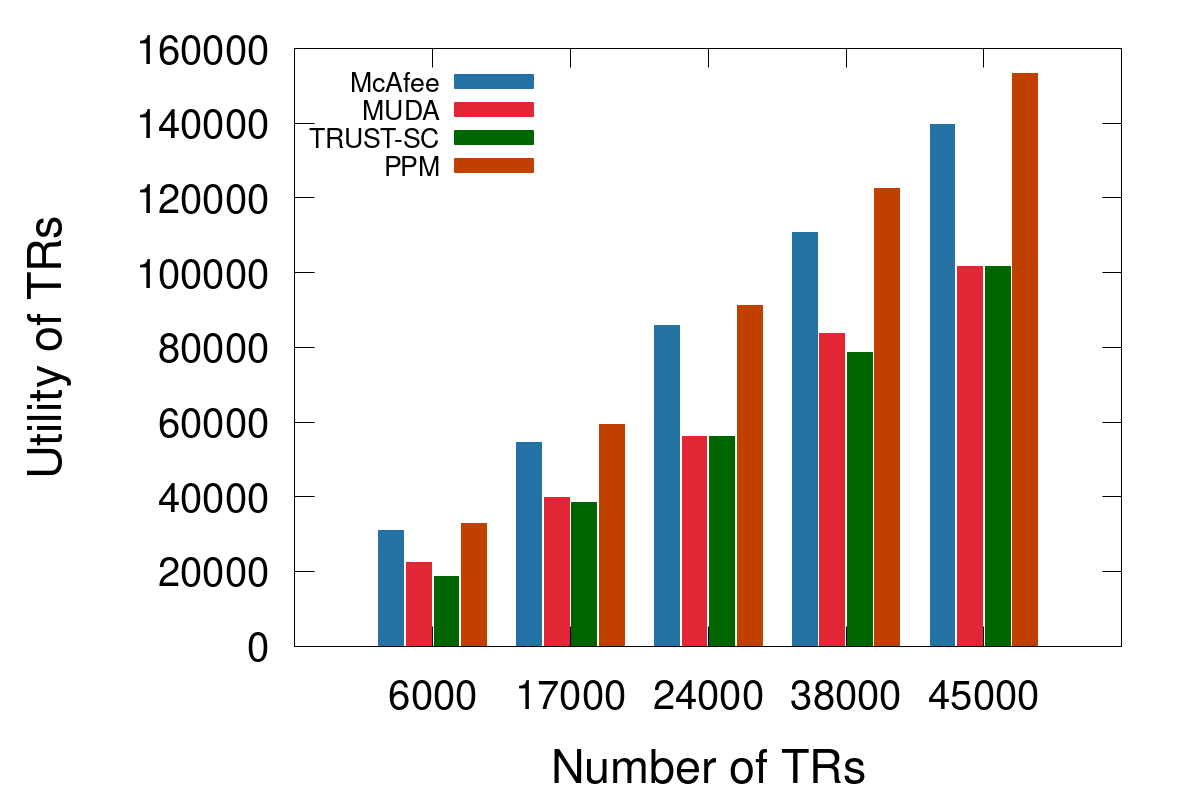}
\subcaption{Utility of TRs with $k = 120$}
\label{sim:111a}
\end{subfigure}%
\caption{Comparison of Utility of TRs with $k = 80, 100$, and $120$.}
\label{Sim:1A}
\end{figure}
On the other hand, when compared with PPM, the utility of TRs in the case of PPM is higher than the utility of TRs in the case of TRUST-SC. It is due to the reason that the payment made to the task requesters in case of PPM is more than the payment made to the task requesters in case of TRUST-SC. Hence, the utility is more in the case of PPM, as shown in Fig. \ref{sim:1a}. Considering the cases with $k-100$ and $120$, similar results are obtained in Fig. \ref{sim:11a} and Fig. \ref{sim:111a} respectively.\\
%\indent Similarly, the same behavior is observed in Fig. \ref{Sim:2A}, which illustrates the performance comparison of TRUST-SC with benchmark mechanisms on the grounds of the utility of TEs with cluster values of $80$, $100$, and $120$. The x-axis of the graphs represents the number of TEs, and the y-axis represents the utility of TEs. From Figs. \ref{sim:2a}, \ref{sim:22a}, and \ref{sim:222a}.
\indent The simulation results presented in Fig. \ref{Sim:2A} illustrate the performance of TRUST-SC with benchmark mechanisms on the grounds of the utility of TEs. In this, the x-axis of the graphs represents the number of TEs, and the y-axis represents the utility of TEs. Here, the utility of TEs is the sum of the utility of winning TEs. In this case, when the cluster number is $80$, then it can be seen in Fig. \ref{sim:2a} that the utility of winning TEs in TRUST-SC is more than the utility of TEs in cases of McAfee and MUDA. It is due to the reason that in the case of TRUST-SC, the payment made to the winning task executors is higher than the payment made to the winning task executors in the case of McAfee and MUDA. On the other hand, when compared with PPM, the utility of TEs in the case of PPM is higher than the utility of TEs in the case of TRUST-SC. It is due to the reason that the payment made to the task executors in case of PPM is more than the payment made to the task executors in case of TRUST-SC. Hence, the utility is more in the case of PPM, as shown in Fig. \ref{sim:2a}. Considering the cases with $k-100$ and $120$, similar results are obtained in Fig. \ref{sim:22a} and Fig. \ref{sim:222a}, respectively.
\begin{figure}[H]
\begin{subfigure}{0.35\textwidth}
\includegraphics[scale = 0.16]{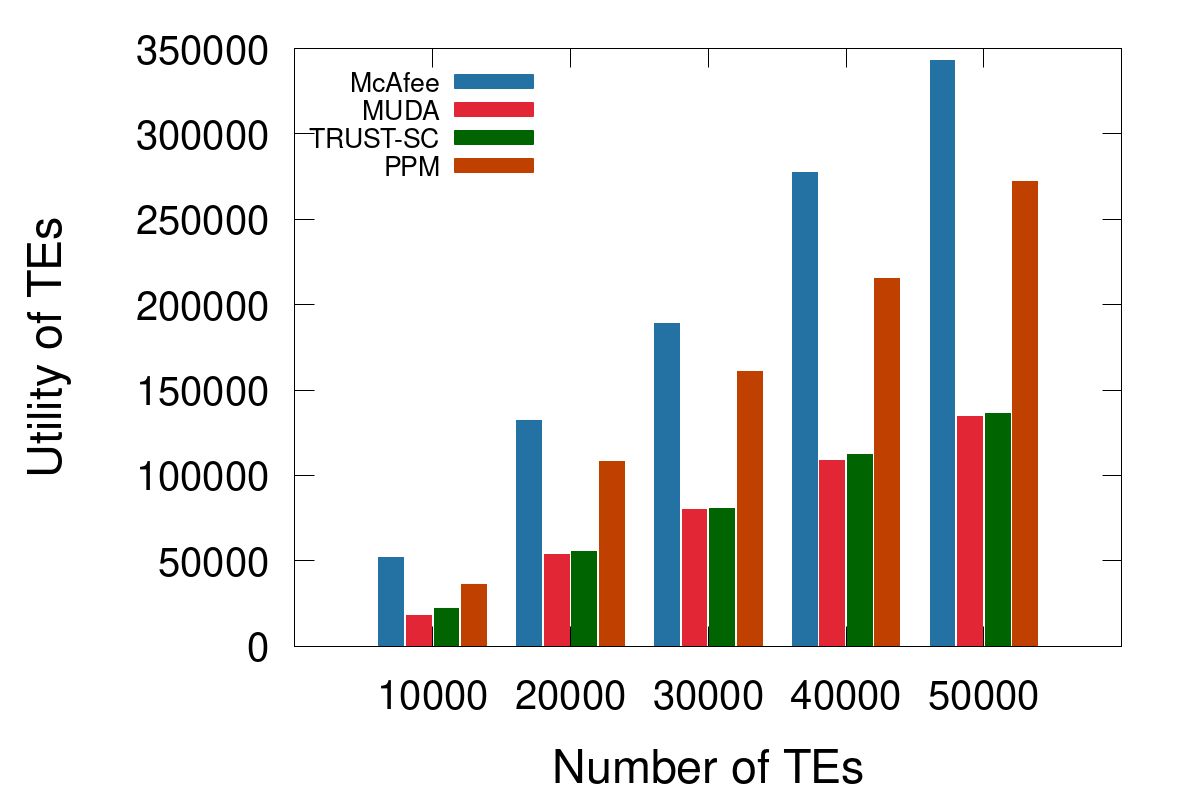}
\subcaption{Utility of TEs with $k = 80$}
\label{sim:2a}
\end{subfigure}%
\begin{subfigure}{0.35\textwidth}
\includegraphics[scale = 0.16]{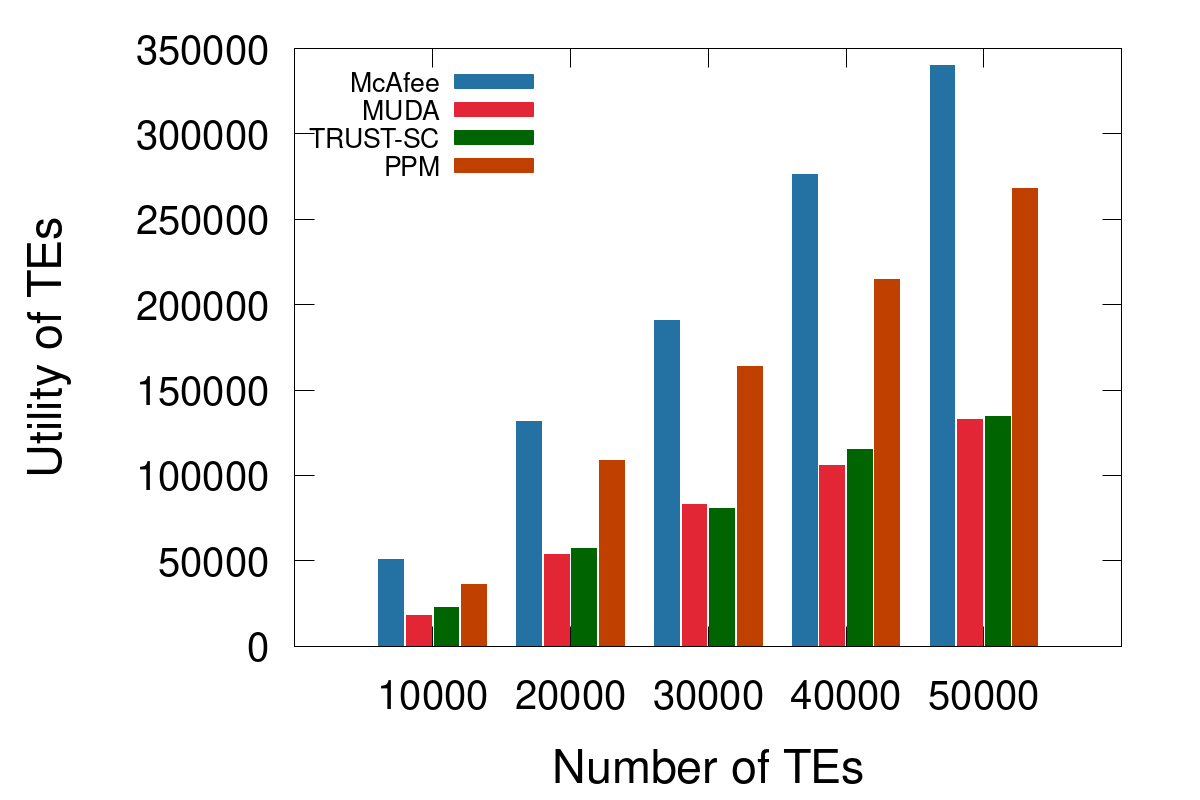}
\subcaption{Utility of TEs with $k = 100$}
\label{sim:22a}
\end{subfigure}%
\begin{subfigure}{0.35\textwidth}
\includegraphics[scale = 0.16]{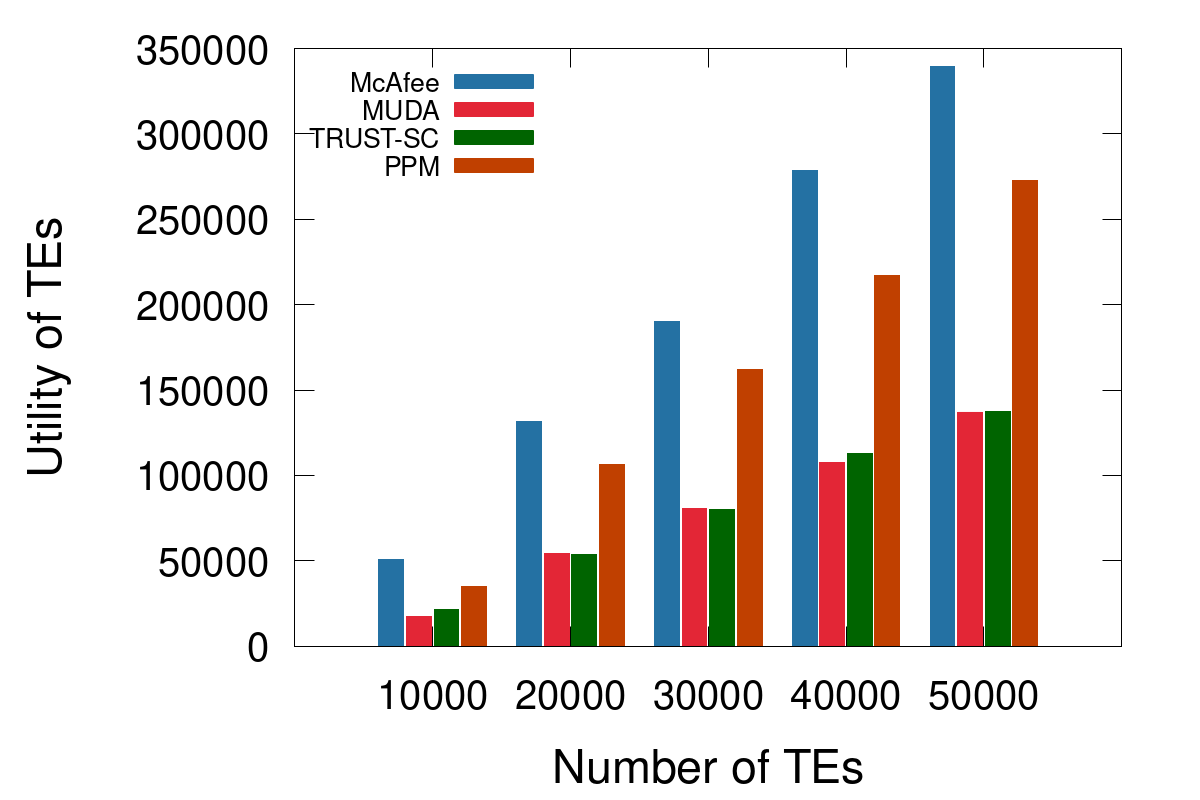}
\subcaption{Utility of TRs with $k = 120$}
\label{sim:222a}
\end{subfigure}%
\caption{Comparison of Utility of TEs with $k = 80, 100$, and $120$.}
\label{Sim:2A}
\end{figure}

 Fig. \ref{Sim:3A} compares TRUST-SC, McAfee, MUDA, and PPM in terms of total payment made to the TEs for cluster values of $80$, $100$, and $120$. The x-axis of the graphs represents the task executors (TEs), and the y-axis represents the total payment. From graphs, it is evident that the total payment to TEs increases steadily with the growth in the number of executors. In this case, when the cluster number is $80$, then it can be seen that in Fig. \ref{sim:3a}, it can be seen that McAfee and PPM yield the highest payments compared to the case of MUDA and TRUST-SC. 
 \begin{figure}[H]
\begin{subfigure}{0.35\textwidth}
\includegraphics[scale = 0.30]{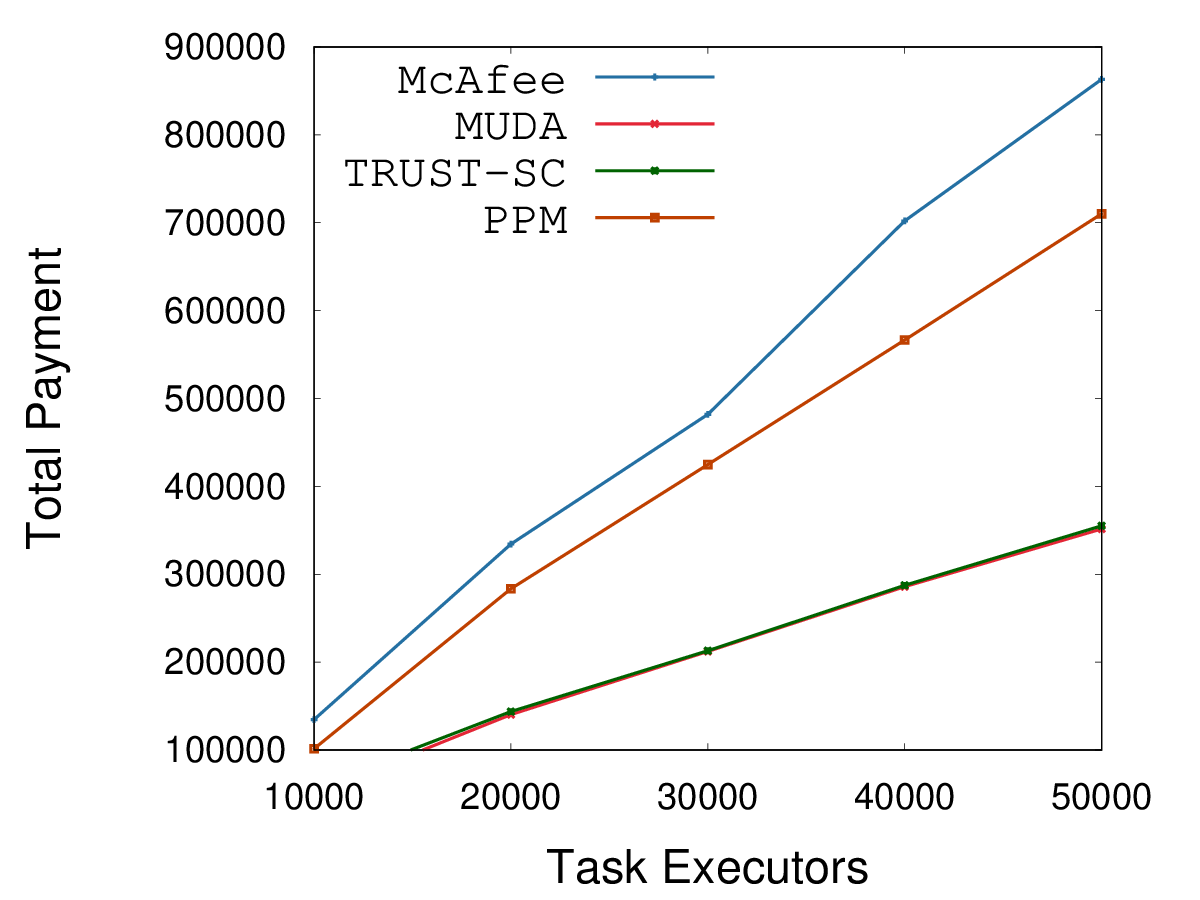}
\subcaption{Payment of TEs with $k = 80$}
\label{sim:3a}
\end{subfigure}%
\begin{subfigure}{0.35\textwidth}
\includegraphics[scale = 0.30]{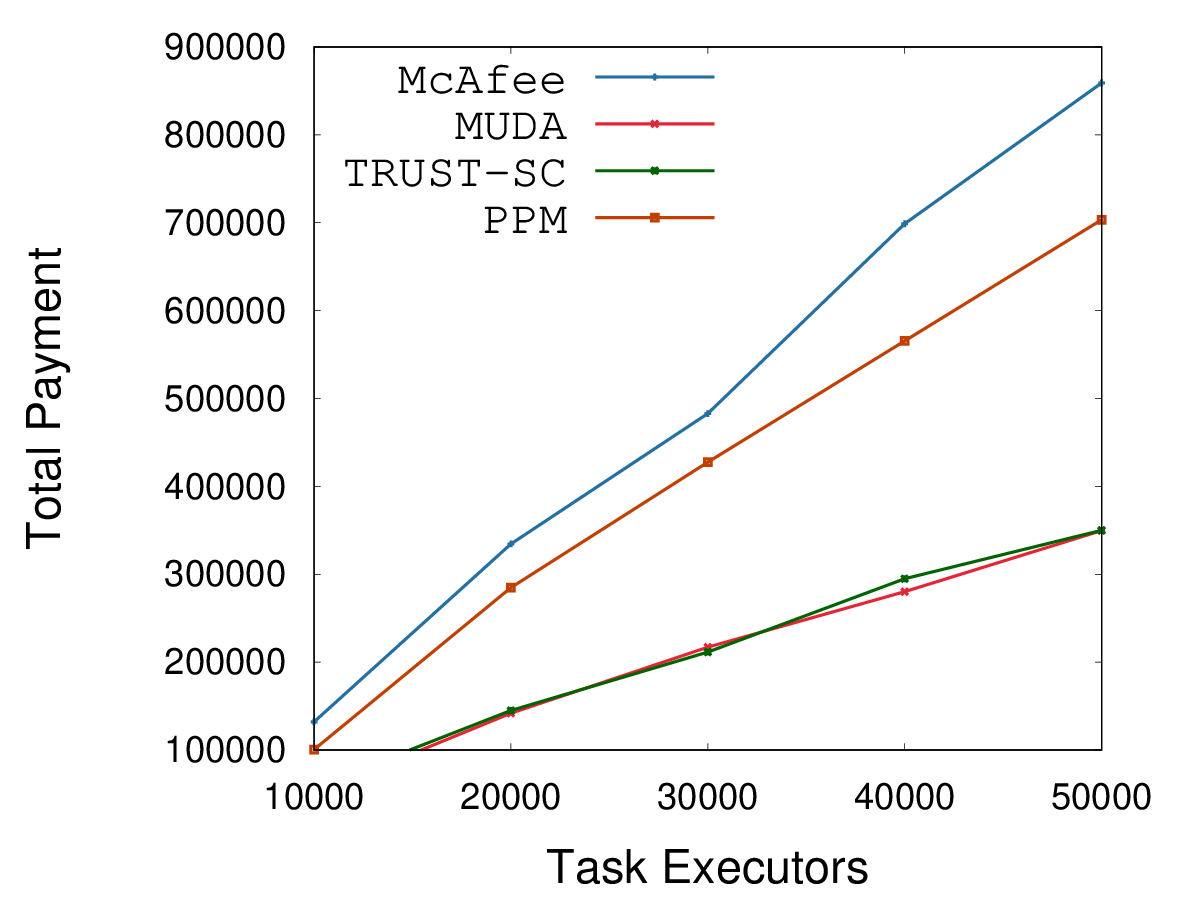}
\subcaption{Payment of TEs with $k = 100$}
\label{sim:33a}
\end{subfigure}%
\begin{subfigure}{0.35\textwidth}
\includegraphics[scale = 0.30]{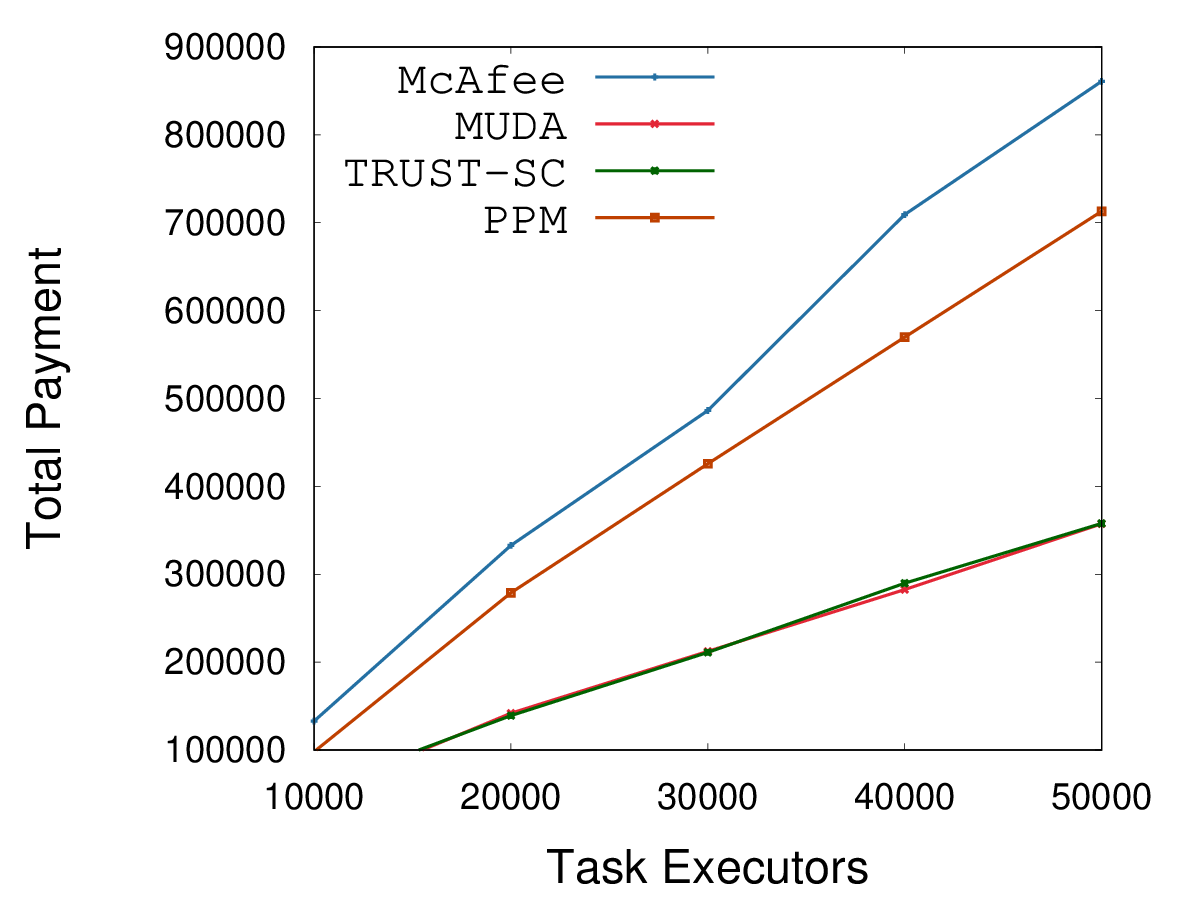}
\subcaption{Payment of TEs with $k = 120$}
\label{sim:333a}
\end{subfigure}%
\caption{Comparison of Payment of TEs with $k = 80, 100$, and $120$.}
\label{Sim:3A}
\end{figure}
 The higher payments in McAfee and PPM can be attributed to their unified market structures without splitting, allowing broader competition and more extensive matching. As a result, more executors participate successfully, leading to greater cumulative payments. MUDA performs slightly below McAfee and PPM, as its pricing and allocation structure, although efficient, does not always maximize total payments to the TEs. TRUST-SC, on the other hand, produces noticeably lower total payments because it operates with market splitting and queue-based trade execution for heterogeneous tasks. Since trades are executed separately within sub-markets, some potential cross-market matches are not realized, which reduces the overall payment distributed to TEs. Similarly, the case with considering $100$ and $120$ number of clusters as shown in Figs. \ref{sim:33a} and \ref{sim:333a}.\\
\indent From Fig. \ref{Sim:4} illustrates a comparison of Tasks Executed by TEs. The x-axis of the graphs represents the task requesters (TRs), and the y-axis represents the number of tasks. The two curves correspond to the total tasks submitted by the requesters and the tasks that were successfully executed. The simulation results in Figs. \ref{sim:4a} and \ref{sim:44a} show that the number of executed tasks closely follows the trend of the total tasks, although it remains consistently lower. Overall, a significant proportion of the submitted tasks is executed across all requesters. 

\begin{figure}[H]
\begin{subfigure}{0.52\textwidth}
\includegraphics[scale = 0.45]{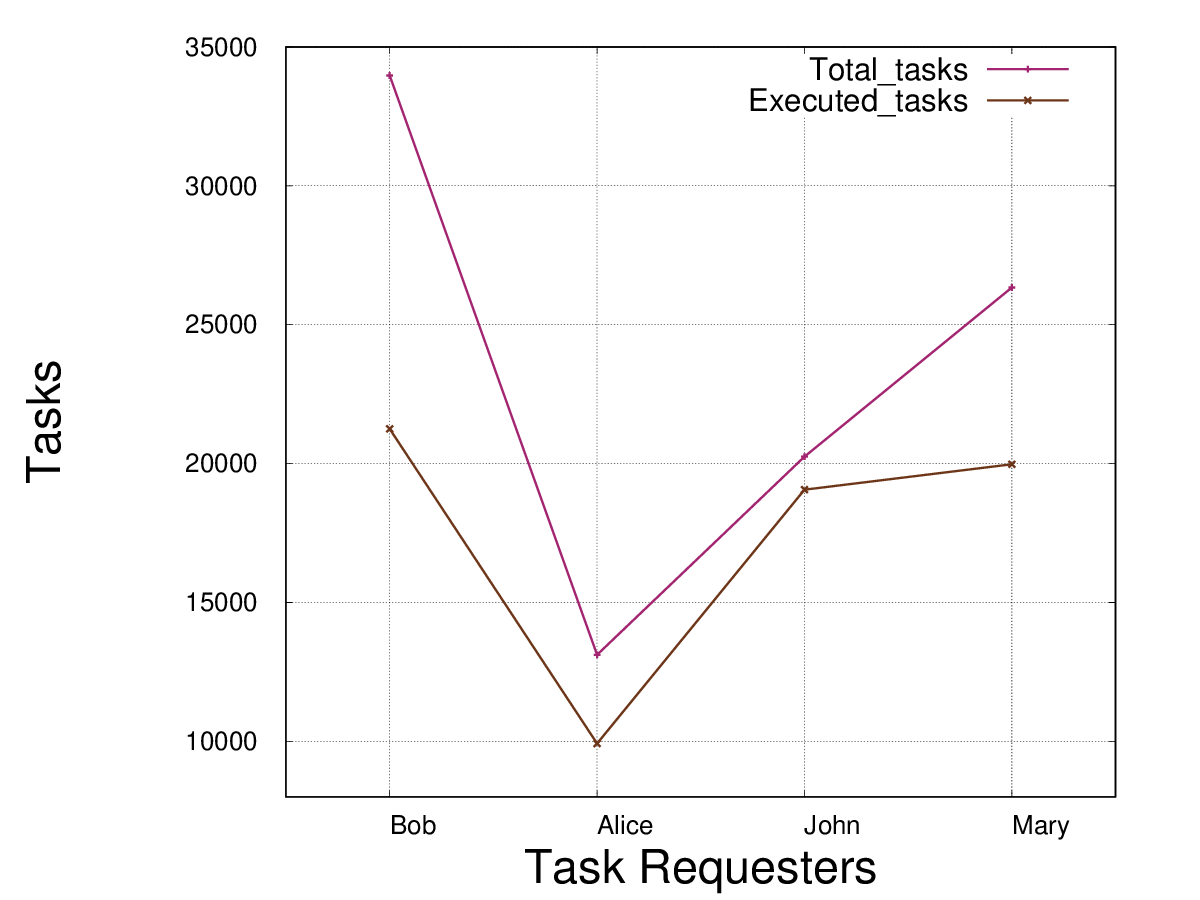}
\subcaption{Tasks Executed by TEs with a large no.of tasks}
\label{sim:4a}
\end{subfigure}%
\begin{subfigure}{0.52\textwidth}
\includegraphics[scale = 0.45]{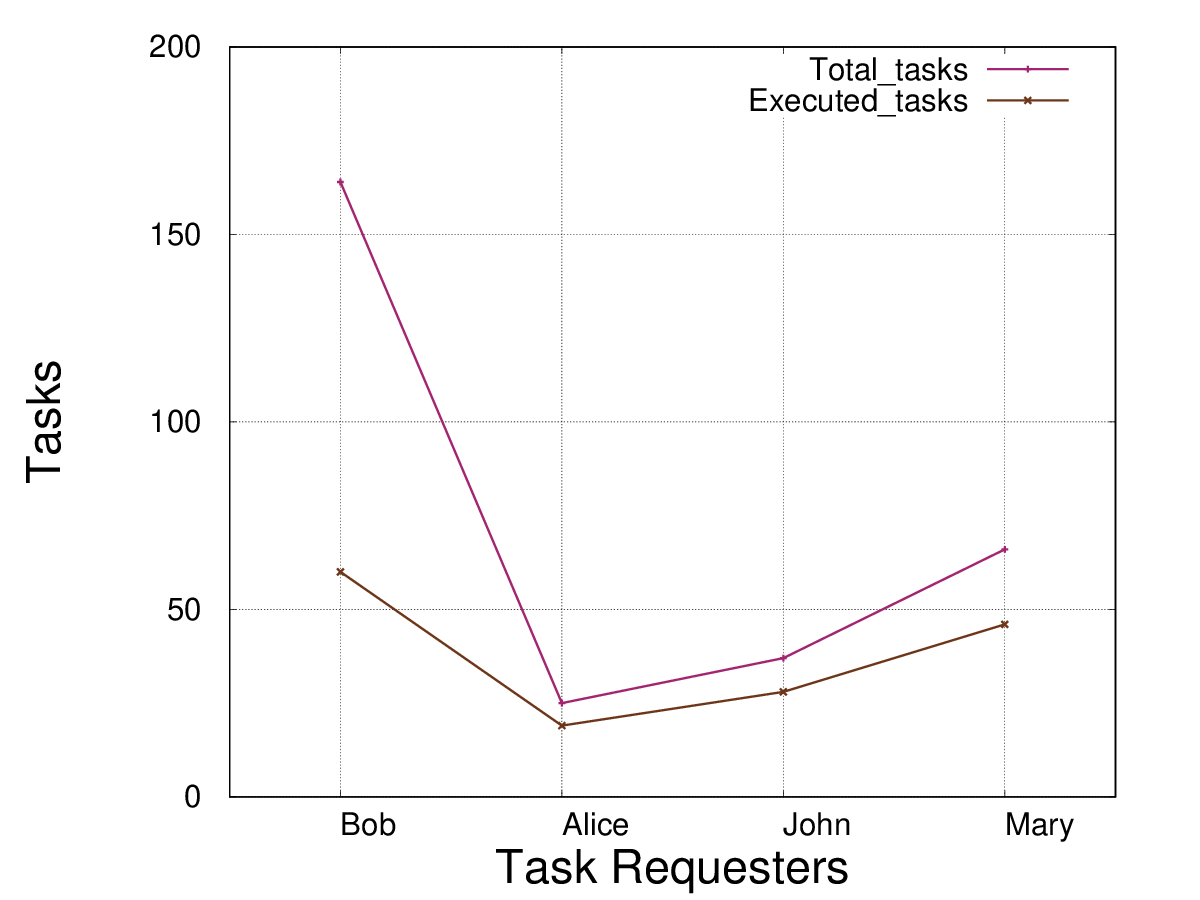}
\subcaption{Tasks Executed by TEs with a small no.of tasks}
\label{sim:44a}
\end{subfigure}%
\caption{Comparison of Tasks Executed by TEs, with large and small no.of tasks}
\label{Sim:4}
\end{figure}
\begin{figure}[H]
\begin{subfigure}{0.35\textwidth}
\includegraphics[scale = 0.30]{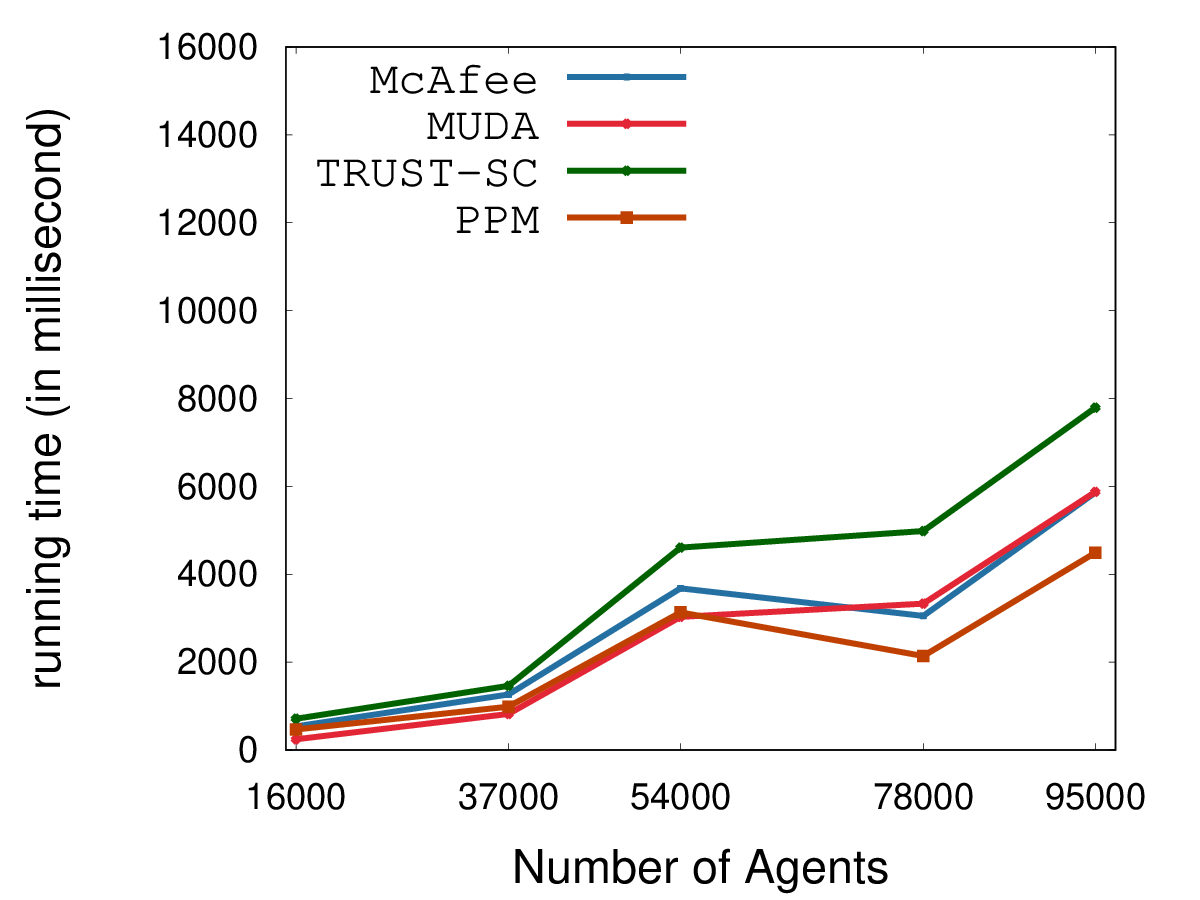}
\subcaption{Running time of agents with $k = 80$}
\label{sim:5a}
\end{subfigure}%
\begin{subfigure}{0.35\textwidth}
\includegraphics[scale = 0.30]{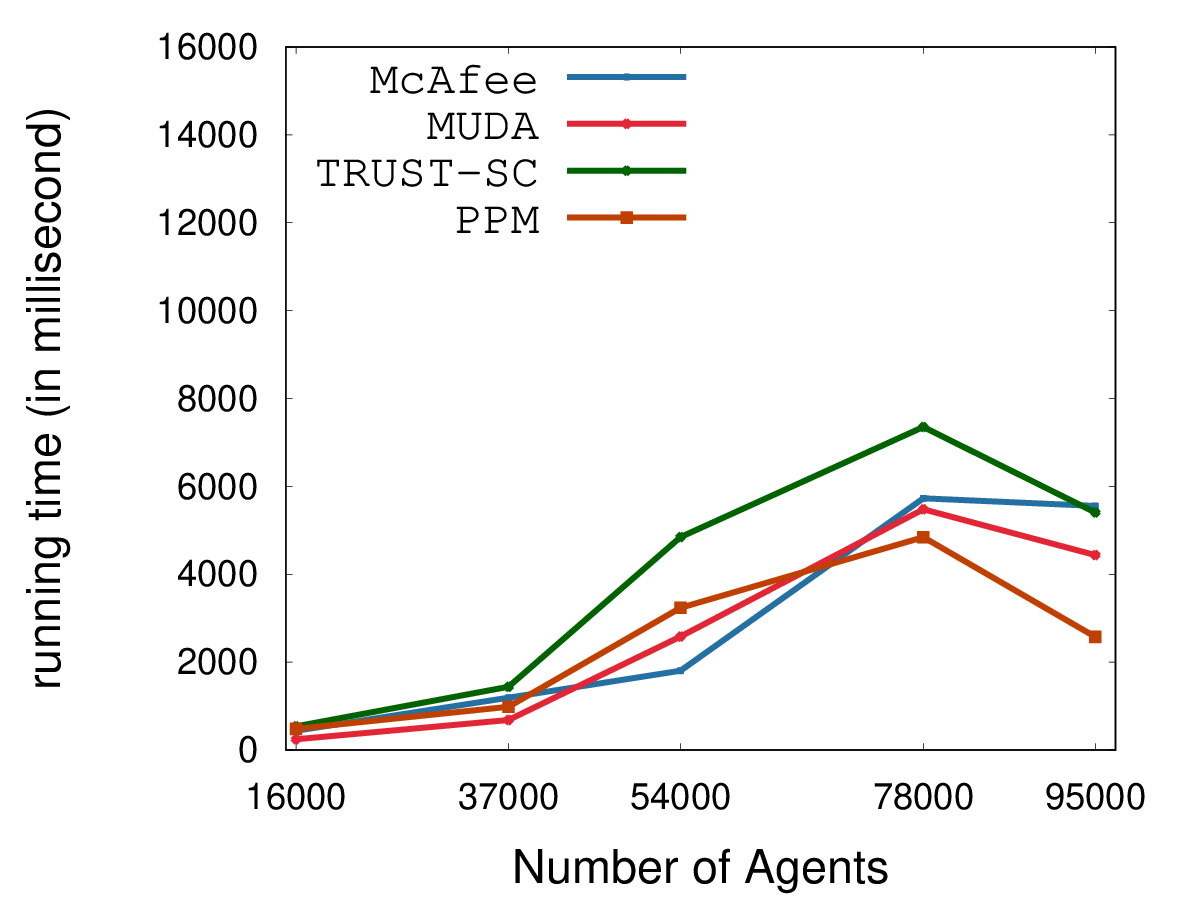}
\subcaption{Running time of agents with $k = 100$}
\label{sim:55a}
\end{subfigure}%
\begin{subfigure}{0.35\textwidth}
\includegraphics[scale = 0.30]{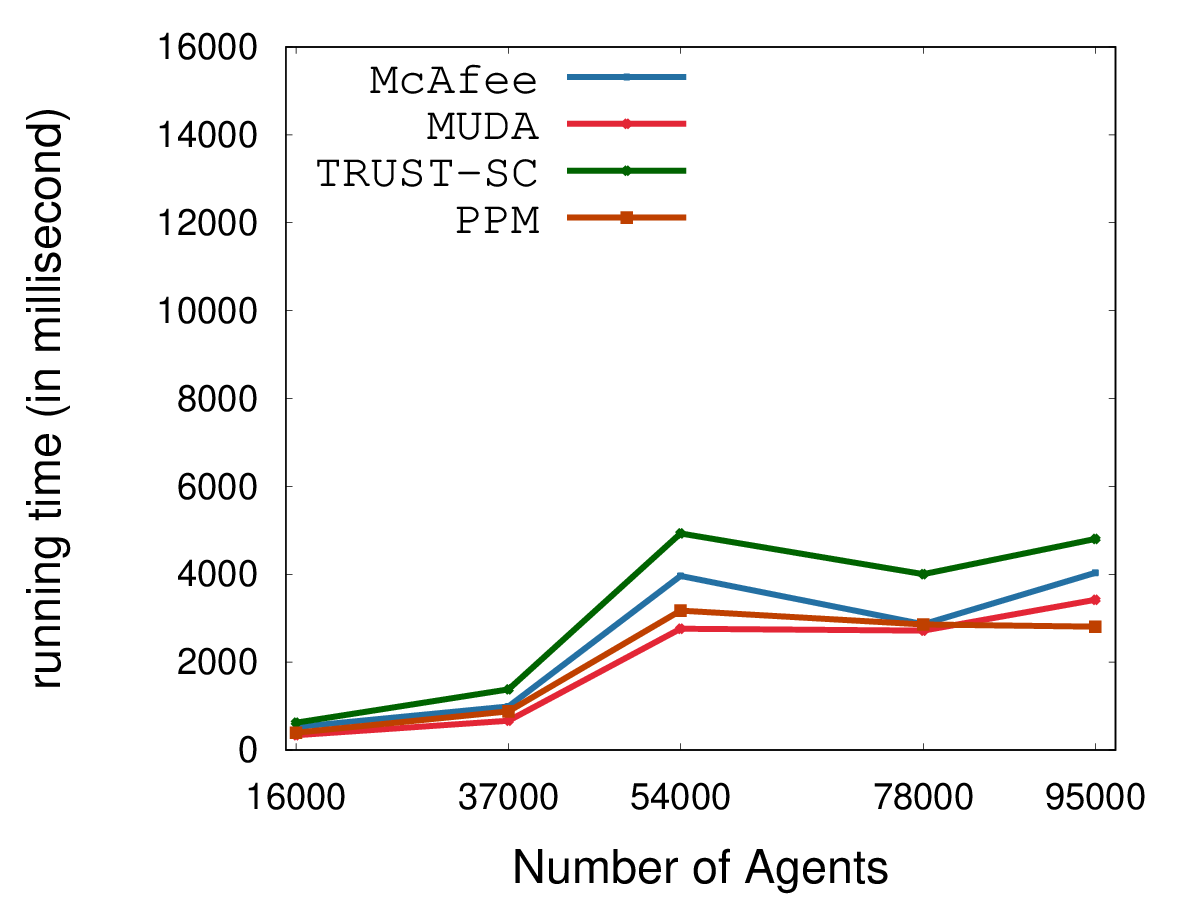}
\subcaption{Running time of agents with $k = 120$}
\label{sim:555a}
\end{subfigure}
\caption{Comparison of running time of Agents with $k = 80, 100$, and $120$.}
\label{sim:5A}
\end{figure}
In Fig. \ref{sim:5A}, the execution time (or running time) of TRUST-SC is compared with the execution time of McAfee, MUDA, and PPM with the cluster values of $80$, $100$, and $120$. The x-axis of the graphs represents the number of agents (both TRs and TEs), and the y-axis represents the running time in milliseconds. In this case, when the cluster number is $80$, then it can be seen that in Figs. \ref{sim:5a} shows that the TRUST-SC takes more time than other benchmark mechanisms. In the case of TRUST-SC, there are three subroutines: 1) Splitting the market, 2) determination of demand and supply, and 3) determination of equilibrium price, that will be executed for allocating the tasks to the task executors and determining their payment in a truthful manner. Similar nature TRUST-SC can be seen for $k-100$ and $120$ in Fig. \ref{sim:55a} and \ref{sim:555a} respectively.

\section{Conclusion and Future Works}
\label{se:conc}
In this paper, one of the spatial crowdsourcing scenarios is studied as a three-tiered, and a novel \emph{Quality-aware Truthful Multi-Unit Double Auction (TRUST-SC)} mechanism is proposed for quality-aware spatial crowdsourcing in strategic environments. In the first tier of the proposed framework, agents are grouped into clusters to reduce complexity and enable efficient localized matching. This improves scalability and structured market organization. Once the cluster formation is done, in the second tier, reliable executors are selected within each cluster based on performance and capability. This ensures accurate and efficient task execution.  In the third tier, tasks are assigned to selected executors, and payment is made to the winning TEs. This guarantees fairness, truthfulness, and efficient resource utilization. Through theoretical analysis, it is shown that the proposed mechanism TRUST-SC is \emph{truthfulness}, \emph{individual rationality}, and \emph{computationally efficient}. In the simulation results, we have seen that the utility of task requesters and executors is higher in the case of  McAfee and PPM than in the case of MUDA and TRUST-SC. The reason is that  McAfee and PPM operate in a unified market without splitting, allowing broader competition in the market. This often increases the total utility of both requesters and executors. While MUDA and TRUST-SC follow market splitting and equilibrium- price based, they may impose allocation constraints with homogeneous and heterogeneous tasks, respectively, which can slightly restrict the surplus utility of requesters and executors.  Same with the case in terms of total payment paid to the TEs. In terms of running time, MUDA takes more execution time than  McAfee, PPM, and MUDA. The reason is that for each task we are splitting the market and calculating the equilibrium price. This will take more time compared to other mechanisms. \\
\indent In the future, one could explore, in addition to the above set-up, each task is endowed with specific start time and finish time deadlines. In such scenarios, designing a time-constrained truthful mechanism for the above-mentioned set-up under spatial crowdsourcing that will ensure the quality of TEs and complete the task within the stipulated deadlines would be a challenging task.

\bibliographystyle{alpha}
\bibliography{phd}

\newpage 
\appendix
\end{document}